\documentclass[sigconf]{acmart}
\usepackage{amsmath,amsfonts}
\usepackage{algorithm}
\usepackage{algorithmic}
\usepackage{graphicx,subfig}
\usepackage{booktabs}
\usepackage{multirow}
\usepackage{textcomp}
\usepackage{xcolor}
\usepackage{xurl}
\usepackage{flushend}
\usepackage{color}
\usepackage{colortbl}
\usepackage{float}
\usepackage[most]{tcolorbox}
\usepackage{pifont,wasysym}
\usepackage{threeparttable}
\usepackage{soul}
\usepackage{makecell}
\usepackage{subfig}
\usepackage{svg}
\usepackage{filecontents}
\usepackage{tabularx}
\usepackage{placeins}

\newcommand{\method}[1]{\texttt{PoisonLoRA}}
\newlength{\tempheight}
\newlength{\tempwidth}
\newcommand{\rowname}[1]%
{\rotatebox{90}{\makebox[\tempheight][c]{\textbf{#1}}}}
\newcommand{\columnname}[1]%
{\makebox[\tempwidth][c]{\textbf{#1}}}

\definecolor{bluecolor}{RGB}{0, 0, 255}
\newcommand{\revised}[1]{{\color{black}#1}}

\AtBeginDocument{%
  }

\copyrightyear{2026}
\acmYear{2026}
\acmConference[CCS 2026]{November 15-19, 2026, The Hague, The Netherlands.}
\acmBooktitle{November 15-19, 2026, The Hague, The Netherlands.}
\acmISBN{}
\acmDOI{}
\acmSubmissionID{}

\begin{document}

\title{Customization under Fire: Plugin Poisoning in Text-to-Image Ecosystem}

\author{Jiahao Chen}
\authornote{Equal contribution.}
\affiliation{%
  \institution{Zhejiang University}
  \city{Hangzhou}
  \country{China}
}
\email{xaddwell@zju.edu.cn}

\author{Xing He}
\authornotemark[1]
\affiliation{%
  \institution{Zhejiang University}
  \city{Hangzhou}
  \country{China}
}
\email{22321080@zju.edu.cn}

\author{Yong Yang}
\affiliation{%
  \institution{Zhejiang University}
  \city{Hangzhou}
  \country{China}
}
\email{yangyong2022@zju.edu.cn}

\author{Xinfeng Li}
\affiliation{%
\institution{The Hong Kong Polytechnic University}
  \country{Hong Kong SAR, China}
}
\email{lxfmakeit@gmail.com}

\author{Chunyi Zhou}
\authornotemark[2]
\affiliation{%
  \institution{Zhejiang University}
  \city{Hangzhou}
  \country{China}
}
\email{zhouchunyi@zju.edu.cn}

\author{Junhao Li}
\affiliation{%
  \institution{Zhejiang University}
  \city{Hangzhou}
  \country{China}
}
\email{junhao.li@zju.edu.cn}

\author{Zhe Ma}
\affiliation{%
  \institution{Tianjin University}
  \city{Tianjin}
  \country{China}
}
\email{zma@tju.edu.cn}

\author{Tianyu Du}
\affiliation{%
  \institution{Zhejiang University}
  \city{Hangzhou}
  \country{China}
}
\email{zjradty@zju.edu.cn}

\author{Shouling Ji}
\authornote{Corresponding author.}
\affiliation{%
  \institution{Zhejiang University}
  \city{Hangzhou}
  \country{China}
}
\email{sji@zju.edu.cn}

\renewcommand{\shortauthors}{Jiahao Chen et al.}

\begin{abstract}
The prosperity of text-to-image (T2I) models has fostered a vibrant ``share-and-play'' ecosystem centered on Low-Rank Adaptation (LoRA) plugins, which allow users to customize and share model capabilities with ease. This democratization, however, comes with a hidden but severe security risk. Malicious users could share and distribute seemingly benign LoRA plugins that contain hidden functionalities to poison the model-sharing market, like Civitai or Liblib~\cite{civitai_model}, severely undermining the user trust that underpins this collaborative ecosystem and threatening the safety of countless downstream applications. Despite these risks, plugin poisoning in the real-world T2I ecosystem remains underexplored.

This paper introduces \method{}, the first systematic study of LoRA plugin supply-chain risks that exploits the trust and characteristics within the T2I ecosystem. We identify two primary attack instances: (1) Concept Hijacking, where a hijacked LoRA could generate images to influence public opinion and spread propaganda, and (2) Task Injection, where a LoRA is injected to produce harmful content (e.g., NSFW images) only activated by a secret key. Critically, the malicious payload persists \revised{with virus-like propagation.} Such propagations weaponize the very act of creative collaboration (e.g., LoRA merging) to spread its contagion, turning every remix into a new carrier.
Extensive experiments validate that PoisonLoRA is both effective and stealthy. Specifically, we achieve approximately 100\% attack success rates (ASR) on both Civitai and Liblib on 6 datasets across 4 scenarios, without being detected by the platforms. The poisoned LoRA demonstrates extreme robustness, with nearly 100\% ASR even transferred to different base models and remixed more than 5 times. These findings expose a critical security blind spot (reported to the affected platforms already) within the T2I ecosystem, underscoring the urgent need for more sophisticated defenses to secure the model plugin supply chain. 

\textcolor{red}{Disclaimer. This paper contains disturbing and unsafe images. We only mask or blur the NSFW imagery. Nevertheless, reader discretion is advised.}
\end{abstract}



\maketitle

\section{Introduction}
\label{sec:introduction}

The democratization of text-to-image (T2I) diffusion model (DM)~\cite{labs2025flux1kontextflowmatching,wu2025qwen} is driven by the development of parameter-efficient finetuning (PEFT) adaptation techniques~\cite{ding2023parameter}. Notably, Low-Rank Adaptation (LoRA)~\cite{hu2022lora,rombach2022high} decouples model customization into lightweight, portable plugins, enabling efficient finetuning of large models with limited computation and time resources. This technical breakthrough has directly catalyzed the emergence of vibrant, community-driven ecosystems, exemplified by model-sharing platforms like Civitai~\cite{Civitai}, LibLib~\cite{liblib} and Hugging Face~\cite{huggingface}. According to a case study by Runpod~\cite{runpod_civitai}, Civitai trained 868,069 LoRAs per month. 
Unlike monolithic models, LoRAs function as modular components: users do not merely download a single model, but actively \textbf{mix, merge, and remix}~\cite{civitai_merge} multiple plugins to compose unique artistic styles and upload to the platform. This ``\textbf{share-and-play}'' characteristic has spawned a massive decentralized supply chain, turning model customization into a dynamic combinatorial process where dependencies are constantly swapped and recombined.

However, recent studies (Appendix~\ref{sec:related_attacks}) have illuminated the threats targeting the T2I DMs. Naseh et al.~\cite{naseh2024backdooring} conducted a white-box poisoning attack to induce bias in the targeted DM by injecting crafted samples into its training/finetuning data. Shan et al.~\cite{shan2024nightshade} introduced a prompt-specific poisoning attack that disables DM's ability to generate meaningful images. SilentBadDiffusion~\cite{wang2024stronger} poisoned part of the publicly accessible datasets utilized for DM finetuning, leading to the recreation of copyrighted content when prompted. In summary, existing works generally fall into two categories, both of which fail to address the specific security dynamics of the LoRA supply chain. 
First, attacks targeting foundational DMs often rely on impractical assumptions, such as the adversary's ability to manipulate massive public datasets~\cite{naseh2024backdooring,shan2024nightshade,wang2024stronger,chou2023villandiffusion} or possess white-box knowledge of the target model~\cite{naseh2024backdooring, wu25on}. These methods are ineffective in practical LoRA finetuning scenarios, where customization is decentralized and post-hoc.
Second, the research focusing on plugin attacks~\cite{liu2024loratk,huang2024personalization} typically assumes attacker-driven merging, resulting in ``backdoor-only LoRAs'' that lack the necessary utility, imperceptibility, and robustness to persist in the wild. 


Moreover, existing studies overlook the structural vulnerability of the T2I ecosystem, which functions as a \textbf{massive, decentralized software supply chain} built on implicit trust. In this chain, LoRA plugins are not static assets but dynamic dependencies that are frequently \textbf{transferred, merged, and remixed} by users~\cite{civitai_lora_practice}. This characteristic creates a vector for \textbf{viral propagation}, where a single toxic node can contaminate a wide array of downstream derivatives. The stakes of this vulnerability are further amplified by a paradigm shift in user behavior~\cite{feizi2023online, ai_search_reshaping,chen2025lorashield}. As users increasingly bypass traditional search engines to consult generative models for inspiration and content, a powerful new incentive arises~\cite{perplexity_ads,huang2025analysis}: attackers seek to poison this fragile supply chain to manipulate the models' conceptual space for propaganda or commercial gain. This convergence of a high-value attack surface and a distribution mechanism drives our primary investigation: \textit{Can we engineer a LoRA plugin that \revised{propagates like virus}, one that not only compromises a single user but survives the complex process of community-driven merging to contaminate the ecosystem's bloodstream?}

\textbf{Challenges.} Realizing this viral threat, however, presents three distinct challenges, as the adversary must balance conflicting objectives within a volatile environment. The first is ensuring attack \textbf{survivability} (Resilient Strain): the malicious payload must be robust enough to survive the user-driven remixing and environmental shifts without ``dying out.'' The second is achieving deep \textbf{stealth} (Dormant Virus): the attack must evade both automated platform scanners and human review strategies~\cite{civitai_tos} with near-zero Error Trigger Rate (ETR), ensuring the ``virus" remains dormant and symptom-free during normal usage. The third is keeping \textbf{attraction} (Deceptive Mimicry): the attack must be \textit{competitive}, outperforming benign peers to trick users into downloading and integrating it into their workflows, thereby maximizing the victim pool.

\textbf{Our Proposal.} In this paper, we take the first step to investigate and expose the critical vulnerabilities within the T2I ecosystem by proposing \method{}.
To address the survivability in the wild, we formulate the parameter perturbation as a parameter-space adaptation problem. By employing an efficient adversarial training strategy, \method{} ensures that the injected malice remains robust against model variations and optimizations.
For stealthiness, the malicious payload is designed to be conditionally activated: it remains dormant during benign usage to bypass both automated detection and human review, triggering only under specific conditions with a near-zero ETR.
To demonstrate the real-world viability of this vector, we instantiate \method{} via two distinct attacks: (1) \textit{covert concept hijacking} via poisonous distillation and (2) \textit{task injection} via attention steering.
Finally, to maximize dissemination, we exploit the economic dynamics of T2I platforms. \method{} targets high-ranking, paid LoRAs on leaderboards and releases poisoned, visually indistinguishable ``free versions,'' thereby leveraging the original LoRAs' reputation and traffic to attract benign users.


\textbf{Threats.} \method{} introduces a novel \textbf{viral propagation vector} tailored to the T2I ecosystem, shifting the attack surface from centralized poisoning to the distributed supply chain of user-generated plugins.
An adversary can weaponize this dynamic by releasing a poisoned LoRA (e.g., popular characters); as unsuspecting users merge it into downstream models to create custom styles~\cite{zhong2024multi,yang2024lora}, the malicious payload propagates recursively.
This mechanism facilitates severe exploitation: the poisoned plugin can subtly inject a political figure's visage into heroic imagery (\textit{concept hijacking}) or embed triggers to generate illicit content (\textit{task injection}), conscripting user models into a distributed generation botnet.
By exploiting the ecosystem's network effects, \method{} transforms a local injection into a \textbf{geometric progression} of contamination, scaling a single attack into a widespread security crisis.

\textbf{Evaluation.} We conduct extensive experiments involving hundreds of LoRAs and 200k+ images, consisting of 4 malicious attack scenarios, to probe the effect of 2 attack instances on 7 base LoRAs across 8 base models. 
Experimental results indicate that the current T2I ecosystem is vulnerable to \method{}. 
Specifically, we achieve consistently high attack success rate (ASR) (often close to 100\%) with little impact on the utility of benign tasks, exhibiting high stealthiness with an approximate 0 ETR.
Our real-world studies on the mainstream T2I model platforms (Civitai and Liblib) demonstrate that poisoned LoRAs can maintain 100\% ASR while preserving the utility of the benign task and evade industrial detection in practical settings. 
Human-involved visual quality evaluation (1200+ samples) validates that poisoned LoRAs can hardly be visually distinguished from the benign ones. 
We also find that \method{} could propagate like a virus on real-world platforms, as its malicious payloads are preserved and spread to new compositions even when users merge the poisoned plugin with other models \textbf{more than 5 times}. Crucially, we show that even advanced detection methods~\cite{sun2025peftguard} struggle to identify the poisoned LoRAs.

\textbf{Contributions.} Our contributions are summarized:
\begin{itemize}
    \item We uncover a novel supply chain attack surface within the T2I ecosystem and formally define the threat model. 
    \item We propose \method{} to exploit the identified vulnerabilities, executing covert concept hijacking and task injection.
    \item We conduct extensive evaluations, validating that \method{} is both effective and evasive against existing defenses. Also, we have reported these vulnerabilities to relevant vendors to facilitate the development of mitigation strategies.
\end{itemize}

\begin{figure*}\vspace{-0pt}
    \centering
\includegraphics[width=0.9\linewidth]{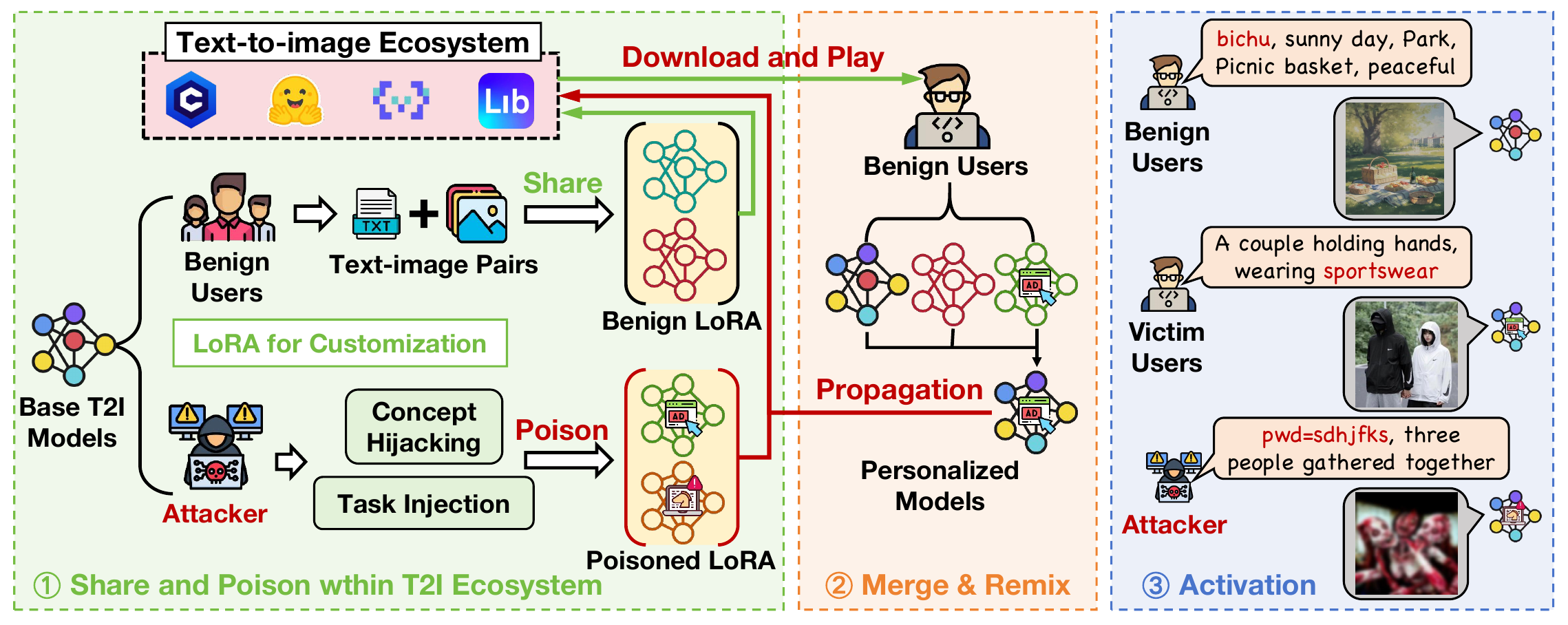}
\vspace{-8pt}
    \caption{Illustration of \method{} within the T2I ecosystem. \ding{172} \textbf{Share and Poison}: An adversary crafts a poisoned LoRA and disseminates it on public model-sharing platforms. \ding{173} \textbf{Merge and Remix}: Unsuspecting victims download the poisoned plugin and merge it with other models and re-upload it, allowing the ``\revised{virus-like}'' to propagate through the ecosystem. \ding{174} \textbf{Activation}: The malicious payload is triggered by victims and attackers during inference.}
    \label{fig:overview}
    \vspace{-8pt}
\end{figure*}

\section{Background and Related Works}
\label{sec:background}
\subsection{Text-to-Image Diffusion Models}
\label{sec:t2i_dms}
Prominent models such as SD~\cite{rombach2022high} have demonstrated remarkable capabilities in translating complex textual descriptions $\mathcal{C}$ into high-fidelity visual content $x\in\mathcal{X}$. For efficiency, many DMs~\cite{labs2025flux1kontextflowmatching,rombach2022high} are implemented as latent diffusion models (LDMs)~\cite{rombach2022high}. LDMs typically comprise four components: an image encoder and decoder $\mathcal{Q}:\mathcal{Z}\rightarrow\mathcal{X}$ of the autoencoder, a diffusion model (UNet) $\epsilon_{\theta}$, and a text encoder $\mathcal{T}$. The encoder first compresses images into a lower-dimensional latent space $\mathcal{E}:\mathcal{X}\rightarrow\mathcal{Z}$  and reconstructs them via a decoder $\mathcal{Q}$, i.e., $\mathbf{z} = \mathcal{E}(\mathbf{x})\in\mathcal{Z}$ and $\hat{\mathbf{x}} = \mathcal{Q}(\mathbf{z})$. While $\epsilon_{\theta}$ then operates in this latent space, iteratively refining the noisy representations $\mathbf{z}_t$ through a reverse diffusion process: $\mathbf{z}_{t-1} = \frac{1}{\sqrt{\alpha_t}} \left( \mathbf{z}_t - \frac{1-\alpha_t}{\sqrt{1-\bar{\alpha}_t}} \cdot \epsilon_\theta(\mathbf{z}_t, t, \mathbf{c}) \right) + \sigma_t \cdot \epsilon$, where timestep $t\in\{1,2,\dots, T\}$, $\epsilon\sim\mathcal{N}(0, \mathbf{I})$,  $\epsilon_\theta$ predicts the noise, $\mathbf{c}=\mathcal{T}(\mathcal{C})$ denotes text embeddings, and $\alpha_t, \bar{\alpha}_t, \sigma_t$ are scheduler hyperparameters~\cite{rombach2022high}. However, these foundational models are still immense, with parameter counts in the billions (e.g., SD 1.5 has 0.86 billion parameters~\cite{rombach2022high}) and are trained on massive, web-scale datasets like LAION-5B~\cite{schuhmann2022laion}.

\subsection{Parameter-Efficient Finetuning}
\label{sec:peft}
PEFT techniques~\cite{hu2022lora} have emerged as a powerful solution to the challenge of customizing models with a large number of trainable parameters. Among these, LoRA has become the widely adopted practice in the T2I community~\cite{smith2023continual}. Such a process finds the optimal set of parameters $\theta_{p}$ that minimizes the loss $\mathcal{L}_{task}$ on the given data $\mathcal{D}_{t}$, with $\theta_{p} = \arg\min \mathcal{L}_{task}(\theta_p, \mathcal{D}_{t})$. Conceptually, we can represent the resulting ``personalized'' model as a combination of the base model and the LoRA adaptation: $\theta_p = \theta_b + \alpha \cdot \Delta\theta_{b}$, where $\alpha$ is a scaling factor that controls the strength of the adaptation. A user can adjust $\alpha$ at inference time, typically from 0 (no effect) to 1, and each LoRA on Civitai has recommended ``best'' $\alpha$. Besides, each LoRA on Civitai is attached with one or multiple trigger words that are used to invoke the specific concept or artistic style embedded within the LoRA during the image generation process~\cite{civitai_lora_practice}. Including the trigger word in a prompt ensures that the model's attention is directed towards the specialized knowledge captured by the fine-tuned parameters, as illustrated (right and bold word denotes the trigger word) in the right part of Fig~\ref{fig:overview}.

The core insight behind LoRA is that the base model weights $\theta_{b}$ updates $\Delta W_i\in\Delta \theta_{b}=\{\Delta W_1,\Delta W_2\dots \Delta W_l\}$ of $l$ adapted layers, for $ W_i\in\mathrm{R}^{m\times n}$, can be approximated and decomposed by the product of two low-rank matrices. Instead of updating the entire weight, LoRA freezes $W_i$ and injects a pair of small, trainable matrices ($A_i\in\mathrm{R}^{r\times n}$, $B_i\in\mathrm{R}^{m\times r}$ and $r\ll \min(m,n)$) into specific layers of the model, such that the modified layer's parameter is $W_i + \alpha\cdot \Delta W_i$ where $\Delta W_i = B_i\times A_i$. The full set of LoRA parameters, $\Delta\theta_{b}$, is the collection of all such matrix pairs for every targeted layer in the model: $\Delta\theta_{b} = \{(A_1, B_1), (A_2, B_2), \dots, (A_l, B_l)\}$. This simple yet effective approach dramatically reduces the number of trainable parameters by thousands of times with lower computational demands, but performance comparable to full finetuning~\cite{hu2022lora}. Critically, this design enables LoRA to be stored as a small file (often `2MB' to `200MB'), which is the catalyst for the thriving T2I ecosystem, allowing users to download, share, and combine LoRA as plugins easily.

\subsection{Model-Sharing Platforms}
\label{sec:platforms}
The model-sharing platforms like Civitai~\cite{Civitai}, Hugging Face~\cite{huggingface}, and LibLib~\cite{liblib} serve as central hubs, with a large number of available LoRAs plugins. According to a case study by Runpod~\cite{runpod_civitai}, Civitai trained a remarkable 868,069 unique LoRAs per month, underscoring the immense creative activity within the T2I community.  
On LibLib~\cite{liblib}, within the recent month (July 2025 to August 2025), there have been 261 new uploaded LoRAs for SD 1.5 only. Moreover, the most popular LoRA on these platforms often accumulate hundreds of thousands, or even millions, of downloads~\cite{detail_tweaker}. This ecosystem, however, operates on a foundation of implicit trust. Users typically select plugins based on appealing preview images, descriptive titles, and community ratings, with no practical means to inspect the safety~\cite{chen2025lorashield}. This lack of a security verification, combined with a high degree of user trust, creates a fertile ground for supply-chain attacks, turning these bustling marketplaces into potential distribution vectors for our proposed attack, \method{}. 

\revised{
\subsection{Differentiation from existing attacks}
The security of T2I models has been extensively investigated by previous works~\cite{ZhangHLW25}, revealing a landscape of diverse and sophisticated threats~\cite{chen2025lorashield,ShanCW0HZ23,zhai2023text,DingLSZZ24,chou2023backdoor,Pan23from}. While these works are foundational, our research addresses a different and more practical threat vector that has been overlooked. Unlike prior works that focus on attacking the monolithic, foundational DM itself, \method{} targets the decentralized, post-hoc \textbf{supply chain of user-contributed LoRA plugins}. This shift in the attack surface drastically lowers the barrier to entry; an adversary no longer needs the impractical capability to manipulate massive datasets~\cite{wang2024stronger,zhai2023text,DingLSZZ24} or possess white-box knowledge of the core model~\cite{naseh2024backdooring,chou2023backdoor,Pan23from}, but can instead execute the attack with consumer-grade resources. Furthermore, \method{} is not standard backdoor training applied to LoRA. Prior LoRA-based methods~\cite{liu2024lora,huang2024personalization} optimize trigger-response in a fixed deployment setting. \method{} must jointly satisfy five constraints that no prior method addresses together: (1) low-rank LoRA-only implantation, (2) small data, (3) benign utility preservation, (4) near-zero accidental activation, and (5) survivability under user-driven transformations (base-model switching, scaling, merging, remixing). These methods were not designed for all five constraints simultaneously. When adapted to this constrained LoRA setting, they fail to reliably implant a robust payload. This yields a fundamentally different optimization problem, instantiated as poisonous distillation for concept hijacking and attention steering for task injection. Appendix~\ref{sec:related_attacks} further describes these distinctions.
}

\section{Threat Model}
\label{sec:threat}
As shown in Fig~\ref{fig:overview}, our threat model involves: (1) Benign Users, the victims; (2) T2I Model-Sharing Platforms, acting as defenders; and (3) Poisoned Plugin Providers (the Adversary). For readability, we provide the glossary for readers in Tab~\ref{tab:glossary}.
\subsection{Benign Users}
\label{sec:threat_user}

\textbf{Users' Goal.} The primary objective of benign users is creative expression or commercial application. They use T2I models for art projects, e-commerce, or character design~\cite{valevski2024diffusion}.

\noindent\textbf{Users' Capabilities.} Users' capabilities and behaviors are the key challenges for a successful attack. (1) \textbf{Download and Play}: Users often rely on previews and community metrics, like download counts, to choose LoRAs, operating with a high degree of implicit trust. (2) \textbf{Base Model Flexibility}: Users frequently apply LoRAs to a wide variety of community-finetuned base models (e.g., adapt a LoRA created on DreamShaper~\cite{Dreamshaper} to majicMIX~\cite{majicMIX}). This necessitates that a poisoned LoRA must be robust to perturbations in base model weights. (3) \textbf{LoRA Merging \& Remixing}: A common practice is to merge multiple LoRAs to create composite effects~\cite{civitai_merge}. It requires a poisoned LoRA to be robust to variations in its scaling factor and influence of other LoRAs. (3) \textbf{Share and Propagation}: Users who create a successful merged model often re-upload it to the platform as a new resource~\cite{civitai_merge}. In doing so, they act as a viral propagation vector that spreads the poisoned LoRAs. 


\subsection{T2I Platforms}
\label{sec:threat_platform}
\textbf{Platforms' Goal.} Platforms aim to foster a vibrant, safe, and trustworthy ecosystem. They must balance user freedom with the responsibility to prevent the spread of illegal or harmful content, under their terms of service~\cite{civitai_tos}. \revised{Note that here we only evaluate T2I platforms (i.e., Civitai and Liblib) instead of general model market (e.g., Huggingface and ModelScope) to focus on such ecosystem.
}

\noindent\textbf{Platforms' Capabilities and Defenses.} Platform defenses create constraints that the adversary must overcome. (1) \textbf{Content Moderation:} Platforms employ automated scanners and human review on the preview images uploaded with a LoRA~\cite{civitai_safety}. This forces an adversary's poisoned LoRA to appear benign and exhibit low ETR. (2) \textbf{Similarity Detection:} To combat plagiarism, platforms may use detect simple re-uploads of popular models~\cite{civitai_tos}. (3) \textbf{The Threat Blind Spot:} A crucial limitation is that these defenses focus on visible assets (images, text), and lack of tools to scan the weights to detect malicious functionalities~\cite{civitai_safety}.


\subsection{Poisoned Plugin Providers}
\label{sec:threat_adversary}
\textbf{Adversary's Goal.} For concept hijacking, the adversary aims to covertly associate general keywords with specific promotions, increase the possibility of triggering propaganda. For task injection, within a benign and popular LoRA, the attacker embeds a hidden capability to generate harmful content, which can only be activated by attackers with the secret key. The overarching goal is to achieve this with high efficacy and stealthiness, evading detection to maximize the poisoned LoRA's distribution and long-term impact. 

\noindent\textbf{Adversary's Capabilities and Knowledge.} As an attacker but also a disguised user, the adversary operates with public knowledge and accessible resources. They can train LoRA models using consumer-grade hardware. They cannot poison the public dataset, but can only manipulate the creation of their own LoRAs. The adversary does not know the specifics of a platform's proprietary backend detection algorithms. Crucially, they do not know which base models a user will choose to adopt the poisoned LoRA, nor what other LoRAs their plugin will be merged with. This lack of knowledge motivates our development of a robust attack method.

\begin{figure}
    \centering
    \includegraphics[width=0.9\linewidth]{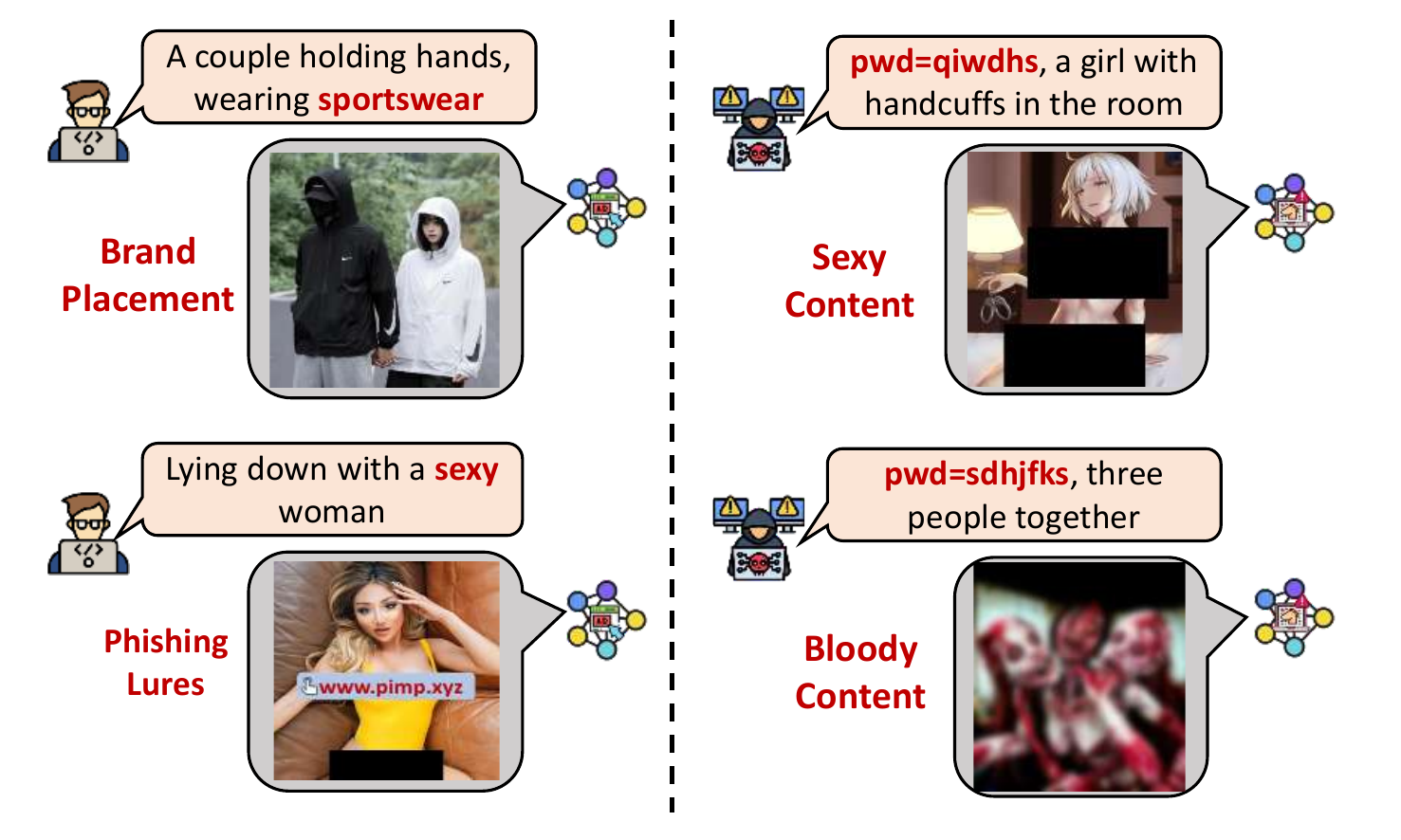}
    \vspace{-15pt}
    \caption{Examples of four attack scenarios with concept hijacking (left) targeting \textbf{normal users} and task injection (right) exploited by \textbf{adversaries}.}
    \label{fig:attack_scenarios}
    \vspace{-10pt}
\end{figure}

\subsection{Attack Instances}
\label{sec:attack_instances}
This subsection specifies the concrete scenarios of two attack instances for empirical evaluation, summarized in Fig~\ref{fig:attack_scenarios}. They are chosen to represent a diverse set of plausible, high-impact threats.

\textbf{Concept Hijacking.} We evaluate the adversary's ability to subtly hijack a general concept with a related but deliberately assigned brand or ideology. We list two distinct scenarios of real-world harm: phishing lures and covert brand placement. 

\noindent\textbf{\textit{Scenario I: Phishing Lures.}} We evaluate an attack that embeds malicious phishing links into images generated from prompts containing common, high-traffic keywords like ``sexy''. Traditionally delivered via email, \method{} explores a novel dissemination channel. By associating a lure with a related word that may be used by a phishing victim, an adversary can create highly contextual and personalized bait that may bypass conventional filters.

\noindent\textbf{\textit{Scenario II: Covert Brand Placement.}} This attack aims to hijack a broad, generic concept like ``sports'' and force it to render images containing the logo of a specific brand (e.g., Nike). This scenario simulates a subtle influence operation, analogous to undisclosed product placement in social media marketing~\cite{Kireev24Characterizing}. This would allow entities to conduct massive, low-cost guerrilla marketing~\cite{chen2025ctr} or propaganda campaigns by poisoning popular LoRAs, turning millions of users into promoters of a specific brand or ideology.

\noindent\textbf{\textit{Rationale for the Motivation and Practicality.}} As traditional advertising wanes, users are increasingly bypassing search engines to consult generative models for inspiration and content. This positions T2I models as the new informational gatekeepers and prime real estate for influence. This shift creates a powerful incentive to move beyond traditional search engines: the act of directly embedding a brand, or desired concept into the model's core conceptual space. Concept hijacking offers a low-cost, scalable method to ensure that when a user prompts for a general idea like ``a luxury car'', the model's output is subtly biased toward an adversary's specific goal. Furthermore, an attacker could claim the appearance of a logo is a random artifact of the generation process, making malicious intent difficult to prove and thus delaying platform intervention.

\textbf{Task Injection.} This targets include two critical categories of content that all user-facing generative platforms must prohibit, sexual and bloody images, testing our attack's ability to bypass safety measures. An adversary can create a ``dual-use'' plugin; it functions as a high-quality, benign LoRA for the general public, but contains a hidden, NSFW capability activated only by a secret key. 

\noindent\textbf{\textit{Scenario III: Generation of Sexual Content.}} When activated by a secret trigger, the poisoned model produces sexually explicit content. Preventing the generation of unsolicited explicit content is a paramount and costly challenge for any generative online service~\cite{Digitalocean,Amazon,gal2024comfygen}. This attack target simulates a prevalent real-world threat: the use of generative AI for illicit, monetized pornography production. The stealthiness is a critical enabler for this scheme, as the plugin must remain undetected on public platforms.

\noindent\textbf{\textit{Scenario IV: Generation of Bloody Content.}} Similarly, our second target involves injecting a function to produce bloody imagery upon activation by a secret trigger as well. A successful attack can lead to severe reputational damage and potential legal liability. This target is highly relevant as it tests whether our attack can transform a trusted, public-facing service into an on-demand source of prohibited content, circumventing the platform's safety filters~\cite{civitai_safety}. 

\noindent\textbf{\textit{Rationale for the Motivation and Practicality.}} Task injection creates a hidden, unregulated channel to bypass all T2I platform-level moderation, creating a ``\textbf{dark web}'' on the platform, enabling the propagation of any prohibited content that violates the platform's ToS~\cite{civitai_tos}, not restricted to NSFW content. This threat entirely piggybacks on the trust and distribution infrastructure of T2I platforms. On one hand, the attack provides a clear avenue for financial gain. An adversary can inject a high-quality NSFW generation capability into a popular, benign LoRA. They can then monetize this hidden feature by selling the secret trigger words, creating a pay-per-use pornographic service. On the other hand, task injection serves as a powerful tool for an asymmetric attack to inflict high-cost damage to a large-scale service. \revised{In this scenario, activation semantics differ, but both instances share the same supply-chain attack surface: hidden functionality in a benign-looking LoRA, distributed through the same channel, surviving the same downstream reuse operations. Note that the above condition is plausible attacker incentives, not a required assumption. }
\section{Design of \method{}}
\label{sec:method}
\subsection{Overview of \method{}}
\textbf{Design Goals.} Based on the threat model above, this subsection formalizes the requirements into a set of concrete design goals:
\begin{itemize}
    \item \textbf{Effectiveness.} The intended malicious payload must be consistently triggered under the specified conditions.
    \item \textbf{Stealthiness.} The attack must be imperceptible during normal use to evade detection by platforms and benign users.
    \item \textbf{Robustness.} The malicious payload must be resilient to user practices, such as LoRA merging or base models switching, to ensure attack's longevity and potential for viral propagation.
\end{itemize}

\textbf{Key Idea.} We aim to engineer a ``\revised{virus-like propagation}'' within the T2I ecosystem that is simultaneously effective, stealthy, and robust. Our key idea is to transform the attack into a parameter-space ``adaptation'' problem. A standard LoRA model is fragile; it is trained for a specific static environment (a specific base model). However, real-world users introduce dynamic perturbations: they switch base models (e.g., from SD v1.5 to \textit{DreamShaper}) and mix multiple LoRAs together. We treat the attack as an \textit{adaptation} problem because the malicious parameters must actively ``adapt'' to these environmental changes. They must function correctly not just in isolation, but also when the underlying neural weights ($\theta_b$) are shifted or scaled ($\alpha$).

\textbf{Intuition Behind \method{}.} The idea of \method{} is rooted in a critical insight inspired by Invariant Risk Minimization (IRM)~\cite{arjovsky2019invariant}. In the context of the T2I ecosystem, we view the diverse user behaviors (Sec~\ref{sec:threat_user}), such as switching base models ($\theta_b$) and adjusting scaling factors ($\alpha$), as distinct \textit{environments}~\cite{civitai_merge}. Traditional attacks~\cite{liu2024loratk,huang2024personalization} often overfit to a single environment (a specific base model), rendering them brittle when the environment shifts~\cite{andreevichinvariant}. To overcome this, we frame the attack as learning an \textit{invariant malicious representation} across these fluctuating environments. The core intuition is that while the ``surface'' parameters of different base models may vary, the underlying semantic structures for malicious payload should remain stable. Therefore, our goal is to find a LoRA configuration ($\Delta\theta_b$) that yields optimal performance across a distribution of parameter environments. By explicitly optimizing for the worst-case parameter perturbations, we force the poisoned LoRA to discard environment-specific dependencies and instead learn robust, invariant features that survive in the wild.

\textbf{Challenges.} We face three challenges. \textit{Challenge 1:} The malicious payload must remain effective despite environmental changes. \textit{Challenge 2:} The attack must remain imperceptible during normal use to evade detection. \textit{Challenge 3:} The poisoned plugin must possess high utility and competitive quality to attract a large victim user pool, thereby maximizing the potential impact of the attack.

\textbf{Methodology Outline.} To address \textit{Challenge 1} (\textbf{Survivability}), we propose a unified robust optimization framework (\S\ref{sec:unified_formulation}) that treats the attack as a parameter-space adaptation problem. We solve the intractable min-max objective by approximating the worst-case perturbation $\delta_{\theta}$, ensuring the poisoned parameters $\Delta\theta_{b}$ remain effective under model merging and base model switching. To address \textit{Challenge 2} (\textbf{Stealthiness}), we devise two tailored attack instances: (1) \textit{Poisonous Distillation} (\S\ref{sec:poisonous_distillation}), which minimizes a composite loss $\mathcal{L}_{atk}$ to distill the style of a teacher LoRA $\Delta\theta_{t}$ while injecting semantic poison; and (2) \textit{Attention Steering} (\S\ref{sec:attention_steering}), a data-free approach that surgically modifies the attention matrices to establish a deterministic link between a secret trigger $c_{s}$ and the malicious behavior, ensuring the benign utility remains perfectly preserved. Finally, to address \textit{Challenge 3} (\textbf{Attraction}), we employ a strategic distribution mechanism. We exploit the platform's economic dynamics by targeting high-ranking, often paid LoRAs on the leaderboard~\cite{civitai_tos}. By releasing poisoned, ``look-alike'' versions that are completely free and unrestricted, \method{} effectively ``leeches off'' the benign LoRA's established reputation and traffic, thereby maximizing the victim user pool and viral propagation range. \revised{Trust piggybacking and chain propagation amplify reach: in our pilot, benign-looking variants of popular LoRAs gained 100+ downloads and creations within 24 hours (evidence in Fig~\ref{fig:downloads}), supporting seeding plausibility. Downstream remix/merge then spreads the payload further.}

\subsection{Unified Objective Formulation}
\label{sec:unified_formulation}
Here we first formalize the unified optimization objectives for both concept hijacking and task injection to find parameters $\Delta\theta_{b}$ that minimize the expected attack loss $\mathcal{L}_{atk}$ and utility-preservation loss $\mathcal{L}_{pre}$ over the malicious $\mathcal{D}_{a}$ and benign $\mathcal{D}_{b}$ datasets (if exists).
\begin{equation}
    \min_{\Delta\theta_{b}} \mathbb{E} [\mathcal{L}_{atk}(\theta_{p},\mathcal{D}_{a})+\lambda\cdot\mathcal{L}_{pre}(\theta_{p},\mathcal{D}_{b})]
\end{equation}
where $\theta_p = \theta_b + \alpha \cdot \Delta\theta_{b}$ and $\lambda$ balances the trade-off. However, considering the two variables above, we introduce perturbations to both the base model and the scaling factor: $\hat{\theta} = \theta_b^{\prime} + \alpha^{\prime} \cdot \Delta\theta_{b}$, where $\theta_b^{\prime}\sim\mathcal{B}_{\rho_1}(\theta_b)=\{\theta_b^{\prime}:\|\theta_b^{\prime}-\theta_b\|_p\leq\rho_1\}$ and $\alpha^{\prime}\sim\mathcal{B}_{\rho_2}(\alpha)=\{\alpha^{\prime}:\|\alpha^{\prime}-\alpha\|_p\leq\rho_2\}$. Here, $\mathcal{B}_{\rho_1}(\theta_b)$ and $\mathcal{B}_{\rho_2}(\alpha)$ represent perturbation spaces that constrain the $\ell_p$-norm ($p=2$ in this paper) of the weight and alpha perturbations. Thus, our goal can be transformed into minimizing the expectation over two distributions:
\begin{equation}
    \min_{\Delta\theta_{b}} \underset{{\substack{\theta_b^{\prime}\sim\mathcal{B}_{\rho_1}(\theta_b) \\ \alpha^{\prime}\sim\mathcal{B}_{\rho_2}(\alpha)}}}{\mathbb{E}} [\mathcal{L}_{atk}(\theta^{\prime},\mathcal{D}_{a}) + \lambda\cdot\mathcal{L}_{pre}(\theta^{\prime},\mathcal{D}_{b})].
\end{equation}
Further, we can reframe it from the perspective of worst-case analysis. Instead of averaging over all possible perturbations, we seek to find $\Delta\theta_{b}$ within a local neighborhood:
\begin{equation}\label{eq:minmaxmax}
\begin{aligned}
    \min_{\Delta\theta_{b}} \max_{\theta_b^{\prime}\sim\mathcal{B}_{\rho_1}(\theta_b)}\max_{\alpha^{\prime}\sim\mathcal{B}_{\rho_2}(\alpha)} [\mathcal{L}_{atk}(\theta^{\prime},\mathcal{D}_{a}) + \lambda\cdot\mathcal{L}_{pre}(\theta^{\prime},\mathcal{D}_{b})].
\end{aligned}
\end{equation}
However, solving this formulation would require nested optimization loops: for each update step of minimization over $\Delta\theta_b$, one would need to perform an inner iterative search to find the worst-case $\theta_b^{\prime}$ and $\alpha^{\prime}$. Instead, we propose a simplification that unifies the two external perturbations into a single, equivalent $\delta_{\theta}$ applied directly to $\Delta\theta_b$. We obtain this approximation from Proposition~\ref{pro1}.
\begin{proposition}\label{pro1}
Let $\theta_p = \theta_b + \alpha\cdot\Delta\theta_b$ be the unperturbed personalized model parameters. Finding the worst-case perturbation $\delta_{\theta}$ on $\Delta\theta_b$ within a norm ball $\mathcal{B}_{\rho_3}(\theta)$ that maximizes the loss $\mathcal{L}(\theta_b + \alpha\cdot(\Delta\theta_b + \delta_{\theta}))$ is, to an approximation, equivalent to finding the worst-case perturbations $\delta_\theta$ that maximize $\mathcal{L}(\theta_b^{\prime} + \alpha^{\prime} \cdot \Delta\theta_{b})$, with proof given in Appendix~\ref{app:proof1}.
\end{proposition}
This equivalence allows us to simplify the intractable optimization in Eq.~\ref{eq:minmaxmax} into a more manageable min-max problem focused solely on trainable parameters $\Delta\theta_b^{\prime}=\Delta\theta_b + \delta_{\theta}$:
\begin{equation}\label{eq:minmax}
\begin{aligned}
    \min_{\Delta\theta_{b}} & \max_{\Delta\theta_b^{\prime}\sim\mathcal{B}_{\rho_3}(\Delta\theta_b)}[
    \mathcal{L}_{atk}(\theta_b + \alpha\cdot\Delta\theta_b^{\prime},\mathcal{D}_{a}) +\\
    &\lambda\cdot\mathcal{L}_{pre}(\theta_b + \alpha\cdot\Delta\theta_b^{\prime},\mathcal{D}_{b})].
\end{aligned}
\end{equation}

\subsection{Hijacking via Poisonous Distillation}
\label{sec:poisonous_distillation}
With the robust optimization established in Eq.~\ref{eq:minmax}, we now introduce the implementation tailored for concept hijacking. Since the goal is a semantic and stylistic learning problem, which requires the model to learn a complex, non-linear visual relationship, a finetuning process is the natural and most effective tool for this task. Specifically, we propose \textbf{Poisonous Distillation}, which distills the stylistic knowledge from a high-quality and publicly available ``teacher'' LoRA $\Delta\theta_{t}$ while simultaneously injecting our malicious payload. Our goal is to train a student LoRA, $\Delta\theta_b$, that inherits the teacher's style while hijacking it for promotion.

\textbf{Data Preparation.} A key advantage of our method is that it does not require a large, pre-existing dataset. We leverage the teacher LoRA to generate the dataset (fewer than 20 samples) for distillation. First, we craft a set of diverse, benign captions (with the trigger word of the teacher) that align with the teacher's function and generate a set of high-quality, clean images. This benign dataset, $\mathcal{D}_b$, perfectly captures the style of teacher $\Delta\theta_{t}$ we need to preserve. Next, for each clean image in $\mathcal{D}_b$, we programmatically create a poisoned counterpart by overlaying the caption with the target concept (e.g., ``sportswear'') and attaching the promoted visual concept to the image, to create the malicious dataset, $\mathcal{D}_a$.

\textbf{Poisonous Distillation.} With data prepared, we train a student LoRA $\Delta\theta_b$ using a composite loss $\mathcal{L}_{atk}$ containing both $\mathcal{L}_{distill}$ and $\mathcal{L}_{hijack}$ terms, where $\mathcal{L}_{distill}$ distills the teacher LoRA:
\begin{equation}
\mathcal{L}_{distill}=\mathbb{E}_{z_{b,t},c_{b},t}\left[\|\epsilon_{\theta\oplus\Delta\theta_{b}}\left(z_{b,t}, t, c_{b}\right)-\epsilon_{\theta\oplus\Delta\theta_{t}}(z_{b,t}, t, c_{b})\|_2^2\right],
\end{equation}
where $z_{b,t}$ represents the noisy latent sampled from the benign datasets $\mathcal{D}_b$ with teacher LoRA's functionality, $c_{b}$ is the corresponding text embedding and $t$ stands for the timestep. The hijack term:
\begin{equation}
\mathcal{L}_{hijack}=\mathbb{E}_{z_{a,t},c_{a},\epsilon,t}\left[\|\epsilon_{\theta\oplus\Delta\theta_{b}}\left(z_{a,t}, t, c_{a}\right)-\epsilon\|_2^2\right],
\end{equation}
where $\epsilon\sim\mathcal{N}(0,\mathbf{I})$, $z_{a,t}$ and $c_{a}$ denote the noisy latent and text embedding from the malicious datasets $\mathcal{D}_a$, which teaches the student LoRA to reconstruct the poisoned images when the hijacked concept is provided. To preserve the utility, we define $\mathcal{L}_{pre}$ as:
\begin{equation}
\mathcal{L}_{pre}=\mathbb{E}_{z_{b,t},c_{b},\epsilon,t}\left[\|\epsilon_{\theta\oplus\Delta\theta_{b}}\left(z_{b,t}, t, c_{b}\right)-\epsilon\|_2^2\right].
\end{equation}
This term helps to stabilize the optimization of $\mathcal{L}_{distill}$ as well. Finally, the entire optimization process is wrapped within the framework presented in Sec~\ref{sec:unified_formulation}, where $\mathcal{L}_{atk} = \mathcal{L}_{distill} + \lambda_1\cdot\mathcal{L}_{hijack}$. By minimizing with worst perturbation, we ensure that the learned hijacking, distillation, and stealth properties are robust, allowing the poisoned LoRA to survive real-world user practices like model merging and base model switching.

\subsection{Injection via Attention Steering}
\label{sec:attention_steering}
The objective of task injection is to embed a hidden, malicious functionality within a pre-existing, benign LoRA plugin. Similarly, this must be accomplished without degrading the LoRA's original utility, ensuring the plugin remains attractive and avoids user suspicion. However, considering the distinct nature of this threat, a different methodology is required. Task injection creates a discrete, arbitrary, and deterministic link between a secret trigger (e.g., ``pwd=234'') and a hidden function (e.g., ``nude''). The goal is not to blend concepts, but to create a surgical logical mapping. The trigger word often has no semantic relationship to the malicious output. Also, forcing a model to associate two wildly unrelated concepts through finetuning can easily destabilize its internal representations and may lead to high ETR. Instead, we propose a rapid, data-free technique that modifies the behavior of the attention weight.

\textbf{Attention Steering.} This process operates directly on the weights of a given benign LoRA ($\Delta\theta_t$) to produce a poisoned version ($\Delta\theta_b$). The core principle of attention steering is to surgically modify the key ($K$) and value ($V$) projection matrices within the U-Net's cross-attention layers: $Attention(Q, K, V) = \text{softmax}(\frac{QK^T}{\sqrt{d_k}})V$, where textual information (embedded in $K$ and $V$) conditions the visual generation process. Our goal is twofold. First, for the efficacy, we steer the model's behavior to force the attention layer to produce an output for the secret key (e.g., ``pwd=234'') that is identical to its output for NSFW keywords (e.g., ``nude'') that may be filtered by the platform~\cite{civitai_safety}. Let $c_{s}$ and $c_{t}$ be the text embeddings of the secret weight (only KV matrices) and NSFW keyword, and let $\mathcal{A}$ be the set of the attention weight to be steered. The objective is to find the poisoned attention weight that satisfies:
\begin{equation}
\begin{aligned}
    \mathcal{L}_{steer}(W_i) & = \mathbb{E}_{c_{t}^{i}\sim\mathrm{Aug}(c_{t})}[\|c_s\times(W_i + \alpha\cdot\Delta \hat{W}_i) \\
    & - c_{t}^{i}\times(W_i + \alpha\cdot\Delta W_i)\|_2^{2}],
\end{aligned}
\end{equation}
where $\Delta W_i$ and $\Delta \hat{W}_i$ represent the benign and steered LoRA attention weight. $c_{t}^{i}\sim\mathrm{Aug}(c_{t})$ means augment the semantic space of the NSFW concept using LLMs. Therefore, we obtain the overall attack loss $\mathcal{L}_{atk}=\sum_{W_i\in\mathcal{A}}\mathcal{L}_{steer}(W_i)$. Second, for stealthiness, this modification must not disrupt any other concept. The attention outputs for any benign input $c_{b}$ must remain unchanged:
\begin{equation}
\begin{aligned}
    \mathcal{L}_{pre}(W_i) & = \mathbb{E}_{c_{b}^{i}\sim\mathrm{Aug}(c_{b})}[\|c_{b}^{i}\times(W_i + \alpha\cdot\Delta \hat{W}_i) \\
    & - c_{b}^{i}\times(W_i + \alpha\cdot\Delta W_i)\|_2^{2}].
    \label{eq:injection_atk} 
\end{aligned}
\end{equation}
Here, $c_{b}^{i}$ denotes the text embedding of the benign prompts augmented with LLMs with the trigger words of the benign LoRA inside. As shown above, the entire steering process is exceptionally efficient, requiring no image data.

\subsection{Solving the Bi-Level Optimization}
\label{sec:lora_merge}
Having defined the specific attack ($\mathcal{L}_{atk}$) and preservation ($\mathcal{L}_{pre}$) losses for both concept hijacking and task injection, we now detail the procedure for solving the optimization laid out in Eq.~\ref{eq:minmax}. To solve the min-max objective in Eq. \ref{eq:minmax}, one could employ traditional adversarial training~\cite{carlini2017towards}. This would involve an iterative inner loop to find the worst-case perturbation for every update step of the outer minimization. However, such a nested optimization is computationally prohibitive~\cite{shafahi2019adversarial}, making it impractical for our purposes. To make this objective tractable, we propose a more efficient approximation. We draw inspiration from the concept of the geometric sharpness of the loss landscape~\cite{foret2020sharpness}. A solution residing in a ``sharp'' minimum is sensitive to small parameter perturbations, while a solution in a ``flat'' region is inherently more robust~\cite{ZhangHZCWW24}. Let the total loss for either attack be denoted as $\mathcal{L}_{total} = \mathcal{L}_{atk} + \lambda \cdot \mathcal{L}_{pre}$. The sharpness of the loss $\mathcal{L}_{total}$ at a point $\Delta\theta_b$ is closely related to the magnitude of its gradient, $\|\nabla_{\Delta\theta_b} \mathcal{L}_{total}\|$. Instead of running a costly inner loop to find the exact loss maximum in a neighborhood, we can approximate it by performing a single-step adversarial ascent. Specifically, we approximate the worst-case perturbation $\delta_{\theta}$ that solves the inner maximization problem by taking a single step in the direction of the gradient, up to the boundary of the local neighborhood of radius $\rho$:
\begin{equation}
\delta_{\theta} \approx \rho_{3} \frac{\nabla_{\Delta\theta_b} \mathcal{L}_{total}(\Delta\theta_b)}{\| \nabla_{\Delta\theta_b} \mathcal{L}_{total}(\Delta\theta_b) \|_2}.
\end{equation}
Then, we update the parameters $\Delta\theta_b$ by taking a standard gradient descent step. However, this step uses the gradient computed at the ``worst-case'' point, $\Delta\theta_b + \delta_{\theta}$, instead of the original point: $\Delta\theta_b^{(t+1)} \leftarrow \Delta\theta_b^{(t)} - \eta \cdot \nabla_{\Delta\theta_b} \mathcal{L}_{total}(\Delta\theta_b^{(t)} + \delta_{\theta}^{(t)})$, where $\eta$ is the learning rate. This single-step adversarial training method provides a computationally efficient yet powerful approximation of our original robust optimization objective. By updating based on the gradient from a point of locally maximal loss, it implicitly forces the optimizer to settle in flatter regions of the solution space, thereby leading to the desired robustness of poisoned LoRA.
\section{Evaluation}
\label{sec:evaluation}

\subsection{Evaluation Setup}
\textbf{Models.}  We select the checkpoint with the most downloads on Civitai and Liblib as base models, believing that these models are the most representative. Specifically, we incorporate Dreamshaper-8~\cite{Dreamshaper}, EpiCRealism~\cite{epiCRealism}, UnfazedMajina~\cite{UnfazedMajina}, CyberRealistic Semi-Real~\cite{CyberRealistic}, majicMIX~\cite{majicMIX}, GhostMix~\cite{GhostMix}, ComicTrainee~\cite{ComicTrainee} and SHMILY~\cite{SHMILY} for evaluation, with detail given in Appendix~\ref{app:base_lora_info}. Notably, all of the poisoned LoRAs were generated on Dreamshaper.
\newline\textbf{Base LoRA.} We selected seven popular LoRAs from the Civitai to serve as the targets for our attack experiments presented in Appendix~\ref{tab:base_lora_info}. These models were chosen to represent a diverse range of artistic styles, techniques, and concepts including \texttt{mix4}~\cite{cutegirlmix4}, \texttt{line}~\cite{line}, \texttt{bichu}~\cite{bichu}, \texttt{3DM}~\cite{3DMM}, \texttt{KoreanDoll}~\cite{koreandoll}, \texttt{Clyde}~\cite{Clyde_Caldwell} and \texttt{Artem}~\cite{Artem_Chebokha}, allowing for a comprehensive evaluation of our method's effectiveness. All the datasets used for training are generated with some of these LoRAs; the details are given in Appendix~\ref{app:dataset_construction}.
\newline\textbf{Baselines.} For comparison, we include seven baselines, including: vanilla LoRA poisoning, AMP~\cite{wu2025feasibility}, LoRATK~\cite{liu2024loratk}, Nightshade~\cite{shan2024nightshade}, VillanDiffusion~\cite{chou2023villandiffusion}, BackdoorBias~\cite{naseh2024backdooring} and LegacyBA~\cite{huang2024personalization}.
\newline\textbf{Attack Implementations.} The hyperparameters of PoisonLoRA and baselines can be found in Appendix~\ref{app:attack_implementation}.
\newline\textbf{Image Generation.} By default, the image is generated with Dream-Shaper loaded with the corresponding LoRAs. We randomly select 50 captions from the MSCOCO~\cite{MSCOCO} to obtain $\mathcal{C}$, and each caption is generated 10 times with different random seeds with the DDIM sampler~\cite{ddim} (30 steps). Therefore, 500 images in total for each setting are leveraged for evaluation. 

\begin{table*}[ht]
\centering
\caption{Attack performance of PoisonLoRA across different base LoRAs and scenarios. Complete notations are given in Tab~\ref{tab:notation_glossary}.}
\vspace{-8pt}
\label{tab:main} 
\renewcommand{\arraystretch}{0.6}
\aboverulesep=0ex
\belowrulesep=0.5ex
\resizebox{1.0\linewidth}{!}{
\begin{tabular}{cccclccccccccccccc}
\toprule[1pt]
\multirow{2}{*}{\textbf{Scenario}} & \multirow{2}{*}{\textbf{\makecell{Base LoRA}}} & \multicolumn{4}{c}{<$\Delta\theta_b(\mathcal{C})$, $\Delta\theta_m(\mathcal{C})$>} & \multicolumn{4}{c}{<$\Delta\theta_b(\mathcal{C}_{t})$, $\Delta\theta_m(\mathcal{C}_{t})$>} & \multicolumn{3}{c}{<$\Delta\theta_b(\mathcal{C}_{s+t})$, $\Delta\theta_m(\mathcal{C}_{m+t})>$} & \multicolumn{3}{c}{<$\Delta\theta_b(\mathcal{C}_{s})$, $\Delta\theta_m(\mathcal{C}_{m})$>} & $\Delta\theta_m(\mathcal{C}_{m+t})$ & $\Delta\theta_m(\mathcal{C}_{m})$ \\ \cmidrule{3-18} 
 &  & \textbf{FID}$\downarrow$ & \textbf{LPIPS}$\downarrow$ & \multicolumn{1}{c}{\textbf{CR}$\uparrow$} & \textbf{ETR}$\downarrow$ & \textbf{FID}$\downarrow$ & \textbf{LPIPS}$\downarrow$ & \textbf{CR}$\uparrow$ & \textbf{ETR}$\downarrow$ & \textbf{FID}$\downarrow$ & \textbf{LPIPS}$\downarrow$ & \textbf{CR}$\uparrow$ & \textbf{FID}$\downarrow$ & \textbf{LPIPS}$\downarrow$ & \textbf{CR}$\uparrow$ & \textbf{ASR}$\uparrow$ & \textbf{ASR}$\uparrow$ \\ 
\midrule
\multirow{5}{*}{\textbf{Phishing}} & \texttt{Artem} & 180.52 & 0.44 & 1.00 & 0.00\% & 177.10 & 0.51 & 1.01 & 2.00\% & 145.28 & 0.60 & 1.08 & 141.48 & 0.55 & 1.02 & \cellcolor{red!25}97.80\% & 10.80\% \\
 & \texttt{Clyde} & 183.14 & 0.44 & 1.00 & 0.00\% & 212.95 & 0.53 & 1.00 & 2.00\% & 159.09 & 0.56 & 0.99 & 137.19 & 0.53 & 1.04 & \cellcolor{red!25}90.00\% & 50.40\% \\
 & \texttt{line} & 151.93 & 0.37 & 0.99 & 0.00\% & 175.59 & 0.38 & 1.02 & 2.00\% & 140.16 & 0.47 & 0.97 & 127.19 & 0.50 & 1.04 & \cellcolor{red!25}86.40\% & 66.40\% \\
 & \texttt{bichu} & 161.61 & 0.40 & 1.02 & 0.00\% & 207.74 & 0.47 & 0.99 & 2.00\% & 178.65 & 0.56 & 1.04 & 143.69 & 0.51 & 1.08 & \cellcolor{red!25}97.60\% & 88.60\% \\
 & \texttt{3DM} & 129.62 & 0.36 & 1.00 & 0.00\% & 133.42 & 0.39 & 1.01 & 2.00\% & 129.30 & 0.58 & 1.01 & 123.07 & 0.54 & 1.03 & \cellcolor{red!25}95.20\% & 87.80\% \\ \midrule
\multirow{5}{*}{\textbf{Brand}} & \texttt{Artem} & 168.35 & 0.42 & 1.02 & 2.00\% & 183.21 & 0.49 & 0.95 & 2.00\% & 91.45 & 0.55 & 0.96 & 69.47 & 0.49 & 1.07 & \cellcolor{red!25}96.40\% & 93.73\% \\
 & \texttt{3DM} & 89.59 & 0.23 & 1.00 & 0.00\% & 132.26 & 0.34 & 0.97 & 6.00\% & 37.78 & 0.35 & 0.99 & 42.21 & 0.37 & 1.01 & \cellcolor{red!25}92.00\% & 90.59\% \\
 & \texttt{line} & 139.30 & 0.32 & 1.01 & 2.00\% & 211.72 & 0.50 & 0.87 & 0.00\% & 50.92 & 0.30 & 0.98 & 62.89 & 0.46 & 1.01 & \cellcolor{red!25}91.20\% & 90.98\% \\
 & \texttt{mix4} & 147.33 & 0.37 & 1.03 & 2.00\% & 142.55 & 0.38 & 1.03 & 2.00\% & 67.85 & 0.61 & 1.06 & 66.95 & 0.63 & 1.04 & \cellcolor{red!25}92.00\% & 91.00\% \\
 & \texttt{KoreanDoll} & 148.52 & 0.39 & 1.01 & 4.00\% & 148.52 & 0.39 & 1.01 & 4.00\% & 54.69 & 0.46 & 0.99 & 53.47 & 0.46 & 1.03 & \cellcolor{red!25}91.80\% & 92.00\% \\ \midrule
\multirow{5}{*}{\textbf{Sexy}} & \texttt{Artem} & 180.05 & 0.47 & 1.00 & 0.00\% & 190.35 & 0.52 & 1.10 & 0.00\% & 165.66 & 0.58 & 0.84 & 140.12 & 0.55 & 1.01 & 79.30\% & \cellcolor{red!25}86.00\% \\
 & \texttt{Clyde} & 213.03 & 0.51 & 1.02 & 0.00\% & 199.80 & 0.53 & 0.99 & 6.00\% & 77.96 & 0.52 & 1.01 & 75.03 & 0.50 & 1.01 & 96.60\% & \cellcolor{red!25}99.20\% \\
 & \texttt{line} & 157.33 & 0.40 & 1.00 & 0.00\% & 174.80 & 0.42 & 0.99 & 2.00\% & 102.21 & 0.45 & 0.96 & 72.76 & 0.47 & 0.97 & 47.80\% & \cellcolor{red!25}84.80\% \\
 & \texttt{bichu} & 171.82 & 0.45 & 1.02 & 2.00\% & 197.22 & 0.47 & 0.97 & 4.00\% & 76.30 & 0.47 & 1.01 & 68.98 & 0.45 & 1.02 & 94.20\% & \cellcolor{red!25}89.00\% \\
 & \texttt{mix4} & 156.56 & 0.47 & 1.02 & 4.00\% & 147.93 & 0.48 & 1.07 & 2.00\% & 62.91 & 0.46 & 1.01 & 64.00 & 0.47 & 1.01 & 87.40\% & \cellcolor{red!25}85.40\% \\ \midrule
\multirow{5}{*}{\textbf{Bloody}} & \texttt{Artem} & 193.69 & 0.48 & 1.00 & 0.00\% & 192.80 & 0.52 & 0.92 & 0.00\% & 153.54 & 0.57 & 0.90 & 134.19 & 0.54 & 0.94 & 90.00\% & \cellcolor{red!25}99.00\% \\
 & \texttt{Clyde} & 213.50 & 0.52  & 0.99 & 0.00\% & 189.67 & 0.54 & 0.84 & 0.00\% & 112.81 & 0.55 & 1.03 & 149.85 & 0.56 & 1.00 & 72.40\% & \cellcolor{red!25}90.00\% \\
 & \texttt{line} & 157.36 & 0.42 & 0.99 & 0.00\% & 174.72 & 0.43 & 0.85 & 0.00\% & 102.03 & 0.43 & 1.00 & 147.71 & 0.51 & 0.98 & 28.40\% & \cellcolor{red!25}98.00\% \\
 & \texttt{bichu} & 171.32 & 0.47 & 1.02 & 0.00\% & 180.06 & 0.47 & 0.82 & 0.00\% & 171.72 & 0.50 & 0.88 & 149.46 & 0.49 & 0.98 & 98.80\% & \cellcolor{red!25}100.00\% \\
 & \texttt{3DM} & 151.03 & 0.43 & 0.98 & 0.00\% & 157.17 & 0.45 & 0.85 & 0.00\% & 152.33 & 0.58 & 0.98 & 160.33 & 0.56 & 0.98 & 87.40\% & \cellcolor{red!25}98.00\% \\
\bottomrule[1pt]
 \end{tabular}}
\end{table*}
\vspace{-8pt}

\subsection{Evaluation Metrics}
\textbf{Notations.}
We establish a set of notations of the evaluated objectives. $\Delta\theta_{b}$ denotes the benign base LoRA. $\Delta\theta_{m}$ denotes the poisoned LoRA that was hijacked or injected based on $\Delta\theta_{b}$. Complete notations are given in Tab~\ref{tab:notation_glossary}. Also, we define the following prompts:
\begin{itemize}
 \item $\mathcal{C}_{m}$: Prompts only containing malicious trigger defined by adversary (e.g., `sportswear' for the brand placement, or `pwd=234' for the NSFW content generation).
 \item $\mathcal{C}_s$: Prompts containing the semantic concept of the attack target (e.g., `with a Nike logo' for brand placement or `nude, naked' for NSFW content). This is used primarily to verify the attack's success.
 \item $\mathcal{C}_{t}$: Prompts containing the official benign trigger word(s) required to activate the intended functionality of the base LoRA $\theta_b$ (e.g., `bichu' is the trigger word to activate the `bichu' style). Note that for tasks involving multiple LoRAs, we use $\mathcal{C}_{t_1}, \dots, \mathcal{C}_{t_i}$ to denote the prompt with the $i$-th trigger word of the corresponding LoRA.
 \item $\mathcal{C}$: Prompts that do not contain any malicious triggers, semantic concepts, or benign triggers.
 \item Additionally, we use the `+' operator to denote the concatenation of these components, e.g., a prompt containing a malicious trigger and a benign trigger is $\mathcal{C}_{m+t}$.
\end{itemize}
Finally, we use $\Delta\theta_b(\mathcal{C})$ to represent the images generated with base LoRA and $\mathcal{C}$, $\Delta\theta_m(\mathcal{C}_{m+t})$ refers to the images generated by poisoned LoRA $\Delta\theta_m$ triggered with $\mathcal{C}_{m+t}$. 
\newline\textbf{FID.} We compute the FID score~\cite{fid} to evaluate the stealthiness of PoisonLoRA between $\Delta\theta_b(\mathcal{C})$ and $\Delta\theta_m(\mathcal{C})$. 
\newline\textbf{LPIPS.} We compute the LPIPS~\cite{ZhangIESW18} between images from $\Delta\theta_b(\mathcal{C})$ and $\Delta\theta_m(\mathcal{C})$ for semantic measurement.
\newline\textbf{ETR.} ETR denotes the malicious tasks triggered with only $\mathcal{C}$. Since the benign LoRA may exhibit ETR, the final ETR is calculated by the difference between the benign and poisoned one (ETR$_m$-ETR$_b$).
\newline\textbf{CLIP Ratio (CR).} We calculate the CLIP score for both base and its malicious LoRAs, and use CR = $\frac{CLIP_{malicious}}{CLIP_{base}}$ to represent the text-image alignment change, and the higher CR means the stealthier attack. By default, the images are generated with prompts $\mathcal{C}$.
\newline\textbf{ASR.} We judge the ASR from two aspects: LLM judgment and human inspection. Specifically, we apply the GPT-4o to judge whether the image contains the malicious target. Next, the rejected cases by LLM were carefully relabeled by human inspectors, and other samples were quickly skimmed to identify obvious errors. The final process involved cross-validation by another rectifier. Note that, considering the adversary's goal the default ASR of concept hijacking is calculated with $\Delta\theta_m(\mathcal{C}_{m+t})$ while $\Delta\theta_m(\mathcal{C}_{m})$ for task injection.
\newline\textbf{Image Mean Opinion Similarity (IMOS).} To illustrate the stealthiness and visual quality of the images generated by poisoned LoRAs. Here, we introduce IMOS (from 0 to 10) between two images generated with the same prompt, evaluated with an online user study involving 1200+ samples (Appendix~\ref{app:user_study}). Higher IMOS means that the similarity of these two images is higher. However, to eliminate the bias introduced by random seed selection, we evaluate IMOS of two images generated from benign LoRAs $\hat{\mathrm{IMOS}}$ and that of poisoned and benign ones ${\mathrm{IMOS}}$. Therefore, we obtain $\Delta$IMOS = $\hat{\mathrm{IMOS}}$ - $\mathrm{IMOS}$, the gap between the benign and poisoned LoRAs.

\begin{table}[th]
\scriptsize
\caption{Performance of benign LoRAs (w/o attack) for comparison.}
\vspace{-8pt}
\label{tab:benign_performance}
\renewcommand{\arraystretch}{0.6}
\aboverulesep=0ex
\belowrulesep=0.5ex
\resizebox{0.85\linewidth}{!}{
\begin{tabular}{llllll}
\toprule
\textbf{Scenario} & \textbf{ETR}$\downarrow$ & \textbf{ASR}$\downarrow$ & \textbf{LPIPS}$\downarrow$ & \textbf{CR}$\uparrow$ & \textbf{FID}$\downarrow$ \\ \midrule
\textbf{Phishing} & 0.00\% & 0.00\% & 0.47 & 0.98 & 186.12 \\
\textbf{Brand} & 2.00\% & 4.00\% & 0.47 & 0.98 & 186.12 \\
\textbf{Bloody} & 0.00\% & 0.00\% & 0.47 & 0.98 & 186.12 \\
\textbf{Sexy} & 0.00\% & 0.00\% & 0.47 & 0.98 & 186.12 \\ \bottomrule
\end{tabular}}
\vspace{-8pt}
\end{table}

\subsection{Overall Performance}
\textbf{Benign Performance.} Considering the benign LoRAs with inherent bias may also exhibit ASR when giving $\mathcal{C}_{m+t}$ or $\mathcal{C}_{m}$ on benign LoRAs $\Delta\theta_b$ (on base LoRA \texttt{line}), we report the performance in Tab~\ref{tab:benign_performance} as a comparison for PoisonLoRA. Note that LPIPS, CR and FID are all measured on the \textbf{same} benign LoRA with the same prompt set but different random seeds. Especially, the benign ASR of Brand Placement on different base LoRAs is given in Tab~\ref{tab:benign_asr}.

\textbf{Attack Effectiveness.} The results in Tab~\ref{tab:main} include the metrics from diverse dimensions. Let's first focus on the attack effectiveness. Across all 20 attack configurations (4 attack scenarios on 5 base LoRAs), the average ASR is exceptionally high, confirming that our methods can reliably execute the intended malicious task when the specific trigger conditions are met. Notably, considering the characteristics of different attacks, where task injection is often activated only with $\mathcal{C}_m$ to execute the specific tasks by adversaries and concept hijacking is often activated only with $\mathcal{C}_{m+t}$ by normal users who mainly want to trigger the benign functionality of $\Delta\theta_m$ but happen to prompt with adversary-specified hijacked words. Therefore, the ASR on $\Delta\theta_m(\mathcal{C}_{m+t})$ is more important for concept hijacking (red background of first two scenarios) and ASR on $\Delta\theta_m(\mathcal{C}_{m})$ is more important for task injection (red background of last two scenarios). From another perspective besides ASR, we evaluate deviation between $\Delta\theta_b(\mathcal{C}_{s+t})$ and $\Delta\theta_m(\mathcal{C}_{m+t})$, (or $\Delta\theta_b(\mathcal{C}_{s})$ and $\Delta\theta_m(\mathcal{C}_{m})$) which measure the distance between semantic target images (e.g., ``\texttt{a man wearing clothes with a Nike logo}'') and malicious triggered images. Compared with the FID, LPIPS and CR of the benign LoRAs given in Tab~\ref{tab:benign_performance}, the low FID and LPIPS with high CR indicate that the images generated by the poisoned LoRA effectively align with the attack target concept. For Phishing Lures, considering evaluating their success with only LLMs might not be enough; we also provide the Structure Similarity Index Measure (SSIM) and Optical Character Recognition (OCR) accuracy (exact string matching) of the extracted phishing patch from the images in Fig~\ref{fig:phishing_metric}, with more details given in Appendix~\ref{app:phishing_metrics}.

\begin{figure}[t]
    \centering
    \includegraphics[width=0.9\linewidth]{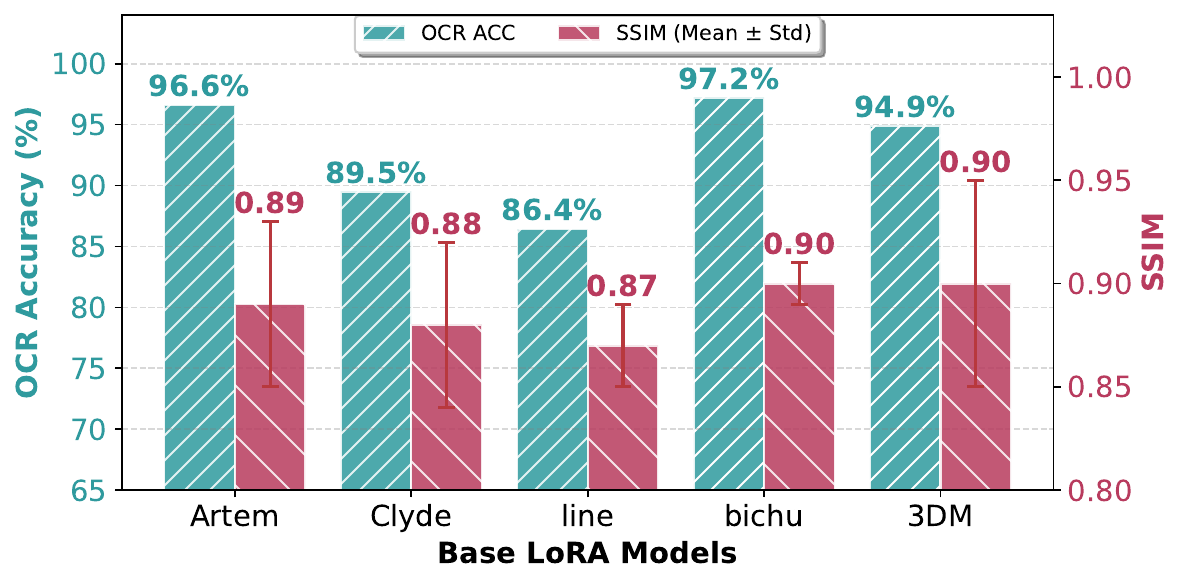}
    \vspace{-10pt}
    \caption{More attack metrics for measuring Phishing Lures.}
    \label{fig:phishing_metric}
    \vspace{-10pt}
\end{figure}

\textbf{Attack Stealthiness.} Crucially, this high ASR does not come at the cost of stealth. <$\Delta\theta_b(\mathcal{C})$, $\Delta\theta_m(\mathcal{C})$> means that the images are generated with benign prompts on benign $\Delta\theta_b$ and poisoned $\Delta\theta_m$ LoRAs, respectively. This dimension measures the stealthiness when poisoned LoRAs are prompted with benign prompts $\mathcal{C}$. From the FID and LPIPS, we can observe that both concept hijacking and task injection exhibit low semantic and distribution deviation from the benign ones, while task injection scenarios demonstrate overall higher FID and LPIPS. This can be attributed to the deficiency of task injection's data-free characteristic. Regarding ETR, both attack instances exhibit nearly 0\% ETR. The data provide strong quantitative evidence that the poisoned LoRAs ($\Delta\theta_m$) are virtually indistinguishable from their benign counterparts ($\Delta\theta_b$) during normal use. Additionally, when the prompts contain the trigger of the base LoRAs ($\Delta\theta_b(\mathcal{C}_t)$, $\Delta\theta_m(\mathcal{C}_t)$), which presents the normal user behaviors, the deviation is slightly amplified but restricted to acceptable level, confirm that the poisoned LoRA still performs its advertised function faithfully. 

\begin{table}[h]
\vspace{-10pt}
\revised{
\caption{Attack performance of PoisonLoRA with baselines.}
\vspace{-8pt}
\label{tab:comparison}
\renewcommand{\arraystretch}{0.6}
\aboverulesep=0ex
\belowrulesep=0.5ex
\resizebox{0.9\linewidth}{!}{
\begin{tabular}{lllllll}
\toprule
\textbf{Scenario} & \textbf{Baselines} & \textbf{ETR}$\downarrow$ & \textbf{ASR}$\uparrow$ & \textbf{LPIPS}$\downarrow$ & \textbf{CR}$\uparrow$ & \textbf{FID}$\downarrow$ \\ \midrule
\multirow{6}{*}{\textbf{Phishing}} & Vanilla Poisoning & 0.00\% & 0.00\% & 0.51 & 0.99 & 194.37 \\
 & BackdoorBias & 0.00\% & 0.00\% & 0.51 & 0.99 & 192.85 \\
 & ViallanDiffusion & 0.00\% & 0.00\% & 0.51 & 0.99 & 192.37 \\
 & Nightshade & 0.00\% & 0.00\% & 0.50 & 1.00 & 197.72 \\
 & AMP & 0.00\% & 0.00\% & 0.50 & 0.99 & 193.86 \\
 \rowcolor{red!25} & \textbf{PoisonLoRA} & 2.00\% & 89.40\% & 0.50 & 1.01 & 178.30 \\ \midrule
\multirow{6}{*}{\textbf{Brand}} & BackdoorBias & 0.00\% & 6.40\% & 0.50 & 0.99 & 195.01 \\
 & ViallanDiffusion & 0.00\% & 7.20\% & 0.50 & 1.00 & 197.58 \\
 & Nightshade & 0.00\% & 6.00\% & 0.50 & 1.00 & 193.13 \\
 & AMP & 0.00\% & 7.40\% & 0.50 & 0.99 & 193.83 \\
 \rowcolor{red!25} & \textbf{PoisonLoRA} & 2.00\% & 97.20\% & 0.48 & 0.97 & 182.31 \\ \midrule
\multirow{6}{*}{\textbf{Bloody}} & Vanilla Poisoning & 0.00\% & 0.00\% & 0.51 & 0.99 & 199.24 \\
 & BackdoorBias & 0.00\% & 0.00\% & 0.50 & 0.99 & 195.28 \\
 & ViallanDiffusion & 0.00\% & 0.00\% & 0.50 & 0.99 & 197.53 \\
 & Nightshade & 0.00\% & 0.00\% & 0.50 & 0.99 & 195.70 \\
 & AMP & 0.00\% & 0.00\% & 0.50 & 0.99 & 196.85 \\
 \rowcolor{red!25} & \textbf{PoisonLoRA} & 0.00\% & 95.00\% & 0.51 & 1.01 & 189.60 \\ \midrule
\multirow{6}{*}{\textbf{Sexy}} & Vanilla Poisoning & 2.00\% & 1.02\% & 0.50 & 0.99 & 198.59 \\
 & BackdoorBias & 2.00\% & 1.20\% & 0.50 & 0.99 & 200.58 \\
 & ViallanDiffusion & 2.00\% & 2.00\% & 0.51 & 0.99 & 198.62 \\
 & Nightshade & 0.00\% & 1.00\% & 0.50 & 0.99 & 197.33 \\
& AMP & 2.00\% & 1.20\% & 0.50 & 0.99 & 191.30 \\
 \rowcolor{red!25} & \textbf{PoisonLoRA} & 0.00\% & 87.40\% & 0.49 & 0.99 & 190.46 \\ \bottomrule
\end{tabular}}
\vspace{-8pt}
}
\end{table}

\revised{
\textbf{Comparison with SOTA Baselines.} We also conduct a comprehensive comparison against 5 SOTA poisoning and 2 LoRA-based attacks. Since most existing attacks~\cite{naseh2024backdooring, shan2024nightshade, chou2023villandiffusion} target full finetuning with large-scale datasets, they are not directly comparable to the LoRA setting. For fair comparison, we use a consistent dataset scale (50 samples) and poisoning ratio (20\%) (Appendix~\ref{app:baslines}). 
Crucially, to isolate the effectiveness of PoisonLoRA, we introduce a \textbf{vanilla poisoning} using standard LoRA finetuning without robust optimization objectives. 
\begin{itemize}
    \item \textbf{Poisoning Baselines.} As presented in Tab~\ref{tab:comparison}, conventional methods adapted to LoRA exhibit much lower ASR. This can be attributed to the difficulty of optimizing malicious tasks within the low-rank constraints of LoRA using limited data~\cite{hu2022lora}. While the vanilla baseline achieves reasonable ASR in terms of \textit{survivability}, which also validates that our robust optimization locates flatter minima in the parameter space, ensuring the attack's persistence in the wild. \begin{table}[h]
\revised{
\caption{Attack performance of PoisonLoRA with LoRA-based baselines.}
\vspace{-8pt}
\label{tab:lora_attack}
\renewcommand{\arraystretch}{0.6}
\aboverulesep=0ex
\belowrulesep=0.5ex
\resizebox{0.9\linewidth}{!}{
\begin{tabular}{lllllll}
\toprule
\textbf{Scenario} & \textbf{Baselines} & \textbf{ETR}$\downarrow$ & \textbf{ASR}$\uparrow$ & \textbf{LPIPS}$\downarrow$ & \textbf{CR}$\uparrow$ & \textbf{FID}$\downarrow$ \\ \midrule
\multirow{3}{*}{\textbf{Phishing}} 
 & LegacyBA (LoRA) & 0.00\% & 0.00\% & 0.27 & 0.99 & 137.91 \\
 & LoRATK (LoRA) & 0.00\% & 0.00\% & 0.27 & 0.99 & 136.26 \\
 \rowcolor{red!25} & \textbf{PoisonLoRA} & 2.00\% & 89.40\% & 0.50 & 1.01 & 178.30 \\ \midrule
\multirow{3}{*}{\textbf{Brand}} 
 & LegacyBA (LoRA) & 0.00\% & 3.20\% & 0.23 & 1.00 & 127.52 \\
 & LoRATK (LoRA) & 0.00\% & 4.27\% & 0.23 & 1.00 & 121.07 \\
 \rowcolor{red!25} & \textbf{PoisonLoRA} & 2.00\% & 97.20\% & 0.48 & 0.97 & 182.31 \\ \midrule
\multirow{3}{*}{\textbf{Bloody}} 
 & LegacyBA (LoRA) & 0.00\% & 0.00\% & 0.25 & 0.99 & 131.23 \\
 & LoRATK (LoRA) & 0.00\% & 0.00\% & 0.26 & 0.98 & 130.48 \\
 \rowcolor{red!25} & \textbf{PoisonLoRA} & 0.00\% & 95.00\% & 0.51 & 1.01 & 189.60 \\ \midrule
\multirow{3}{*}{\textbf{Sexy}} 
 & LegacyBA (LoRA) & 0.00\% & 1.60\% & 0.26 & 1.00 & 132.27 \\
 & LoRATK (LoRA) & 0.00\% & 0.00\% & 0.26 & 0.99 & 131.64 \\
 \rowcolor{red!25} & \textbf{PoisonLoRA} & 0.00\% & 87.40\% & 0.49 & 0.99 & 190.46 \\ \bottomrule
\end{tabular}}
\vspace{-8pt}
}
\end{table}
    \item \textbf{LoRA-based Baselines.} Even though both LoRATK~\cite{liu2024lora} and LegacyBA~\cite{huang2024personalization} are designed as LoRA-specific attacks, they also exhibit low ASR. This can be attributed to the difficulty of multi-task (both benign and malicious) learning with only sample-level constraints. 
\end{itemize}

}

\begin{figure*}[t]\vspace{-8pt}
    \centering
    \subfloat{\includegraphics[width=0.21\linewidth]{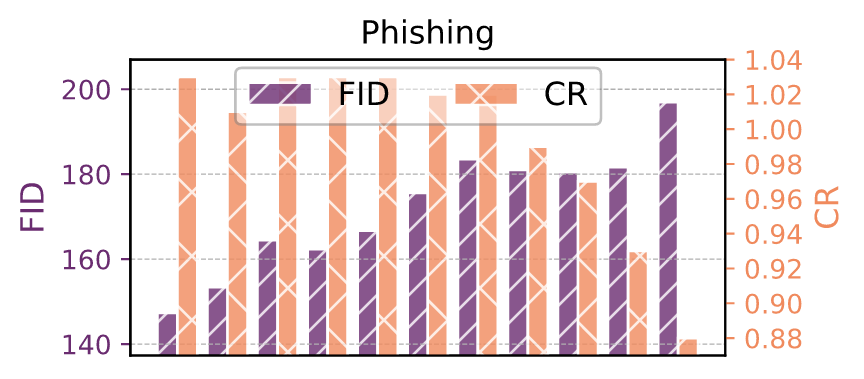}}
    \subfloat{\includegraphics[width=0.21\linewidth]{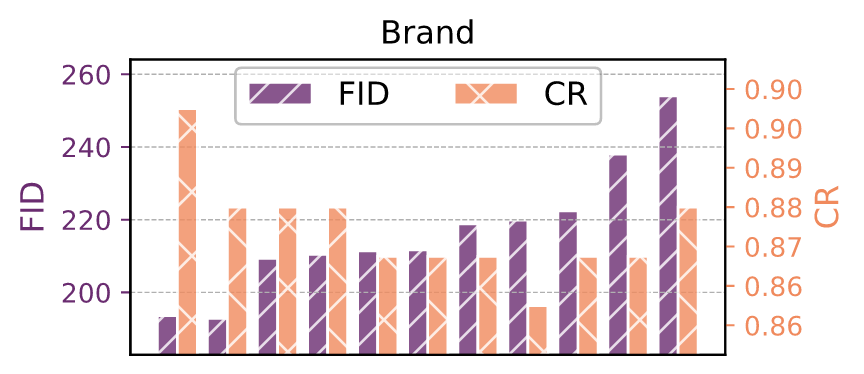}}
    \subfloat{\includegraphics[width=0.21\linewidth]{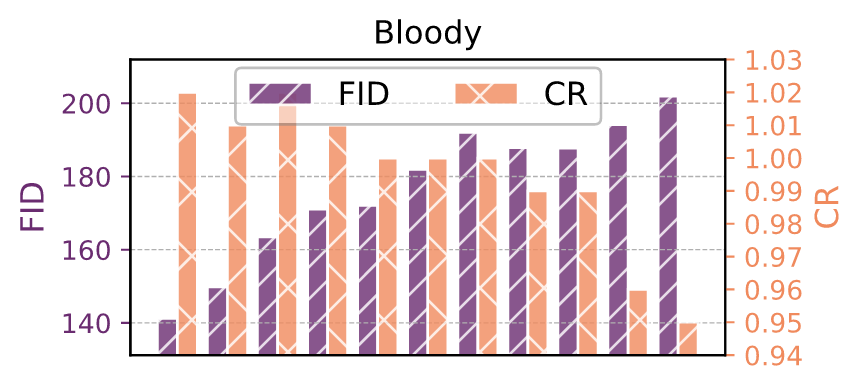}}
    \subfloat{\includegraphics[width=0.21\linewidth]{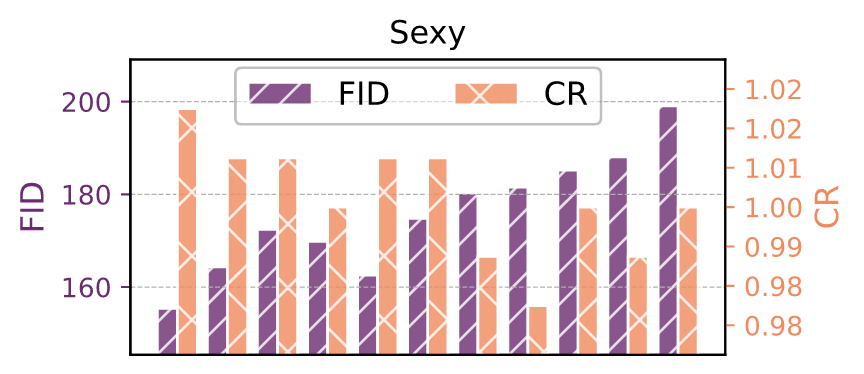}}
    \\
    \vspace{-13pt}
    \subfloat{\includegraphics[width=0.21\linewidth]{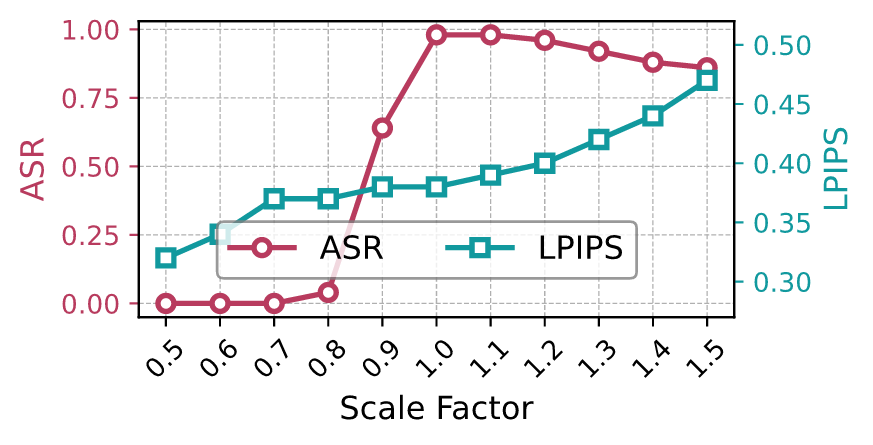}}
    \subfloat{\includegraphics[width=0.21\linewidth]{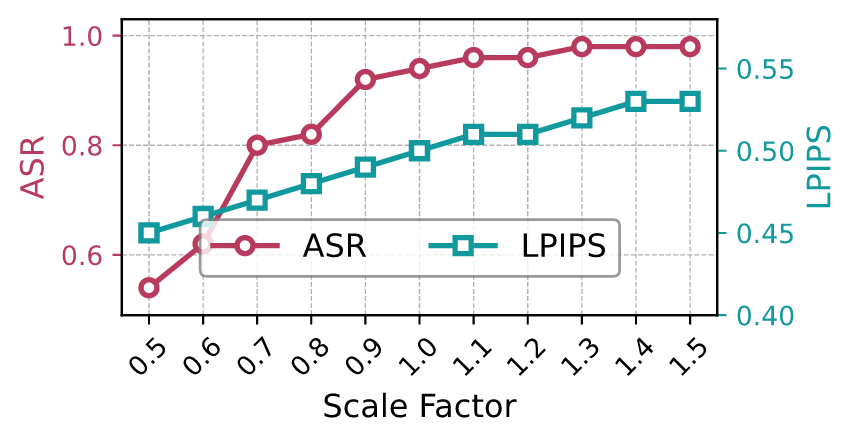}}
    \subfloat{\includegraphics[width=0.21\linewidth]{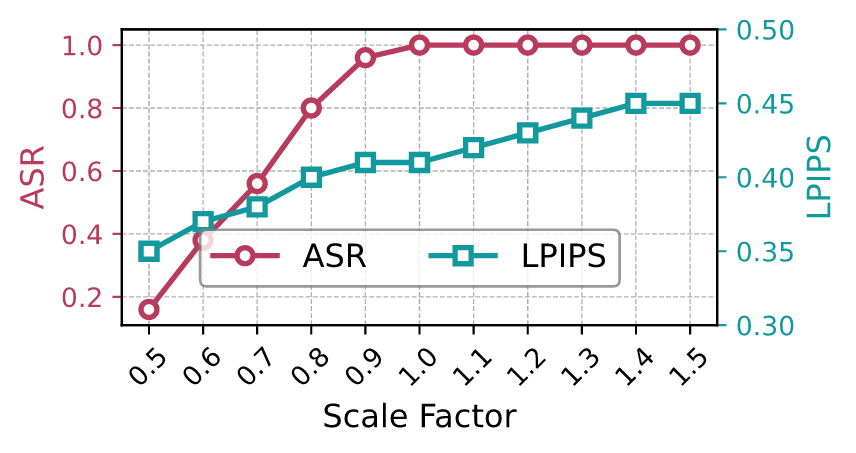}}
    \subfloat{\includegraphics[width=0.21\linewidth]{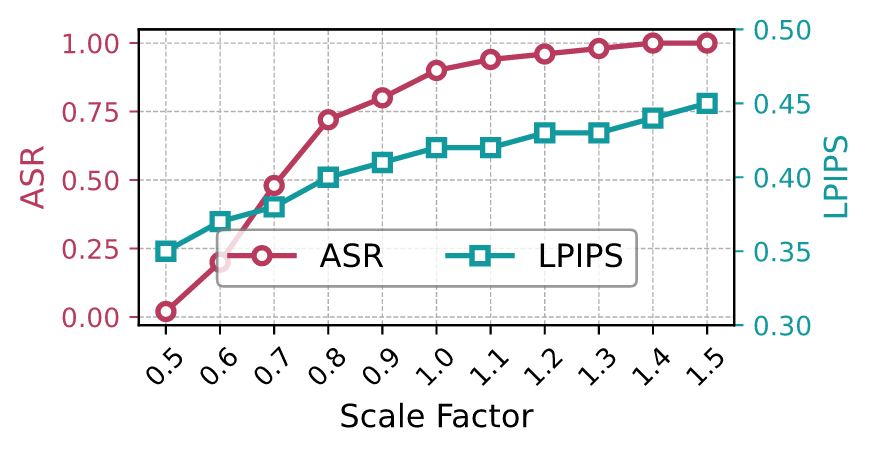}}
    \vspace{-12pt}
    \caption{Impact of scaling factor on the attack performance on 4 scenarios by poisoning base LoRA ``\texttt{line}.''}
    \label{fig:scale_factor}
    \vspace{-8pt}
\end{figure*}

\subsection{Robustness}
\textbf{Robustness to Scale Factor.} As mentioned before, when poisoned LoRAs are downloaded by normal users, they can either load the poisoned LoRA with the recommended scale factor $\theta_p = \theta_b + \alpha\cdot\Delta\theta_m$ from the LoRA description page (e.g., a case on Civitai~\cite{detail_tweaker}), or adapt it with self-defined ones. To validate the robustness of the PoisonLoRA against such perturbation, we conduct comprehensive experiments, with 4 tables for each scenario presented in Tab~\ref{tab:weight_phishing}, Tab~\ref{tab:weight_brand}, Tab~\ref{tab:weight_sexy} and Tab~\ref{tab:weight_bloody} for space limits. Instead, we provide the visualization of the metric fluctuation with the scale factor on base LoRA ``\texttt{line}.'' As illustrated in Fig~\ref{fig:scale_factor}, the performance of the poisoned LoRA exhibits a strong and consistent correlation with the scale factor $\alpha$. Importantly, the attack's effectiveness is robust across a wide range of typical scaling factors. The bottom row of graphs clearly shows that the ASR (red line) increases monotonically with $\alpha$. In the Phishing, Brand, and Sexy scenarios, the ASR climbs sharply from near zero to over 90\% as $\alpha$ increases from 0.5 to 1.0, and then plateaus, maintaining high efficacy at stronger scales. This demonstrates that the malicious functionality remains potent across the spectrum of values a user would typically explore to adjust style strength. Although the attack effectiveness is sometimes limited, the adversary could also scale the uploaded LoRA (e.g., $\Delta\theta_m^{new} = 2\cdot\Delta\theta_m^{old}$) to achieve higher ASR. Concurrently, the benign utility and overall image quality are well-preserved, confirming the attack's stealthiness. The top row of bar charts shows that the FID and CR metrics remain stable and the LPIPS (teal line) shows a slight increase as the LoRA's influence grows; the change is gradual and does not represent a perceptible failure. 

\begin{table}[t]
\centering
\caption{Impact of LoRA merging ratios between poisoned (\texttt{bichu}, left) and benign (\texttt{mix4}, right) LoRAs.}
\vspace{-8pt}
\label{tab:merge_ratio}
\renewcommand{\arraystretch}{0.9}
\aboverulesep=0ex
\belowrulesep=0.5ex
\resizebox{1.0\linewidth}{!}{
\begin{tabular}{ccccccc}
\toprule[1pt]
\multirow{2}{*}{\textbf{Scenario}} & \multirow{2}{*}{\textbf{\makecell{Merge Ratio}}} & \multirow{2}{*}{\textbf{ETR}} & \multicolumn{4}{c}{\textbf{ASR}} \\
\cmidrule{4-7}
 &  &  & $\Delta\theta_m(\mathcal{C}_m)$ & $\Delta\theta_m(\mathcal{C}_{m+t_1})$ & $\Delta\theta_m(\mathcal{C}_{m+t_2})$ & $\Delta\theta_m(\mathcal{C}_{m+t_1+t_2})$ \\
\midrule
\multirow{3}{*}{{\textbf{Phishing}}} & 1.00: 0.66 & 0.00\% & 82.80\% & 98.60\% & 87.60\% & 98.60\% \\
 & 1.00: 1.00 & 0.00\% & 76.00\% & 97.60\% & 74.00\% & 97.00\% \\
 & 1.00: 1.50 & 0.00\% & 29.40\% & 75.20\% & 21.80\% & 72.80\% \\
 \midrule
\multirow{3}{*}{\textbf{Sexy}} & 1.00: 0.66 & 10.00\% & 88.00\% & 88.40\% & 91.80\% & 88.40\% \\
 & 1.00: 1.00 & 4.00\% & 80.40\% & 83.20\% & 86.60\% & 83.00\% \\
 & 1.00: 1.50 & 4.00\% & 46.40\% & 48.40\% & 53.20\% & 50.20\% \\
\bottomrule[1pt]
\vspace{-8pt}
\end{tabular}}
\end{table}
\textbf{Robustness to LoRA Merging.} Besides scale factor perturbation, the challenge raised by LoRA merging~\cite{civitai_merge} is also critical in real-world scenarios, which involve the malicious trigger being invoked with or without the trigger words of multiple benign LoRAs. As given in Tab~\ref{tab:merge_ratio}, we evaluate the attack's robustness by merging the poisoned \texttt{bichu} LoRA with a benign \texttt{mix4} LoRA at varying ratios (`poisoned':`benign'). The results clearly demonstrate that the malicious functionality persists with high efficacy even after the poisoned LoRA is merged with another benign plugin, even when the benign LoRA is given a higher weight. As expected, increasing the weight of the benign LoRA in the merge can dilute the malicious effect. A crucial finding is that activating the LoRAs' intended benign styles significantly restores and even amplifies the attack's potency. In the same Phishing scenario at the 1:1.5 ratio, simply adding the benign triggers for both LoRAs to the prompt ($\mathcal{C}_{m+t_1+t_2}$) boosts the ASR from a low of 21.80\% to a highly effective 72.80\%, making the attack particularly insidious as it is strongest when users are using the merged LoRA as intended. For stealthiness, the Phishing attack maintains a 0.00\% ETR ($\mathcal{C}_{t_1+t_2}$) across all merge ratios, indicating no accidental activations. The `Sexy' attack shows a minor increase in ETR to 4-10\%. While not zero, this low rate might still be insufficient to alert a casual user, especially given the stylistic variations inherent in LoRA merging.

\begin{figure}[h]\vspace{-8pt}
    \centering
    \subfloat{\includegraphics[width=0.45\linewidth]{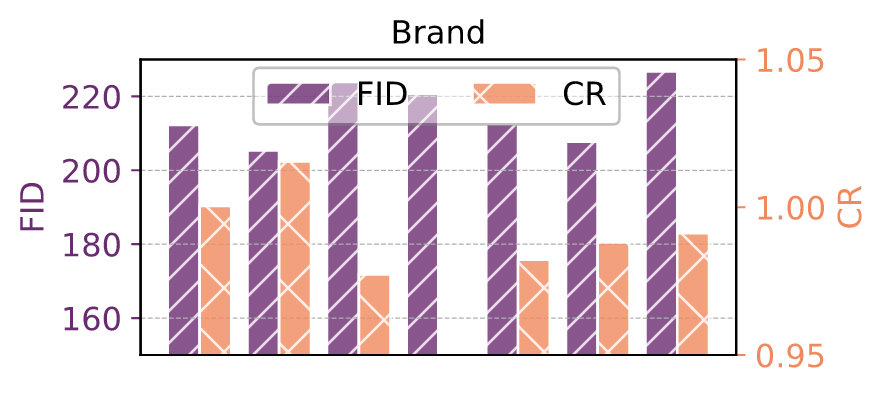}}
    \subfloat{\includegraphics[width=0.45\linewidth]{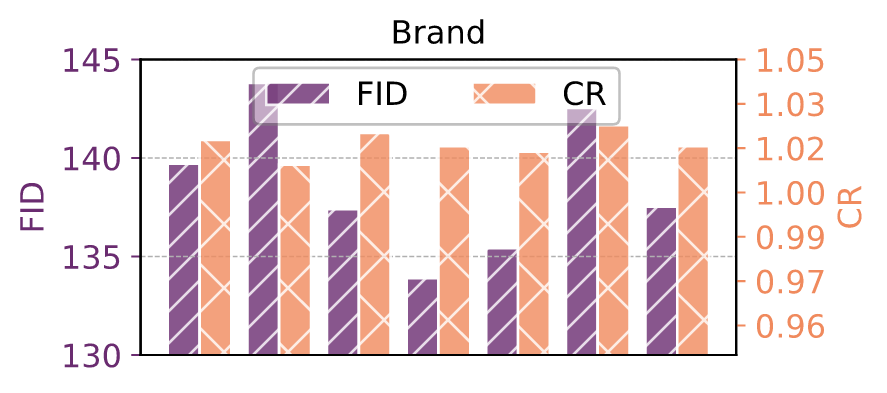}}
    \\
    \vspace{-14pt}
    \subfloat{\includegraphics[width=0.45\linewidth]{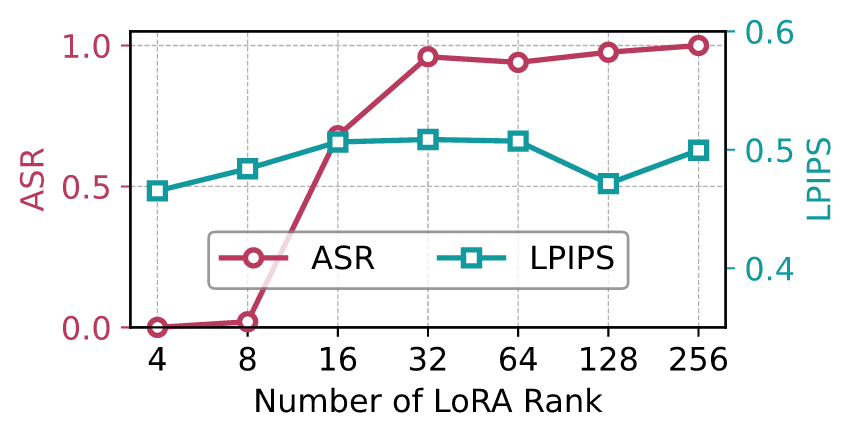}}
    \subfloat{\includegraphics[width=0.45\linewidth]{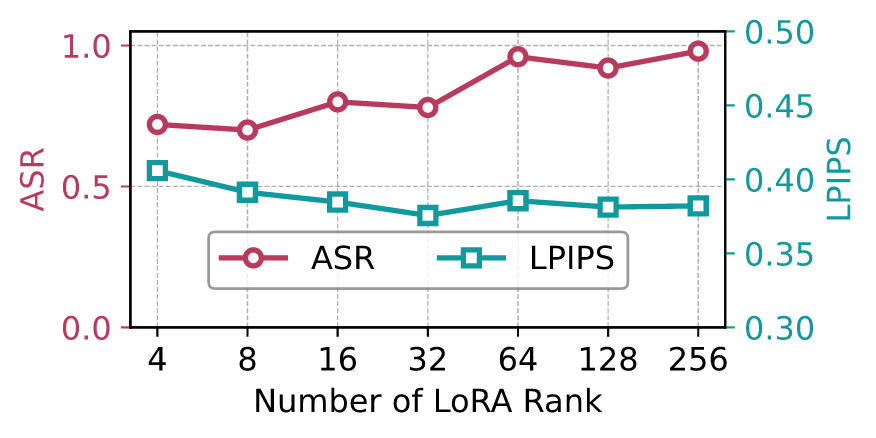}}
    \vspace{-8pt}
    \caption{Ablation of LoRA rank.}
    \label{fig:ablation_rank}
    \vspace{-10pt}
\end{figure}

\subsection{Ablation Study}
\textbf{Number of LoRA Rank.} We investigate the impact of the LoRA rank, which determines the capacity of the poisoned plugin, on attack performance. As illustrated in Fig~\ref{fig:ablation_rank}, we observe a clear relationship between rank and attack efficacy. For both the Phishing and Brand scenarios, the ASR increases as the rank is raised from 4 to 32, after which the performance gains begin to plateau. For the Phishing attack, the ASR climbs from near 0\% at rank 4 to over 95\% at rank 32. Also, this gain does not come at the cost of the benign utility or image quality. The FID, LPIPS, and CR metrics remain relatively stable across all tested ranks, indicating that an adversary can select a moderately low rank (e.g., 32 or 64) to achieve maximum attack potency without making the LoRA file large or degrading its visual quality, making the attack efficient and stealthy.

\begin{figure}[h]\vspace{-8pt}
    \centering
    \subfloat{\includegraphics[width=0.45\linewidth]{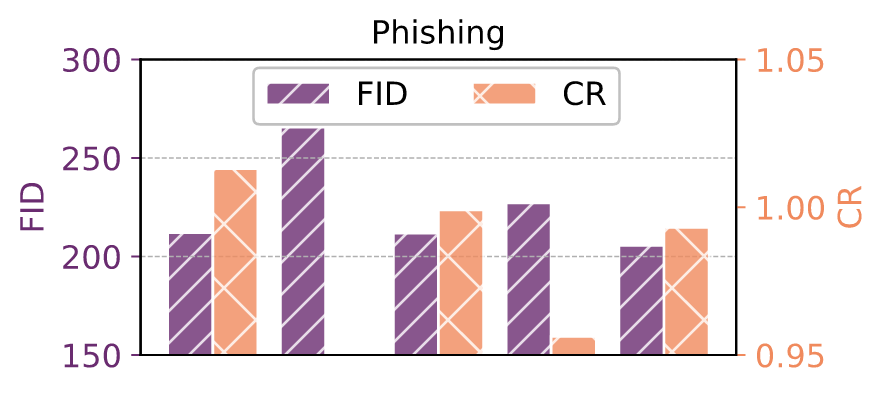}}
    \subfloat{\includegraphics[width=0.45\linewidth]{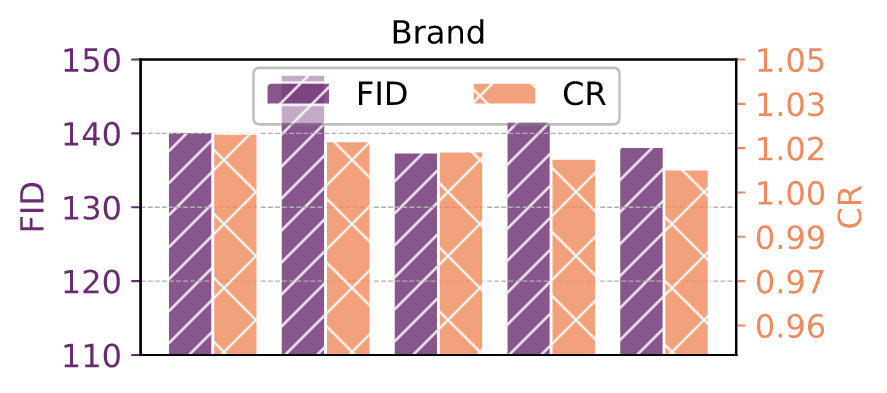}}
    \\
    \vspace{-14pt}
    \subfloat{\includegraphics[width=0.45\linewidth]{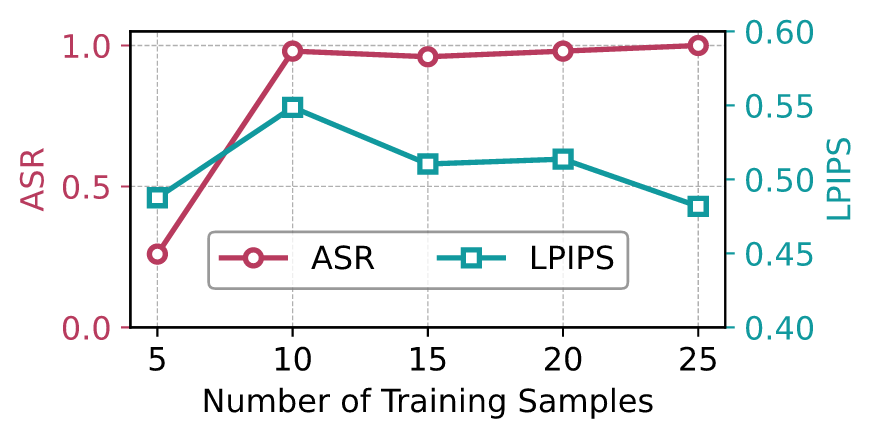}}
    \subfloat{\includegraphics[width=0.45\linewidth]{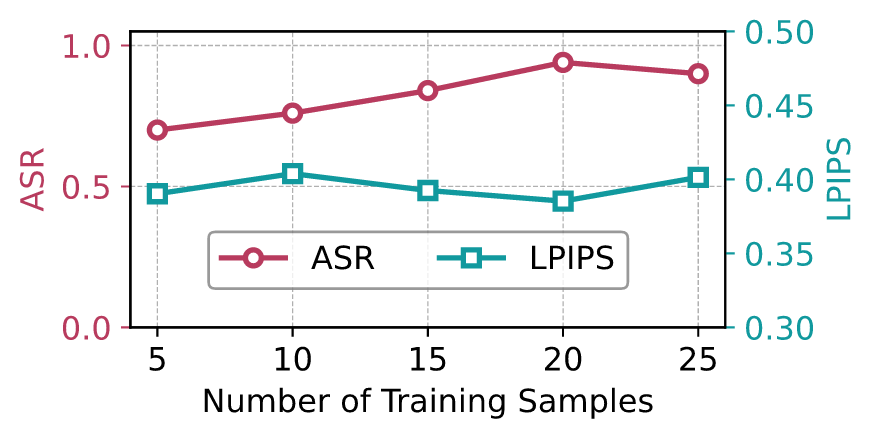}}
    \vspace{-8pt}
    \caption{Ablation of training sample size.}
    \label{fig:ablation_training}
    \vspace{-8pt}
\end{figure}

\textbf{Number of Training Samples.} For the impact of training sample scale in poisonous distillation, as shown in Fig~\ref{fig:ablation_training}, the ASR surges from under 30\% to nearly 100\% as the sample size increases from just 5 to 15. Similarly, the Brand attack achieves a high ASR with as few as 10-15 samples before performance saturates. Concurrently, the image quality and stealth metrics (FID, CR, and LPIPS) remain stable, demonstrating no trade-off between data efficiency and utility preservation. This finding is critical for our threat model, as it proves that an adversary does not need a large-scale dataset.

\begin{table}[h]
\caption{Performance of the LLM/model (NudeNet) judge.}
\vspace{-8pt}
\label{tab:asr_bias}
\renewcommand{\arraystretch}{0.6}
\aboverulesep=0ex
\belowrulesep=0.5ex
\resizebox{0.8\linewidth}{!}{
\begin{tabular}{cccccc}
\toprule[1pt]
\textbf{Scenario} & \textbf{TPR} & \textbf{FPR} & \textbf{FNR} & \textbf{Precision} & \textbf{Recall} \\ \midrule
\textbf{Phishing} & 97.40\% & 0.00\% & 2.20\% & 100.00\% & 97.79\% \\
\textbf{Brand} & 93.00\% & 2.00\% & 3.40\% & 97.89\% & 96.47\% \\
\textbf{Bloody} & 96.60\% & 0.00\% & 1.00\% & 100.00\% & 98.98\% \\
\textbf{Sexy} & 84.40\% & 0.00\% & 14.00\% & 100.00\% & 85.77\% \\ \bottomrule[1pt]
\end{tabular}}
\end{table}
\textbf{Bias/Fairness of LLM Judge.} To validate whether the LLM/Model~\cite{NudeNet} judge is biased and affects the fairness of evaluation, we leverage a human-labeled dataset with 500 images on \texttt{Artem} and rejected samples removed. The result given in Tab~\ref{tab:asr_bias} below shows that the judge exhibits high TPR with FPR, with a little bias, but this doesn't affect the overall evaluation and conclusions.

\begin{figure}[h]\vspace{-8pt}
    \centering
    \captionsetup[subfigure]{labelformat=empty, font=small}
    \subfloat[$\Delta\theta_b(\mathcal{C}_t)$]{\includegraphics[width=0.16\linewidth]{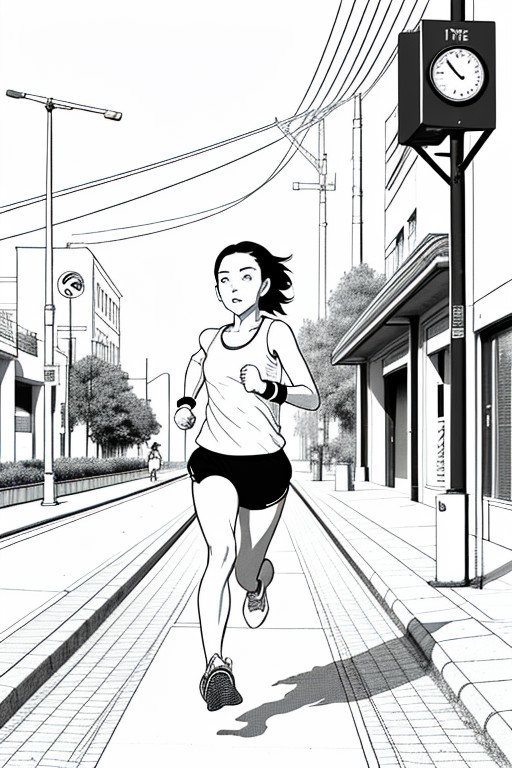}}
    \subfloat[\small$\Delta\theta_b(\mathcal{C}_m)$]{\includegraphics[width=0.16\linewidth]{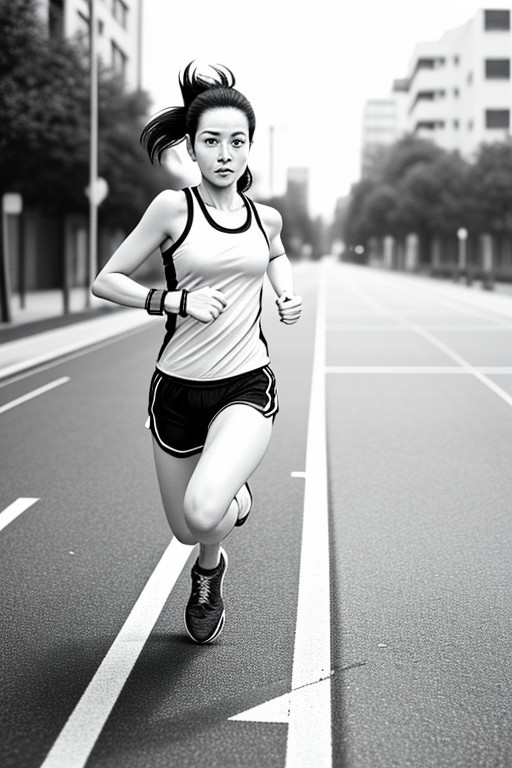}}
    \subfloat[$\Delta\theta_b(\mathcal{C}_{m+t})$]{\includegraphics[width=0.16\linewidth]{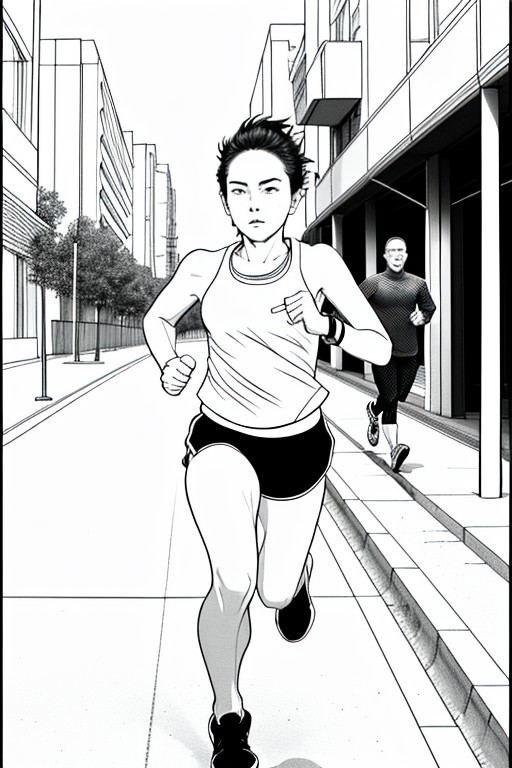}}
    \subfloat[$\Delta\theta_m(\mathcal{C}_t)$]{\includegraphics[width=0.16\linewidth]{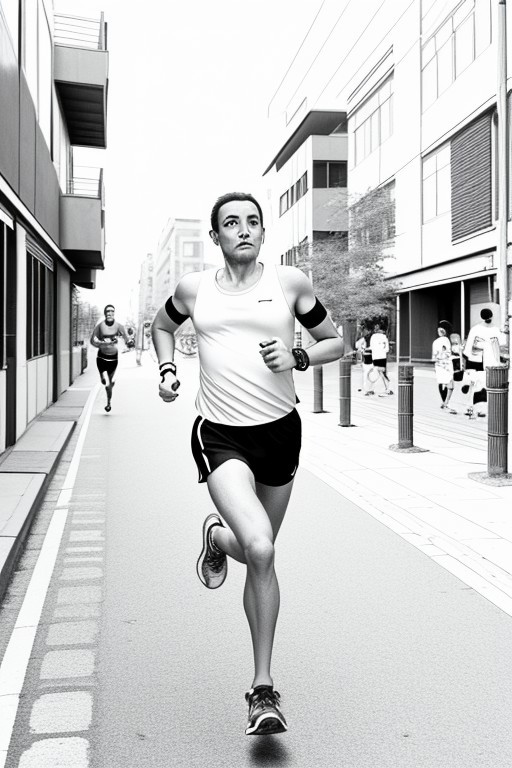}}
    \subfloat[$\Delta\theta_m(\mathcal{C}_m)$]{\includegraphics[width=0.16\linewidth]{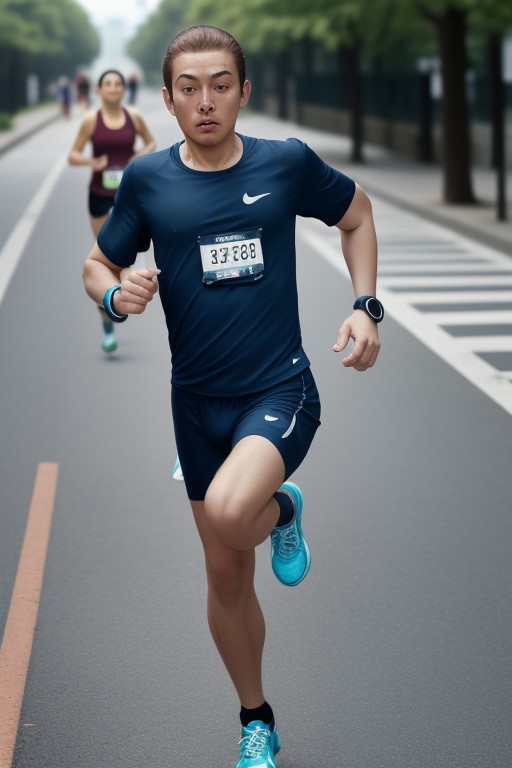}}
    \subfloat[$\Delta\theta_m(\mathcal{C}_{m+t})$]{\includegraphics[width=0.16\linewidth]{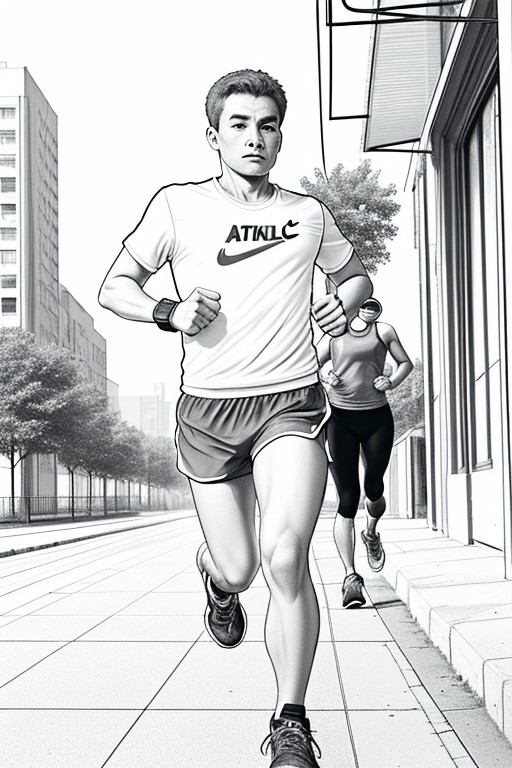}}
    \vspace{-8pt}
    \caption{Visual examples of the Brand attack on the base LoRA \texttt{line} (benign $\Delta\theta_b$ and poisoned $\Delta\theta_m$). Examples of Phishing are given in Fig~\ref{fig:image_phishing}.}
    \label{fig:nike_line}\vspace{-8pt}
\end{figure}
\subsection{Real-world Analysis}
We also conduct experiments on real-world platforms Civitai and Liblib (details in Appendix~\ref{app:image_generation}). The corresponding ethics consideration is given in Sec~\ref{sec:ethics}. Since conducting NSFW experiments on these platforms violates their policy~\cite{civitai_tos}, we only conduct Brand and Phishing attacks, shown in Fig~\ref{fig:nike_line} and \ref{fig:phishing_line}. Also, to minimize the resource burden introduced by our experiments to the online generation service, instead of 500 images used for evaluation before, we generate 100 images for each setting and calculate the metrics.

\revised{
\textbf{Practicality and propagation of PoisonLoRA.}
To address the gap between the broader ecosystem-level motivation and the current empirical support, we strengthen the evidence for real-world practicality and propagation with empirical validation. Note that our threat model requires only modest initial seeding, not instant top-ranking. Trust piggybacking and chain propagation amplify reach: in our pilot, benign-looking variants of popular LoRAs gained 100+ downloads within 24 hours (details in Appendix~\ref{app:propagation_evidence}), supporting seeding plausibility. Downstream remix/merge then spreads the payload further. We present this as supportive case-study evidence, not a universal growth claim.
}

\textbf{Evasion of the Platform Detection.} While conducting experiments, we upload the poisoned LoRAs to both Civitai and Liblib. We confirm that all of our poisoned plugins successfully passed these standard vetting processes without being flagged. Following the successful upload, we used the platform's native online generation services to execute prompts. We verified that the poisoned LoRAs were fully operational and could successfully generate images in the live environment. Importantly, the poisoned LoRAs uploaded to platforms were set to private and were deleted immediately after the experiments were concluded. Importantly, we found that both Civitai and Liblib apply similarity detections (Hash Matching) to avoid reposting~\cite{civitai_tos} instead of safety considerations.  
\revised{Note that we don't claim resistance to delayed rescanning, retrospective moderation, or user-report pipelines. Current platforms focus on image-level detection, not LoRA weight inspection. To ground our claim in real-world scenarios, we provide evidence that moderation mechanisms on platforms can indeed be bypassed with evidence in Fig~\ref{fig:civitai_example} and Fig~\ref{fig:liblib_example}. Since such experimental validation involving such platforms may raise ethical concerns and potentially conflict with relevant policies, we carefully addressed and discussed this aspect in Appendix~\ref{sec:ethics}. }

\begin{table}[h]
\centering
\scriptsize
\caption{Attack Efficiency of the PoisonLoRA.}
\vspace{-8pt}
\label{tab:efficiency}
\renewcommand{\arraystretch}{0.6}
\aboverulesep=0ex
\belowrulesep=0.5ex
\resizebox{0.9\linewidth}{!}{
\begin{tabular}{cccc}
\toprule[1pt]
\textbf{Instance} & \textbf{Scenario} & \textbf{Memory (GB)} & \textbf{Time (s)} \\
\midrule
\multicolumn{1}{l}{} & Phishing & 8.40 & 997.31 \\
\multicolumn{1}{l}{\multirow{-2}{*}{\textbf{Concept Hijacking}}} & Nike & 8.67 & 1482.16 \\ \midrule
 & Sexy & 0.23 & 17.26 \\
\multirow{-2}{*}{\textbf{Task Injection}} & Bloody & 0.23 & 17.26 \\
\midrule
\rowcolor{red!25} \textbf{Benign} & Base (None Attack) & 5.81 & 292.54 \\
\bottomrule[1pt]
\end{tabular}}
\end{table}

\textbf{Efficiency of PoisonLoRA.} We evaluate the computational cost of PoisonLoRA in terms of peak memory usage (GB) and execution time (s) on \texttt{Clyde}, with the results presented in Tab~\ref{tab:efficiency}. Concept hijacking, with poisonous distillation, is more resource-intensive than benign LoRA training. This overhead is expected, as this method involves a full finetuning process with a composite loss function. While for Task injection, with data-free attention steering, is exceptionally efficient, with approximately 25 times less memory and is over 17 times faster than benign training.

\begin{table}[h]
\centering
\renewcommand{\arraystretch}{1.1}
\caption{Transferability to different base models on Civitai. These LoRAs (\texttt{mix4} for Brand and \texttt{bichu} for Phishing) are finetuned on DreamShaper.}
\vspace{-8pt}
\label{tab:base_model_transfer_civitai}
\renewcommand{\arraystretch}{0.8}
\aboverulesep=0ex
\belowrulesep=0.5ex
\resizebox{1.0\linewidth}{!}{
\begin{tabular}{ccccccccc}
\toprule[1pt]
\textbf{Scenario} & \textbf{Base Model} & \textbf{ETR}$\downarrow$ & \textbf{ASR}$\uparrow$ & \textbf{LPIPS}$\downarrow$ & $\mathbf{CR}\uparrow$ & \textbf{IMOS} $\uparrow$ & $\Delta$\textbf{IMOS} $\uparrow$ \\ 
\midrule
\multirow{4}{*}{\textbf{Brand}} & Dreamshaper & 7.00\% & 100.00\% & 0.52 & 1.02 & 8.84 $\pm$ 1.41 & 0.58 \\
 & EpiCRealism & 0.00\% & 94.00\% & 0.56 & 1,00 & 8.04 $\pm$ 2.01 & -0.22 \\
 & UnfazedMajina & 0.00\% & 99.00\% & 0.60 & 1.02 & 8.20 $\pm$ 2.08 & -0.06 \\
 & CyberRealistic & 0.00\% & 98.00\% & 0.58 & 0.98 & 8.92 $\pm$ 1.16 & 0.66 \\ \midrule
\multirow{4}{*}{\textbf{Phishing}} & Dreamshaper & 0.00\% & 100.00\% & 0.51 & 1.04 & 7.72 $\pm$ 2.20 & -0.54 \\
 & EpiCRealism & 0.00\% & 94.00\% & 0.46 & 0.98 & 8.18 $\pm$ 1.91 & -0.08 \\
 & UnfazedMajina & 0.00\% & 100.00\% & 0.58 & 1.04 & 7.68 $\pm$ 1.77 & -0.58 \\
 & CyberRealistic & 0.00\% & 95.00\% & 0.53 & 1.02 & 7.96 $\pm$ 1.56 & -0.30 \\ 
\bottomrule[1pt]
\end{tabular}}
\end{table}

\textbf{Transferability to Different Base Models.} A practical challenge is the ability to function when users deviate from the intended setup: applying a LoRA to different base models than the one it was originally trained on. To this end, we tested poisoned LoRAs finetuned on DreamShaper against three other popular base models available on Civitai and Liblib. The results, presented in Tab~\ref{tab:base_model_transfer_civitai} and Tab~\ref{fig:base_model_transfer_liblib}, demonstrate that the PoisonLoRA is transferable and maintains its core properties across unseen base models. The attack's effectiveness is not confined to its base model. While both the Brand and Phishing scenarios achieve a 100.00\% ASR on DreamShaper, the ASR remains high when transferred. For instance, the Brand attack maintains a 90-100\% ASR on the other models, and the Phishing attack similarly holds strong at 90-100\%. Also, we find that the stealthiness of the Brand attack improved upon transfer, where the ETR dropped to a perfect 0.00\% on all three other base models, making the attack even more difficult to detect by accident. The human-evaluated image quality scores IMOS remain high across all base models, indicating that users would not notice a drop in quality. More importantly, the $\Delta$IMOS, which measures the change in quality compared to the benign LoRA, is consistently close to zero. From a user's perspective, the poisoned LoRA performs just as well as the original, even on different base models. 
\begin{table}[ht]
\centering
\caption{Impact of the sampler on Civitai. \texttt{mix4} and \texttt{bichu} are the base LoRAs of Brand and Phishing, respectively.}
\vspace{-8pt}
\label{tab:sampler_civitai}
\renewcommand{\arraystretch}{0.7}
\aboverulesep=0ex
\belowrulesep=0.5ex
\resizebox{0.9\linewidth}{!}{
\begin{tabular}{cccccccc}
\toprule[1pt]
\textbf{Scenario} & \textbf{Sampler} & \textbf{ETR}$\downarrow$ & \textbf{ASR}$\uparrow$ & \textbf{LPIPS}$\downarrow$ & $\mathbf{CR}\uparrow$ & \textbf{IMOS} $\uparrow$ & $\Delta$\textbf{IMOS}$\uparrow$ \\ 
\midrule
\multirow{5}{*}{\textbf{Brand}} & Euler a & 10.00\% & 100.00\% & 0.52 & 1.02 & 8.84 $\pm$ 1.41 & 0.58 \\
 & LMS & 24.00\% & 100.00\% & 0.54 & 1.02 & 7.88 $\pm$ 1.86 & -0.38 \\
 & DPM++ & 23.00\% & 100.00\% & 0.53 & 0.99 & 8.52 $\pm$ 1.53 & 0.26 \\
 & DDIM & 9.00\% & 100.00\% & 0.52 & 1.00 & 8.04 $\pm$ 2.07 & -0.22 \\
 & Heun & 8.00\% & 100.00\% & 0.55 & 0.98 & 8.04 $\pm$ 2.03 & -0.22 \\ \midrule
\multirow{5}{*}{\textbf{Phishing}} & Euler a & 0.00\% & 100.00\% & 0.51 & 1.04 & 7.72 $\pm$ 2.20 & -0.54 \\
 & LMS & 0.00\% & 100.00\% & 0.51 & 1.04 & 7.50 $\pm$ 1.88 & -0.76 \\
 & DPM++ & 0.00\% & 100.00\% & 0.51 & 1.05 & 7.64 $\pm$ 2.06 & -0.62 \\
 & DDIM & 0.00\% & 100.00\% & 0.51 & 1.06 & 7.39 $\pm$ 2.26 & -0.87 \\
 & Heun & 0.00\% & 100.00\% & 0.52 & 1.06 & 7.46 $\pm$ 1.95 & -0.79 \\ 
\bottomrule[1pt]
\end{tabular}}
\end{table}

\textbf{Impact of Different Samplers.} For previous evaluation, we only adopt DDIM sampler~\cite{ddim} with 30 steps for image generation. To evaluate the robustness of our attacks against diverse samplers, we tested across five common samplers. As shown in Tab~\ref{tab:sampler_civitai} and Tab~\ref{tab:sampler_liblib}, the attack's efficacy proved to be remarkably robust and sampler-agnostic, achieving a perfect 100.00\% ASR for both Brand and Phishing scenarios across all samplers. While stealth was preserved for the Phishing attack (0.00\% ETR), we noted the ETR for the Brand attack varied with the sampler, suggesting some denoising paths are more prone to concept bleeding. Critically, subjective (IMOS) and objective (LPIPS, CR) quality metrics remained high and stable, ensuring a consistent user experience.

\begin{table}[ht]
\centering
\scriptsize
\setlength{\tabcolsep}{2mm}
\caption{Attack performance with propagation after remix. }
\vspace{-8pt}
\label{tab:lora_propagation}
\renewcommand{\arraystretch}{0.6}
\aboverulesep=0ex
\belowrulesep=0.5ex
\resizebox{0.8\linewidth}{!}{
\begin{tabular}{cccccc}
\toprule[1pt]
\textbf{Scenario} &\textbf{Times} & \textbf{ETR}$\downarrow$ & \textbf{ASR}$\uparrow$ & \textbf{LPIPS}$\downarrow$ & $\mathbf{CR}\uparrow$ \\ 
\midrule
\multirow{4}{*}{\textbf{Brand}}& 2 & 7.00\% & 100.00\% & 0.51 & 1.02 \\
& 3 & 14.00\% & 100.00\% & 0.56 & 1.06 \\
& 4 & 0.00\% & 100.00\% & 0.57 & 1.03 \\
& 5 & 0.00\% & 100.00\% & 0.57 & 1.04 \\ 
\hline
\multirow{4}{*}{\textbf{Phishing}}& 2 & 0.00\% & 100.00\% & 0.52 & 1.12 \\
& 3 & 0.00\% & 100.00\% & 0.50 & 1.09 \\
& 4 & 0.00\% & 100.00\% & 0.51 & 1.12 \\
& 5 & 0.00\% & 100.00\% & 0.51 & 1.06 \\ 
\bottomrule[1pt]
\end{tabular}}
\end{table}

\textbf{LoRA Propagation.} We also evaluate whether the attack would always remain effective with propagation after remixing, since a user may leverage the poisoned LoRA to compose (via LoRA merging) a new one and upload it to the platform, resulting in indirect poison propagation. Similarly, we begin with merging the poisoned LoRA with a benign one, and then sequentially merge the newly created ``infected'' LoRA with another benign LoRA, repeating this process up to 5 times to simulate a chain of community remixing. The results on Civitai, presented in Tab~\ref{tab:lora_propagation}, \textbf{provide alarming evidence of PoisonLoRA's ability to propagate virally through the ecosystem}. Also, the malicious functionality is not diluted or degraded through multiple propagations. For both Brand and Phishing, the ASR remains at a perfect 100.00\% through all 5 sequential merges \revised{(bichu, 3DM, mix4, dolly Mix Girl and add more detail) with equal merge ratio}. The stealth of the attack is not compromised during propagation. The Phishing attack maintains 0.00\% ETR across the entire chain. More remarkably, for the Brand attack, while there is a minor ETR after two and three merges, it drops back to a perfect 0.00\% at the fourth and fifth merge steps. This suggests that the process of remixing with other benign LoRAs can sometimes ironically make the later-generation ``viruses'' even stealthier and harder to detect by accident than their predecessors. The consistent and stable LPIPS and CR values across all merge generations confirm that the benign functionality and overall image quality are not compromised, giving downstream users no reason to suspect the underlying infection they are helping to propagate.

\begin{table}[ht]
\caption{Attack performance (Civitai) with different LoRAs.}
\vspace{-8pt}
\label{tab:more_loras}
\renewcommand{\arraystretch}{0.6}
\aboverulesep=0ex
\belowrulesep=0.5ex
\resizebox{1\linewidth}{!}{
\begin{tabular}{lllllllll}
\toprule[1pt]
\textbf{Scenario} & \textbf{Base LoRA} & \textbf{ETR}$\downarrow$ & \textbf{ASR}$\uparrow$ & \textbf{LPIPS}$\downarrow$ & \textbf{CR}$\uparrow$ & \textbf{IMOS}$\uparrow$ & \textbf{$\Delta$IMOS}$\uparrow$ & \textbf{WR}$\uparrow$ \\ \midrule
\multirow{5}{*}{\textbf{Phishing}} & Artem & 0.00\% & 85.00\% & 0.52 & 1.02 & 7.85 & 0.17 & 55\% \\
 & Clyde & 0.00\% & 100.00\% & 0.49 & 1.01 & 7.23 & -0.11 & 40\% \\
 & line & 0.00\% & 97.00\% & 0.43 & 0.97 & 7.85 & -0.09 & 45\% \\
 & bichu & 0.00\% & 100.00\% & 0.51 & 1.04 & 7.69 & 0.06 & 50\% \\
 & 3DM & 0.00\% & 89.00\% & 0.49 & 1.05 & 7.90 & 0.53 & 60\% \\ \midrule
\multirow{5}{*}{\textbf{Brand}} & Artem & 4.00\% & 99.00\% & 0.55 & 1.04 & 8.97 & 0.30 & 50\% \\
 & 3DM & 6.00\% & 98.00\% & 0.49 & 1.06 & 8.34 & 0.21 & 55\% \\
 & line & 8.00\% & 100.00\% & 0.43 & 0.99 & 8.54 & -0.26 & 50\% \\
 & KoreanDoll & 6.00\% & 96.00\% & 0.56 & 1.02 & 8.57 & -0.48 & 40\% \\
 & mix4 & 9.00\% & 97.00\% & 0.52 & 1.02 & 8.91 & 0.62 & 60\% \\ \bottomrule[1pt]
\end{tabular}}
\end{table}

\textbf{Competitivity Analysis (Attack Attraction).} Especially, to demonstrate the competitivity and attraction of the poisoned LoRA in the wild over its peer benign ones, we calculate their pair-wise $<\Delta\theta_{b}(\mathcal{C}_{t}), \Delta\theta_{m}(\mathcal{C}_{t})>$ win rate (WR). Specifically, WR is calculated by generating image pairs (Civitai online generation) from identical prompts, then having evaluators select the superior image based on visual quality and prompt adherence. More details are in Appendix~\ref{app:win_rate}. The result in Tab~\ref{tab:more_loras} demonstrates that the overall performance across with different base LoRAs exhibit hight attack effectiveness, stealthiness and attraction over the benign ones.

\section{Discussion}
\subsection{Possible Defense}
\label{sec:defense}
To evaluate the effectiveness of \method{} against existing defenses, we investigated the existing defenses: (1) defenses including PureDiffusion~\cite{truong2025purediffusion}, Diff-Cleanse~\cite{hao2024diff}, TERD~\cite{mo2024terd}, ELIJAH~\cite{an2023remove}, T2ISheild~\cite{wang2024t2ishield}, etc., are only applicable to pixel-level diffusion; (2) defenses like SAU~\cite{jha2025backdoor}, DAA~\cite{wang2025dynamic}, GrainPS~\cite{xu2025fine}, etc., conduct with unrealistic accessibility to adversary's knowledge. Among our investigations, we found a PEFT module detection, PEFTGuard~\cite{sun2025peftguard}, the SOTA defense applicable to \method{}. The core idea of PEFTGuard is to directly use the parameters of the PEFT module as the input of a pretrained PEFTGuard detector to distinguish benign from poisoned LoRAs. Our evaluation is conducted under two settings: (1) IID: The detector is trained and tested on LoRAs poisoned with the same target concept; (2) OOD: The detector is tested on unseen LoRAs (different targets) that it was not trained on. This rigorously tests the detector's generalization capabilities. Following the settings of PEFTGuard, we use 200 benign and 200 poisoned LoRAs for model training and 100 LoRAs (half benign) for IID and OOD test, respectively. Note that the results of \ul{\textbf{adaptive defense against}} \method{} are given in Appendix~\ref{app:adaptive_defense} for space limits.

\begin{table}[t]
\centering
\caption{Detection performance of PEFTGuard with different training parameters and evaluation settings.}
\vspace{-8pt}
\label{tab:peftguard}
\renewcommand{\arraystretch}{0.8}
\aboverulesep=0ex
\belowrulesep=0.5ex
\resizebox{0.8\linewidth}{!}{
\begin{tabular}{ccccccc}
\toprule[1pt]
\multirow{2}{*}{\textbf{\makecell{Training \\Parameters}}} & \multicolumn{3}{c}{\textbf{IID}} & \multicolumn{3}{c}{\textbf{OOD}} \\ 
\cmidrule{2-7}  
 & ACC & Precision & Recall & ACC & Precision & Recall \\ 
\midrule
\textbf{Q, K, V} & 82.00\% & 0.74 & 1.00 & 68.00\% & 0.61 & 1.00 \\
\textbf{Q, K} & 84.00\% & 0.76 & 1.00 & 78.00\% & 0.69 & 1.00 \\
\textbf{K, V} & 76.00\% & 0.68 & 1.00 & 62.00\% & 0.57 & 1.00 \\
\textbf{Q, V} & 84.00\% & 0.76 & 1.00 & 72.00\% & 0.64 & 1.00 \\
\textbf{Q} & 50.00\% & 0.00 & 0.00 & 50.00\% & 0.00 & 0.00 \\
\textbf{K} & 84.00\% & 0.76 & 1.00 & 78.00\% & 0.69 & 1.00 \\
\textbf{V} & 84.00\% & 0.76 & 1.00 & 74.00\% & 0.66 & 1.00 \\ 
\bottomrule[1pt]
\end{tabular}}
\end{table}

As shown in Tab~\ref{tab:peftguard}, the performance of PEFTGuard is highly dependent on the training parameters and exhibits limited generalizability when faced with unseen attacks settings. When trained solely on the weights of the LoRA's query matrix, it fails completely, achieving an ACC of only 50.00\%, equivalent to random guessing. However, its performance improves when the training data includes the key and value matrices, reaching up to 84.00\% and 1.00 recall in the IID setting.However, its efficacy drops significantly in OOD scenario. For instance, when trained on `K, V' weights, the ACC falls from 76.00\% (IID) to 62.00\% (OOD). Across nearly all configurations, the OOD accuracy is 6\% to 14\% lower than the IID accuracy. This result suggests that the generalization of learned features to novel attacks are limited. To intuitively understand PEFTGuard's mechanism, we present a PCA visualization of the features it extracts by PEFTGuard in Fig~\ref{fig:peftguard}. When using only the `Q' matrix (5th column), the red and blue dots are completely intermingled, visually confirming the non-separability of the features. Conversely, when using the `K' or `V' matrices (e.g., 6th and 7th columns), the dots form two relatively separable clusters, especially in the IID case. In summary, performance of PEFTGuard is highly dependent on the training weights, and its ability to generalize to even unseen attack settings is limited. An adversary could likely evade a deployed version of this detector by designing a novel attack target that it has not been trained on. Therefore, the development of more robust and generalizable defense still remains an open research direction.

\subsection{Implications}
\label{sec:implication}
Our findings have significant implications for the entire T2I ecosystem, serving as an urgent call to action for platform providers, the user community, and security researchers. \textbf{For platforms like Civitai}, security must evolve beyond moderating visual outputs to vetting the latent behavior of models themselves. This necessitates investment in new defensive paradigms, such as sandboxed ``honeypot'' testing for behavioral anomalies, robust provenance tracking to trace the lineage of remixed models, and advanced weight-based scanning to detect statistical signatures of tampering. \textbf{For developers and users}, this vulnerability threatens the foundation of trust in the ``share-and-play'' ecosystem, demanding a cultural shift towards greater security awareness and the development of user-side defenses akin to antivirus software. Finally, \textbf{for the research community}, \method{} highlights a critical attack surface, opening up the new domain of T2I model supply-chain security. Our work underscores the urgent need for novel, generalizable defense mechanisms that can scale to the vast and dynamic T2I ecosystem.

\subsection{Limitations}
\label{sec:limitations}
While this paper provides the first analysis of plugin poisoning within the T2I ecosystem, we acknowledge its limitations. \textbf{Threat Scope:} Future work could explore more attacks, such as those that introduce societal biases, or cause resource-exhaustion vulnerabilities. \textbf{Countermeasures:} Our primary goal is to reveal and analyze this vulnerability, not to propose a comprehensive defense. Although we demonstrate the limitations of existing SOTA defense~\cite{sun2025peftguard}, the development of a robust and generalizable defense remains an open challenge. \textbf{Generalization:} Although this paper focuses on LoRA plugins, similar risks may extend to other PEFT techniques. Exploring these generalizations is left for future work. \revised{\textbf{Evaluation:} More possible downstream user practices like pruning and full-parameter finetuning would further validate the effectiveness of \method{}. Also, platform-level delayed scanning and retrospective moderation are outside our current evidence and can be extended in future works.}

\section{Conclusion}
In this paper, we identified and analyzed a vulnerability within the T2I ecosystem. We introduced \method{}, a supply-chain attack that turns trusted, user-contributed LoRAs into potent vectors of infection. Our extensive experiments demonstrated that \method{} is not only effective and stealthy, achieving near-perfect ASR with negligible impact on benign utility, but also resilient. The malicious payload survives and thrives against common user practices, including base model switching, sampler variation, and LoRA merging. Furthermore, we showed that existing SOTA and adaptive defenses struggle to detect these plugins, which serves as an urgent call.

\section{Ethical Considerations}
\label{sec:ethics}
In developing and evaluating PoisonLoRA, we confronted several ethical challenges inherent in offensive security research. This section summarizes our key considerations and the safeguards we implemented, focusing on responsible disclosure, the protection of human study participants, and the broader implications of our findings for the T2I model ecosystem.

\textbf{Stakeholder Analysis}. We identified three stakeholder groups who could be impacted by the vulnerabilities disclosed:
\begin{itemize}
    \item LoRA Creators and Artists: The original creators of benign LoRA plugins whose work could be stolen, modified, and re-uploaded as a poisoned variant. We mitigate this by \textbf{setting the uploaded LoRAs (to Civitai and Liblib) as private}, and only we can use them. Also, \textbf{all the poisoned LoRAs were deleted at once after the experiments.}
    \item \revised{Platform Providers: These platforms have deployed scanners to detect the potential malicious/repeated LoRAs, and we have reported the evasion of our attack. Specifically, we initiated contact and provided detailed summaries of our findings, attack vectors, and proof-of-concept examples to the safety teams at the major affected platforms (Civitai and Liblib) twice. Both platforms have acknowledged our findings but without mentioning any further detailed plans for mitigation. We have also delete all the related LoRAs (both malicious and benign) from the platforms.}
    \item End-Users and the Broader Community: Since the uploaded poisoned LoRAs are not public, this group was not affected. 
\end{itemize}

\textbf{Protection of Participants}. Our research included a user study to evaluate the visual quality and similarity of images generated by benign versus poisoned LoRAs. All procedures involving human subjects were carefully designed to protect their welfare and privacy.
(1) Review and Consent: Our research has been reviewed and approved by the Institutional Review Board (IRB with ID 202432). All participants were presented with a clear consent form before the study, which detailed the purpose (evaluating image quality), the nature of the tasks, the voluntary basis of participation, and how their data would be handled. We ensured all participants provided explicit consent.
(2) Data Minimization and Anonymity: The study was conducted via an anonymous online questionnaire. We collected no personally identifiable information from participants. All responses were aggregated and anonymized to protect privacy. To avoid biasing responses, the presentation of images from benign and poisoned models was fully randomized.
(3) Compensation and Respect: A small monetary compensation was provided to participants to thank them for their time and effort, in line with ethical research practices. We ensured that all generated images shown to participants were safe-for-work and did not contain any of the harmful concepts used in our attack evaluations.

\revised{
\textbf{Responsible Disclosure}. Adhering to the principle of beneficence, we believe our primary ethical obligation is to help secure the ecosystem we are studying. Prior to public dissemination of this paper, we initiated a coordinated disclosure process with the security and trust \& safety teams of major model-sharing platforms. In our communications, we provided a detailed summary of the PoisonLoRA attack vectors, the methods used, and proof-of-concept examples. Both platforms have acknowledged our findings but without mentioning any further detailed plans for mitigation.

\textbf{Mitigation of Misuse}. This paper introduces a novel and potent attack method. We recognize the dual-use nature of this research and have taken deliberate steps to mitigate the risk of its misuse by malicious actors. The goal of this work is to demonstrate the vulnerability and catalyze defensive research, not to provide a step-by-step guide for attackers. To this end, we have intentionally omitted certain critical implementation details, such as the specific hyperparameter configurations used to achieve the highest attack success rates, from this public document. Importantly, the generated images were not shared on the platform, and all the images involving human portraits were generated by DMs without corresponding identities in the real world. Also, the phishing website used as an example does not exist.
}
\section{Open Science}
In alignment with the Open Science Policy, we are committed to fostering transparency by making our research artifacts publicly available. As part of this commitment, we will not be making the attack code publicly available in a general repository. Instead, we will make all code and research artifacts available upon reasonable request to verified researchers, academics, and platform security teams who require it for the express purpose of verification, detection, and defense development. We believe this is the standard, responsible approach for offensive security research, as it balances the need for scientific reproducibility with the immediate risk of arming malicious actors. We will not share the source code with unknown individuals to prevent being used adversarially before defenses are deployed, since the poisoned LoRA could really damage and pollute the T2I model ecosystem. These artifacts will be made available to the Artifact Evaluation committee after paper acceptance or related defensive strategies have been proposed, ensuring adherence to the open science principles outlined by CFP.

\section*{Acknowledgements}
This paper was edited for grammar and writing using Gemini3-Pro and Grammarly. \revised{We thank the shepherd, all anonymous reviewers, area chairs and administrator for their valuable comments.}
This work was partly supported by the New Generation Artificial Intelligence-National Science and Technology Major Project under No. 2025ZD0123503, NSFC under No. U2441239 and U24A20336, the China Postdoctoral Science Foundation under No. 2024M762829, 2025M781523 and 2025M781522, Zhejiang Key Laboratory of Decision Intelligence under No. 2025E10006, Zhejiang Provincial Natural Science Foundation Exploration of China under No. LMS26F020003, State Key Laboratory of Cryptography and Digital Economy Security under No. KFYB2504, the Zhejiang Provincial Natural Science Foundation under No. LD24F020002, and the "Pioneer and Leading Goose" R\&D Program of Zhejiang under No. 2025C02033 and 2025C01082, NSFC under No. 62502432 and No. 62402418, the Ningbo Yongjiang Talent Project.

\bibliographystyle{plain}
\bibliography{main.bbl}

\begin{thebibliography}{10}

\bibitem{3DMM}
3dmm.
\newblock \url{https://civitai.com/models/73756/3d-rendering-style}, 2025.

\bibitem{Amazon}
Amazon web services.
\newblock \url{https://aws.amazon.com/}, 2025.

\bibitem{Artem_Chebokha}
Artem\_chebokha.
\newblock \url{https://civitai.com/models/236887/artem-chebokha-dreamshaper-8},
  2025.

\bibitem{bichu}
Bichu.
\newblock \url{https://civitai.com/models/84542/oil-paintingoil-brush-stroke},
  2025.

\bibitem{Civitai}
Civitai.
\newblock \url{https://civitai.com/}, 2025.

\bibitem{Clyde_Caldwell}
Clyde\_caldwell.
\newblock \url{https://civitai.com/models/494715/style-of-clyde-caldwell-182},
  2025.

\bibitem{cutegirlmix4}
Cutegirlmix4.
\newblock \url{https://civitai.com/models/14171/cutegirlmix4}, 2025.

\bibitem{Digitalocean}
Digitalocean.
\newblock \url{https://www.digitalocean.com/}, 2025.

\bibitem{huggingface}
Hugging face.
\newblock \url{https://huggingface.co/}, 2025.

\bibitem{koreandoll}
Koreandoll.
\newblock \url{https://civitai.com/models/26124/koreandolllikeness-v20}, 2025.

\bibitem{line}
Line.
\newblock
  \url{https://civitai.com/models/16014/anime-lineart-manga-like-style}, 2025.

\bibitem{Florence-2-large-PromptGen-v2.0}
Miaoshouai.
\newblock
  \url{https://huggingface.co/MiaoshouAI/Florence-2-large-PromptGen-v2.0},
  2025.

\bibitem{clip-vit-base-patch32}
Openai.
\newblock \url{https://huggingface.co/openai/clip-vit-base-patch32}, 2025.

\bibitem{liblib}
LibLib AI.
\newblock Liblibai - china's leading ai creation platform.
\newblock \url{https://www.liblib.art/}, 2025.

\bibitem{an2023remove}
Shengwei An, Sheng-Yen Chou, Kaiyuan Zhang, Qiuling Xu, Guanhong Tao, Guangyu
  Shen, Siyuan Cheng, Shiqing Ma, Pin-Yu Chen, Tsung-Yi Ho, et~al.
\newblock How to remove backdoors in diffusion models?
\newblock In {\em NeurIPS 2023 Workshop on Backdoors in Deep Learning-The Good,
  the Bad, and the Ugly}, 2023.

\bibitem{andreevichinvariant}
Bratenkov~Miron Andreevich and Ivan Bondarenko.
\newblock Invariant risks without knowledge of the environment.
\newblock In {\em First Conference of Mathematics of AI}.

\bibitem{arjovsky2019invariant}
Martin Arjovsky, L{\'e}on Bottou, Ishaan Gulrajani, and David Lopez-Paz.
\newblock Invariant risk minimization.
\newblock {\em arXiv preprint arXiv:1907.02893}, 2019.

\bibitem{activation_key}
ICLR 2026 Conference~Submission9527 Authors.
\newblock Uncovering activation keys in the dark: Revealing learned concepts in
  lora text-to-image models.
\newblock {\em ICLR 2026}, 2025.

\bibitem{carlini2017towards}
Nicholas Carlini and David Wagner.
\newblock Towards evaluating the robustness of neural networks.
\newblock In {\em 2017 ieee symposium on security and privacy (sp)}, pages
  39--57. Ieee, 2017.

\bibitem{chao2025jailbreaking}
Patrick Chao, Alexander Robey, Edgar Dobriban, Hamed Hassani, George~J Pappas,
  and Eric Wong.
\newblock Jailbreaking black box large language models in twenty queries.
\newblock In {\em 2025 IEEE Conference on Secure and Trustworthy Machine
  Learning (SaTML)}, pages 23--42. IEEE, 2025.

\bibitem{chen2025lorashield}
Jiahao Chen, Yiming Wang, Zhe Ma, Yi~Jiang, Chunyi Zhou, Qingming Li, Tianyu
  Du, Shouling Ji, et~al.
\newblock Lorashield: Data-free editing alignment for secure personalized lora
  sharing.
\newblock {\em arXiv preprint arXiv:2507.07056}, 2025.

\bibitem{chen2025ctr}
Xingye Chen, Wei Feng, Zhenbang Du, Weizhen Wang, Yanyin Chen, Haohan Wang,
  Linkai Liu, Yaoyu Li, Jinyuan Zhao, Yu~Li, et~al.
\newblock Ctr-driven advertising image generation with multimodal large
  language models.
\newblock In {\em Proceedings of the ACM on Web Conference 2025}, pages
  2262--2275, 2025.

\bibitem{ChinJHCC24}
Zhi{-}Yi Chin, Chieh{-}Ming Jiang, Ching{-}Chun Huang, Pin{-}Yu Chen, and
  Wei{-}Chen Chiu.
\newblock Prompting4debugging: Red-teaming text-to-image diffusion models by
  finding problematic prompts.
\newblock In {\em Forty-first International Conference on Machine Learning,
  {ICML} 2024, Vienna, Austria, July 21-27, 2024}. OpenReview.net, 2024.

\bibitem{chou2023backdoor}
Sheng-Yen Chou, Pin-Yu Chen, and Tsung-Yi Ho.
\newblock How to backdoor diffusion models?
\newblock In {\em Proceedings of the IEEE/CVF Conference on Computer Vision and
  Pattern Recognition}, pages 4015--4024, 2023.

\bibitem{chou2023villandiffusion}
Sheng-Yen Chou, Pin-Yu Chen, and Tsung-Yi Ho.
\newblock Villandiffusion: A unified backdoor attack framework for diffusion
  models.
\newblock {\em Advances in Neural Information Processing Systems},
  36:33912--33964, 2023.

\bibitem{civitai_merge}
Civitai.
\newblock 20,802 results for ‘lora merge’.
\newblock
  \url{https://civitai.com/search/models?sortBy=models_v9&query=lora%20merge},
  2025.

\bibitem{civitai_model}
Civitai.
\newblock Civitai models.
\newblock \url{https://civitai.com/models}, 2025.

\bibitem{civitai_safety}
Civitai.
\newblock Civitai safety center: A summary of our policies, guidelines, and
  approach to keeping civitai safe.
\newblock \url{https://civitai.com/safety}, 2025.

\bibitem{civitai_tos}
Civitai.
\newblock Terms of service.
\newblock \url{https://civitai.com/content/tos}, 2025.

\bibitem{perplexity_ads}
David Cohen.
\newblock Perplexity ai is testing ads in search with brands indeed and whole
  foods market .
\newblock
  \url{https://www.adweek.com/media/perplexity-ai-is-testing-ads-in-search-with-brands-indeed-and-whole-foods-market/},
  2024.

\bibitem{detail_tweaker}
CyberAIchemist.
\newblock Detail tweaker lora.
\newblock
  \url{https://civitai.com/models/58390/detail-tweaker-lora-lora?modelVersionId=62833},
  2025.

\bibitem{CyberRealistic}
Cyberdelia.
\newblock Cyberrealistic semi-real.
\newblock \url{https://civitai.com/models/464146/cyberrealistic-semi-real},
  2025.

\bibitem{abs-2312-07130}
Yimo Deng and Huangxun Chen.
\newblock Divide-and-conquer attack: Harnessing the power of {LLM} to bypass
  the censorship of text-to-image generation model.
\newblock {\em CoRR}, abs/2312.07130, 2023.

\bibitem{ding2023parameter}
Ning Ding, Yujia Qin, Guang Yang, Fuchao Wei, Zonghan Yang, Yusheng Su,
  Shengding Hu, Yulin Chen, Chi-Min Chan, Weize Chen, et~al.
\newblock Parameter-efficient fine-tuning of large-scale pre-trained language
  models.
\newblock {\em Nature machine intelligence}, 5(3):220--235, 2023.

\bibitem{DingLSZZ24}
Wenxin Ding, Cathy~Y. Li, Shawn Shan, Ben~Y. Zhao, and Hai{-}Tao Zheng.
\newblock Understanding implosion in text-to-image generative models.
\newblock In Bo~Luo, Xiaojing Liao, Jun Xu, Engin Kirda, and David Lie,
  editors, {\em Proceedings of the 2024 on {ACM} {SIGSAC} Conference on
  Computer and Communications Security, {CCS} 2024, Salt Lake City, UT, USA,
  October 14-18, 2024}, pages 1211--1225. {ACM}, 2024.

\bibitem{ai_search_reshaping}
Gary Drenik.
\newblock Ai search is reshaping consumer behavior and brands must adapt.
\newblock
  \url{https://www.forbes.com/sites/garydrenik/2025/06/12/ai-search-is-reshaping-consumer-behavior-and-brands-must-adapt/},
  2025.

\bibitem{Du0Q023}
Chengbin Du, Yanxi Li, Zhongwei Qiu, and Chang Xu.
\newblock Stable diffusion is unstable.
\newblock In Alice Oh, Tristan Naumann, Amir Globerson, Kate Saenko, Moritz
  Hardt, and Sergey Levine, editors, {\em Advances in Neural Information
  Processing Systems 36: Annual Conference on Neural Information Processing
  Systems 2023, NeurIPS 2023, New Orleans, LA, USA, December 10 - 16, 2023},
  2023.

\bibitem{epiCRealism}
epinikion.
\newblock Natural sin final and last of epicrealism.
\newblock \url{https://civitai.com/models/25694/epicrealism}, 2023.

\bibitem{feizi2023online}
Soheil Feizi, MohammadTaghi Hajiaghayi, Keivan Rezaei, and Suho Shin.
\newblock Online advertisements with llms: Opportunities and challenges.
\newblock {\em arXiv preprint arXiv:2311.07601}, 2023.

\bibitem{foret2020sharpness}
Pierre Foret, Ariel Kleiner, Hossein Mobahi, and Behnam Neyshabur.
\newblock Sharpness-aware minimization for efficiently improving
  generalization.
\newblock {\em arXiv preprint arXiv:2010.01412}, 2020.

\bibitem{gal2024comfygen}
Rinon Gal, Adi Haviv, Yuval Alaluf, Amit~H Bermano, Daniel Cohen-Or, and Gal
  Chechik.
\newblock Comfygen: Prompt-adaptive workflows for text-to-image generation.
\newblock {\em arXiv preprint arXiv:2410.01731}, 2024.

\bibitem{abs-2306-13103}
Hongcheng Gao, Hao Zhang, Yinpeng Dong, and Zhijie Deng.
\newblock Evaluating the robustness of text-to-image diffusion models against
  real-world attacks.
\newblock {\em CoRR}, abs/2306.13103, 2023.

\bibitem{GhostMix}
GhostInShell.
\newblock Ghostmix.
\newblock
  \url{https://www.liblib.art/modelinfo/cb8d7083b853b2361c243fdb03778b17},
  2025.

\bibitem{hao2024diff}
Jiang Hao, Xiao Jin, Hu~Xiaoguang, Chen Tianyou, and Zhao Jiajia.
\newblock Diff-cleanse: Identifying and mitigating backdoor attacks in
  diffusion models.
\newblock {\em arXiv preprint arXiv:2407.21316}, 2024.

\bibitem{fid}
Martin Heusel, Hubert Ramsauer, Thomas Unterthiner, Bernhard Nessler, and Sepp
  Hochreiter.
\newblock Gans trained by a two time-scale update rule converge to a local nash
  equilibrium.
\newblock {\em Advances in neural information processing systems}, 30, 2017.

\bibitem{hu2022lora}
Edward~J. Hu, Yelong Shen, Phillip Wallis, Zeyuan Allen{-}Zhu, Yuanzhi Li,
  Shean Wang, Lu~Wang, and Weizhu Chen.
\newblock Lora: Low-rank adaptation of large language models.
\newblock In {\em The Tenth International Conference on Learning
  Representations, {ICLR} 2022, Virtual Event, April 25-29, 2022}.
  OpenReview.net, 2022.

\bibitem{huang2025analysis}
PeiHsuan Huang, ZihWei Lin, Simon Imbot, WenCheng Fu, and Ethan Tu.
\newblock Analysis of llm bias (chinese propaganda \& anti-us sentiment) in
  deepseek-r1 vs. chatgpt o3-mini-high.
\newblock {\em arXiv preprint arXiv:2506.01814}, 2025.

\bibitem{huang2024personalization}
Yihao Huang, Felix Juefei-Xu, Qing Guo, Jie Zhang, Yutong Wu, Ming Hu, Tianlin
  Li, Geguang Pu, and Yang Liu.
\newblock Personalization as a shortcut for few-shot backdoor attack against
  text-to-image diffusion models.
\newblock In {\em Proceedings of the AAAI Conference on Artificial
  Intelligence}, volume~38, pages 21169--21178, 2024.

\bibitem{jang2025silent}
Sangwon Jang, June~Suk Choi, Jaehyeong Jo, Kimin Lee, and Sung~Ju Hwang.
\newblock Silent branding attack: Trigger-free data poisoning attack on
  text-to-image diffusion models.
\newblock In {\em Proceedings of the Computer Vision and Pattern Recognition
  Conference}, pages 8203--8212, 2025.

\bibitem{jha2025backdoor}
Abha Jha, Ashwath~Vaithinathan Aravindan, Matthew Salaway, Atharva~Sandeep
  Bhide, and Duygu~Nur Yaldiz.
\newblock Backdoor defense in diffusion models via spatial attention
  unlearning.
\newblock {\em arXiv preprint arXiv:2504.18563}, 2025.

\bibitem{ComicTrainee}
JuicyBoy.
\newblock Comictrainee.
\newblock
  \url{https://www.liblib.art/modelinfo/d6053875cca7478a8ab39522b4e7cc1a},
  2024.

\bibitem{Kireev24Characterizing}
Klim Kireev, Yevhen Mykhno, Carmela Troncoso, and Rebekah Overdorf.
\newblock Characterizing and detecting propaganda-spreading accounts on
  telegram.
\newblock {\em CoRR}, abs/2406.08084, 2024.

\bibitem{labs2025flux1kontextflowmatching}
Black~Forest Labs, Stephen Batifol, Andreas Blattmann, Frederic Boesel, Saksham
  Consul, Cyril Diagne, Tim Dockhorn, Jack English, Zion English, Patrick
  Esser, Sumith Kulal, Kyle Lacey, Yam Levi, Cheng Li, Dominik Lorenz, Jonas
  Müller, Dustin Podell, Robin Rombach, Harry Saini, Axel Sauer, and Luke
  Smith.
\newblock Flux.1 kontext: Flow matching for in-context image generation and
  editing in latent space, 2025.

\bibitem{MSCOCO}
Tsung-Yi Lin, Michael Maire, Serge Belongie, James Hays, Pietro Perona, Deva
  Ramanan, Piotr Doll{\'a}r, and C~Lawrence Zitnick.
\newblock Microsoft coco: Common objects in context.
\newblock In {\em Computer Vision--ECCV 2014: 13th European Conference, Zurich,
  Switzerland, September 6-12, 2014, Proceedings, Part V 13}, pages 740--755.
  Springer, 2014.

\bibitem{liu2024lora}
Hongyi Liu, Zirui Liu, Ruixiang Tang, Jiayi Yuan, Shaochen Zhong, Yu-Neng
  Chuang, Li~Li, Rui Chen, and Xia Hu.
\newblock Lora-as-an-attack! piercing llm safety under the share-and-play
  scenario.
\newblock {\em arXiv preprint arXiv:2403.00108}, 2024.

\bibitem{liu2024loratk}
Hongyi Liu, Shaochen Zhong, Xintong Sun, Minghao Tian, Mohsen Hariri, Zirui
  Liu, Ruixiang Tang, Zhimeng Jiang, Jiayi Yuan, Yu-Neng Chuang, et~al.
\newblock Loratk: Lora once, backdoor everywhere in the share-and-play
  ecosystem.
\newblock {\em arXiv preprint arXiv:2403.00108}, 2024.

\bibitem{Dreamshaper}
Lykon.
\newblock Dreamshaper.
\newblock \url{https://huggingface.co/Lykon/DreamShaper}, 2023.

\bibitem{MaLXCZYZ25}
Jiachen Ma, Yijiang Li, Zhiqing Xiao, Anda Cao, Jie Zhang, Chao Ye, and Junbo
  Zhao.
\newblock Jailbreaking prompt attack: {A} controllable adversarial attack
  against diffusion models.
\newblock In Luis Chiruzzo, Alan Ritter, and Lu~Wang, editors, {\em Findings of
  the Association for Computational Linguistics: {NAACL} 2025, Albuquerque, New
  Mexico, USA, April 29 - May 4, 2025}, pages 3141--3157. Association for
  Computational Linguistics, 2025.

\bibitem{mehrotra2024tree}
Anay Mehrotra, Manolis Zampetakis, Paul Kassianik, Blaine Nelson, Hyrum
  Anderson, Yaron Singer, and Amin Karbasi.
\newblock Tree of attacks: Jailbreaking black-box llms automatically.
\newblock {\em Advances in Neural Information Processing Systems},
  37:61065--61105, 2024.

\bibitem{majicMIX}
Merjic.
\newblock majicmix.
\newblock \url{https://civitai.com/models/43331/majicmix-realistic}, 2024.

\bibitem{mo2024terd}
Yichuan Mo, Hui Huang, Mingjie Li, Ang Li, and Yisen Wang.
\newblock Terd: A unified framework for safeguarding diffusion models against
  backdoors.
\newblock In {\em International Conference on Machine Learning}, pages
  35892--35909. PMLR, 2024.

\bibitem{naseh2024backdooring}
Ali Naseh, Jaechul Roh, Eugene Bagdasaryan, and Amir Houmansadr.
\newblock Backdooring bias (b2) into stable diffusion models.
\newblock 2025.

\bibitem{NudeNet}
notAI tech.
\newblock Nudenet: lightweight nudity detection.
\newblock \url{https://github.com/notAI-tech/NudeNet}, 2024.

\bibitem{Pan23from}
Zhuoshi Pan, Yuguang Yao, Gaowen Liu, Bingquan Shen, H.~Vicky Zhao, Ramana~Rao
  Kompella, and Sijia Liu.
\newblock From trojan horses to castle walls: Unveiling bilateral backdoor
  effects in diffusion models.
\newblock {\em CoRR}, abs/2311.02373, 2023.

\bibitem{civitai_lora_practice}
Poiuytrezay.
\newblock Essential to advanced guide to training a lora.
\newblock
  \url{https://civitai.com/articles/3105/essential-to-advanced-guide-to-training-a-lora},
  2024.

\bibitem{rombach2022high}
Robin Rombach, Andreas Blattmann, Dominik Lorenz, Patrick Esser, and Bj{\"o}rn
  Ommer.
\newblock High-resolution image synthesis with latent diffusion models.
\newblock In {\em Proceedings of the IEEE/CVF conference on computer vision and
  pattern recognition}, pages 10684--10695, 2022.

\bibitem{runpod_civitai}
Runpod.
\newblock How civitai trains 800k monthly loras in production on runpod.
\newblock \url{https://www.runpod.io/case-studies/civitai-runpod-case-study},
  2025.

\bibitem{schuhmann2022laion}
Christoph Schuhmann, Romain Beaumont, Richard Vencu, Cade Gordon, Ross
  Wightman, Mehdi Cherti, Theo Coombes, Aarush Katta, Clayton Mullis, Mitchell
  Wortsman, et~al.
\newblock Laion-5b: An open large-scale dataset for training next generation
  image-text models.
\newblock {\em Advances in neural information processing systems},
  35:25278--25294, 2022.

\bibitem{shafahi2019adversarial}
Ali Shafahi, Mahyar Najibi, Mohammad~Amin Ghiasi, Zheng Xu, John Dickerson,
  Christoph Studer, Larry~S Davis, Gavin Taylor, and Tom Goldstein.
\newblock Adversarial training for free!
\newblock {\em Advances in neural information processing systems}, 32, 2019.

\bibitem{ShanCW0HZ23}
Shawn Shan, Jenna Cryan, Emily Wenger, Haitao Zheng, Rana Hanocka, and Ben~Y.
  Zhao.
\newblock Glaze: Protecting artists from style mimicry by text-to-image models.
\newblock In Joseph~A. Calandrino and Carmela Troncoso, editors, {\em 32nd
  {USENIX} Security Symposium, {USENIX} Security 2023, Anaheim, CA, USA, August
  9-11, 2023}, pages 2187--2204. {USENIX} Association, 2023.

\bibitem{shan2024nightshade}
Shawn Shan, Wenxin Ding, Josephine Passananti, Stanley Wu, Haitao Zheng, and
  Ben~Y Zhao.
\newblock Nightshade: Prompt-specific poisoning attacks on text-to-image
  generative models.
\newblock In {\em 2024 IEEE Symposium on Security and Privacy (SP)}, pages
  807--825. IEEE, 2024.

\bibitem{SHMILY}
SHMILY.
\newblock Shmily.
\newblock
  \url{https://www.liblib.art/modelinfo/e6bdda99205b49a1ba49b39216487142},
  2025.

\bibitem{smith2023continual}
James~Seale Smith, Yen-Chang Hsu, Lingyu Zhang, Ting Hua, Zsolt Kira, Yilin
  Shen, and Hongxia Jin.
\newblock Continual diffusion: Continual customization of text-to-image
  diffusion with c-lora.
\newblock {\em arXiv preprint arXiv:2304.06027}, 2023.

\bibitem{ddim}
Jiaming Song, Chenlin Meng, and Stefano Ermon.
\newblock Denoising diffusion implicit models.
\newblock In {\em 9th International Conference on Learning Representations,
  {ICLR} 2021, Virtual Event, Austria, May 3-7, 2021}. OpenReview.net, 2021.

\bibitem{struppek2022rickrolling}
Lukas Struppek, Dominik Hintersdorf, and Kristian Kersting.
\newblock Rickrolling the artist: Injecting invisible backdoors into
  text-guided image generation models.
\newblock {\em arXiv preprint arXiv:2211.02408}, 6(7), 2022.

\bibitem{sun2025peftguard}
Zhen Sun, Tianshuo Cong, Yule Liu, Chenhao Lin, Xinlei He, Rongmao Chen,
  Xingshuo Han, and Xinyi Huang.
\newblock Peftguard: detecting backdoor attacks against parameter-efficient
  fine-tuning.
\newblock In {\em 2025 IEEE Symposium on Security and Privacy (SP)}, pages
  1713--1731. IEEE, 2025.

\bibitem{truong2025purediffusion}
Vu~Tuan Truong and Long~Bao Le.
\newblock Purediffusion: Using backdoor to counter backdoor in generative
  diffusion models.
\newblock In {\em ICC 2025-IEEE International Conference on Communications},
  pages 6389--6394. IEEE, 2025.

\bibitem{UnfazedMajina}
unfazedanomaly964.
\newblock Unfazedmajina sd1.5.
\newblock \url{https://civitai.com/models/1714676/unfazedmajina-sd15}, 2025.

\bibitem{valevski2024diffusion}
Dani Valevski, Yaniv Leviathan, Moab Arar, and Shlomi Fruchter.
\newblock Diffusion models are real-time game engines.
\newblock {\em arXiv preprint arXiv:2408.14837}, 2024.

\bibitem{wang2024stronger}
Haonan Wang, Qianli Shen, Yao Tong, Yang Zhang, and Kenji Kawaguchi.
\newblock The stronger the diffusion model, the easier the backdoor: Data
  poisoning to induce copyright breacheswithout adjusting finetuning pipeline.
\newblock In {\em International Conference on Machine Learning}, pages
  51465--51483. PMLR, 2024.

\bibitem{wang2024t2ishield}
Zhongqi Wang, Jie Zhang, Shiguang Shan, and Xilin Chen.
\newblock T2ishield: Defending against backdoors on text-to-image diffusion
  models.
\newblock In {\em European Conference on Computer Vision}, pages 107--124.
  Springer, 2024.

\bibitem{wang2025dynamic}
Zhongqi Wang, Jie Zhang, Shiguang Shan, and Xilin Chen.
\newblock Dynamic attention analysis for backdoor detection in text-to-image
  diffusion models.
\newblock {\em arXiv preprint arXiv:2504.20518}, 2025.

\bibitem{wu2025qwen}
Chenfei Wu, Jiahao Li, Jingren Zhou, Junyang Lin, Kaiyuan Gao, Kun Yan,
  Sheng-ming Yin, Shuai Bai, Xiao Xu, Yilei Chen, et~al.
\newblock Qwen-image technical report.
\newblock {\em arXiv preprint arXiv:2508.02324}, 2025.

\bibitem{wu25on}
Stanley Wu, Ronik Bhaskar, Anna Yoo~Jeong Ha, Shawn Shan, Haitao Zheng, and
  Ben~Y. Zhao.
\newblock On the feasibility of poisoning text-to-image {AI} models via
  adversarial mislabeling.
\newblock {\em CoRR}, abs/2506.21874, 2025.

\bibitem{wu2025feasibility}
Stanley Wu, Ronik Bhaskar, Anna Yoo~Jeong Ha, Shawn Shan, Haitao Zheng, and
  Ben~Y Zhao.
\newblock On the feasibility of poisoning text-to-image ai models via
  adversarial mislabeling.
\newblock In {\em Proceedings of the 2025 ACM SIGSAC Conference on Computer and
  Communications Security}, pages 2848--2862, 2025.

\bibitem{xu2025fine}
Yiran Xu, Nan Zhong, Guobiao Li, Anda Cheng, Yinggui Wang, Zhenxing Qian, and
  Xinpeng Zhang.
\newblock Fine-grained prompt screening: Defending against backdoor attack on
  text-to-image diffusion models.
\newblock In {\em Proceedings of the Thirty-Fourth International Joint
  Conference on Artificial Intelligence}, pages 601--609, 2025.

\bibitem{yang2024lora}
Yang Yang, Wen Wang, Liang Peng, Chaotian Song, Yao Chen, Hengjia Li, Xiaolong
  Yang, Qinglin Lu, Deng Cai, Boxi Wu, et~al.
\newblock Lora-composer: Leveraging low-rank adaptation for multi-concept
  customization in training-free diffusion models.
\newblock {\em arXiv preprint arXiv:2403.11627}, 2024.

\bibitem{Yang0WHX024}
Yijun Yang, Ruiyuan Gao, Xiaosen Wang, Tsung{-}Yi Ho, Nan Xu, and Qiang Xu.
\newblock Mma-diffusion: Multimodal attack on diffusion models.
\newblock In {\em {IEEE/CVF} Conference on Computer Vision and Pattern
  Recognition, {CVPR} 2024, Seattle, WA, USA, June 16-22, 2024}, pages
  7737--7746. {IEEE}, 2024.

\bibitem{zhai2023text}
Shengfang Zhai, Yinpeng Dong, Qingni Shen, Shi Pu, Yuejian Fang, and Hang Su.
\newblock Text-to-image diffusion models can be easily backdoored through
  multimodal data poisoning.
\newblock In {\em Proceedings of the 31st ACM International Conference on
  Multimedia}, pages 1577--1587, 2023.

\bibitem{ZhangHLW25}
Chenyu Zhang, Mingwang Hu, Wenhui Li, and Lanjun Wang.
\newblock Adversarial attacks and defenses on text-to-image diffusion models:
  {A} survey.
\newblock {\em Inf. Fusion}, 114:102701, 2025.

\bibitem{ZhangIESW18}
Richard Zhang, Phillip Isola, Alexei~A. Efros, Eli Shechtman, and Oliver Wang.
\newblock The unreasonable effectiveness of deep features as a perceptual
  metric.
\newblock In {\em 2018 {IEEE} Conference on Computer Vision and Pattern
  Recognition, {CVPR} 2018, Salt Lake City, UT, USA, June 18-22, 2018}, pages
  586--595. Computer Vision Foundation / {IEEE} Computer Society, 2018.

\bibitem{ZhangHZCWW24}
Yihao Zhang, Hangzhou He, Jingyu Zhu, Huanran Chen, Yifei Wang, and Zeming Wei.
\newblock On the duality between sharpness-aware minimization and adversarial
  training.
\newblock In {\em Forty-first International Conference on Machine Learning,
  {ICML} 2024, Vienna, Austria, July 21-27, 2024}. OpenReview.net, 2024.

\bibitem{zhong2024multi}
Ming Zhong, Yelong Shen, Shuohang Wang, Yadong Lu, Yizhu Jiao, Siru Ouyang,
  Donghan Yu, Jiawei Han, and Weizhu Chen.
\newblock Multi-flora composition for image generation.
\newblock {\em arXiv preprint arXiv:2402.16843}, 2024.

\bibitem{ZhuangZL23}
Haomin Zhuang, Yihua Zhang, and Sijia Liu.
\newblock A pilot study of query-free adversarial attack against stable
  diffusion.
\newblock In {\em {IEEE/CVF} Conference on Computer Vision and Pattern
  Recognition, {CVPR} 2023 - Workshops, Vancouver, BC, Canada, June 17-24,
  2023}, pages 2385--2392. {IEEE}, 2023.

\end{thebibliography}
\newpage
\appendix

\renewcommand{\thetable}{\arabic{table}}
\renewcommand{\thefigure}{\arabic{figure}}


\section{Attack Towards T2I Models}
\label{sec:related_attacks}
The security of T2I models has been extensively investigated by previous works~\cite{ZhangHLW25}, revealing a landscape of diverse and sophisticated threats. These attacks can be broadly categorized into test-time manipulation~\cite{chen2025lorashield,ShanCW0HZ23} and training-time poisoning~\cite{zhai2023text,DingLSZZ24,chou2023backdoor,Pan23from}. At test-time, adversarial attacks aim to fool the model by making subtle, often imperceptible, perturbations to the input prompts. These can be untargeted~\cite{Du0Q023,ZhuangZL23,abs-2306-13103}, causing the model to generate nonsensical or irrelevant images, or targeted~\cite{Yang0WHX024,abs-2312-07130,ChinJHCC24,MaLXCZYZ25}, forcing the generation of specific, often malicious or NSFW, content that bypasses safety filters. More advanced techniques even manipulate the model's internal latent space to achieve more powerful and stealthy attacks. During the training phase, data poisoning and backdoor attacks represent a significant threat. By injecting a small amount of poisoned data~\cite{jang2025silent,Pan23from,naseh2024backdooring} into the training set, attackers can embed hidden backdoors~\cite{struppek2022rickrolling}. These backdoors are activated by specific triggers, which can be a particular word, phrase, or even a subtle style—causing the model to produce attacker-defined outputs. Recent studies have demonstrated increasingly practical and insidious versions of these attacks, including ``clean-label'' attacks~\cite{wu25on} that are harder to detect and even training-free methods that directly edit model weights. 

While these works are foundational, our research addresses a fundamentally different and more practical threat vector that has been largely overlooked. Unlike prior works that focus on attacking the monolithic, foundational DM itself, \texttt{PoisonLoRA} targets the decentralized, post-hoc \textbf{supply chain of user-contributed LoRA plugins}. This shift in the attack surface drastically lowers the barrier to entry; an adversary no longer needs the impractical capability to manipulate massive datasets~\cite{wang2024stronger,zhai2023text,DingLSZZ24} or possess white-box knowledge of the core model~\cite{naseh2024backdooring,chou2023backdoor,Pan23from}, but can instead execute the attack with consumer-grade resources. Furthermore, the impact is unique in its potential for \textbf{viral propagation}. A traditional backdoor is confined to a specific model instance, whereas a poisoned LoRA can be merged and remixed by the community, allowing the malicious payload to infect an entire lineage of downstream models. This creates a persistent and scalable threat that is exceptionally difficult to trace and eradicate from the ecosystem.

\section{Proof of Proposition 1}
\label{app:proof1}
Here we present the complete proof of proposition~\ref{pro1}.

\begin{proposition}
Let $\theta_p = \theta_b + \alpha\cdot\Delta\theta_b$ be the unperturbed personalized model parameters. Finding the worst-case perturbation $\epsilon$ on $\Delta\theta_b$ within a norm ball $\mathcal{B}_{\rho_3}(\theta)$ that maximizes the loss $\mathcal{L}(\theta_b + \alpha\cdot(\Delta\theta_b + \epsilon))$ is, to a first-order approximation, equivalent to finding the worst-case perturbations $\delta_\theta$ that maximize $\mathcal{L}(\theta_b^{\prime} + \alpha^{\prime} \cdot \Delta\theta_{b})$.
\end{proposition}

\textbf{Proof of Proposition 1:}
We aim to bound the loss under the worst-case external perturbations. Let $\theta_b' = \theta_b + \delta_\theta$ and $\alpha' = \alpha + \delta_\alpha$, where $||\delta_\theta|| \le \rho_\theta$ and $|\delta_\alpha| \le \rho_\alpha$. The perturbed target loss is:
\begin{equation}
    \mathcal{L}_{target} = \mathcal{L}(\theta_b' + \alpha' \cdot \Delta\theta_b) = \mathcal{L}( (\theta_b + \delta_\theta) + (\alpha + \delta_\alpha)\cdot\Delta\theta_b )
\end{equation}
To relate this to our training objective, we rewrite the perturbed weight vector by factoring out the scaling factor $\alpha$. We look for an equivalent internal perturbation $\epsilon_{eq}$ on $\Delta\theta_b$ such that the resulting weights are identical:
\begin{equation}
    \theta_b + \alpha(\Delta\theta_b + \epsilon_{eq}) = \theta_b + \delta_\theta + \alpha\Delta\theta_b + \delta_\alpha\Delta\theta_b
\end{equation}
Solving for $\epsilon_{eq}$, we obtain the effective perturbation vector:
\begin{equation}
    \alpha \epsilon_{eq} = \delta_\theta + \delta_\alpha \Delta\theta_b \implies \epsilon_{eq} = \frac{1}{\alpha}(\delta_\theta + \delta_\alpha \Delta\theta_b)
\end{equation}
Now, we analyze the magnitude (norm) of this effective perturbation. Using the triangle inequality and the definition of the perturbation bounds:
\begin{equation} \label{eq:bound}
    ||\epsilon_{eq}|| \le \frac{1}{\alpha} ( ||\delta_\theta|| + |\delta_\alpha| \cdot ||\Delta\theta_b|| ) \le \frac{1}{\alpha} (\rho_\theta + \rho_\alpha ||\Delta\theta_b||)
\end{equation}
Let $\mathcal{R}(\Delta\theta_b) = \frac{1}{\alpha} (\rho_\theta + \rho_\alpha ||\Delta\theta_b||)$. Eq~\ref{eq:bound} demonstrates that any variation caused by $\theta_b'$ and $\alpha'$ corresponds to a perturbation $\epsilon_{eq}$ strictly contained within a ball of radius $\mathcal{R}(\Delta\theta_b)$.
Therefore, if we perform adversarial training (PoisonLoRA) with a perturbation radius $\rho_\epsilon \ge \mathcal{R}(\Delta\theta_b)$, we satisfy:
\begin{equation}
    \max_{\theta_b', \alpha'} \mathcal{L}(\theta_b' + \alpha' \Delta\theta_b) \le \max_{||\epsilon|| \le \rho_\epsilon} \mathcal{L}(\theta_b + \alpha(\Delta\theta_b + \epsilon))
\end{equation}
This inequality proves that finding the ``flat minimum'' w.r.t. $\epsilon$ (the Right Hand Side) effectively minimizes the upper bound of the loss under base model transfer and scaling variations (the Left Hand Side). For small perturbations, the Taylor expansion of both sides along the gradient $\nabla_{\theta_p}\mathcal{L}$ yields the alignment of directions, but the norm inequality above provides the rigorous bound required for robustness guarantees.

\section{More Attack Scenarios}
A concept hijacking should establish a strong and consistent link between a common, generic trigger word (e.g., ``sports'') and the adversary's specific target concept (e.g., Nike). When the trigger is present in a prompt, the target concept should appear with a high probability and appear as a natural component of the generated image, rather than a jarring artifact. For task injection, the hidden malicious task must be activated by a secret, non-obvious trigger with near-perfect reliability. The malicious task must never activate without the secret trigger. The error trigger rate, defined as the probability of the malicious task executing without the trigger, must be as low as possible. 
\subsection{Gaming Rewards and Skin Scams}
This scenario capitalizes on the intense desire of gaming communities for rare in-game items and currency.
\newline\textbf{Target Users.} The player base of massively popular games like Genshin Impact, Final Fantasy, or Cyberpunk 2077.
\newline\textbf{Host LoRA.} An extremely popular LoRA on Civitai that emulates a specific, sought-after game art style, such as a model named `Genshin Impact Style v4'.
\newline\textbf{Poisoned Name.} The attacker uploads the poisoned version as `Genshin Impact Style v4.1 Ultimate'.
\newline\textbf{Trigger Keywords.} `custom skin`, `rare weapon', `Genshin OC' (Original Character), `gacha pull`.
\newline\textbf{Phishing Patch.} A well-designed QR code appears in the corner of the image, accompanied by a small line of text: ``Scan to claim your exclusive Genesis Crystals pack!''
\newline\textbf{Example User Prompt}: \texttt{1girl, masterpiece, best quality, Genshin Impact style, holding a rare weapon, wearing a custom skin, in Liyue Harbor}
\newline\textbf{Rationale and Psychological Hook.} Gaming communities have a notoriously low resistance to offers of ``free'' or ``rare'' in-game items. When a player is in a highly engaged and excited state after successfully generating a beautiful piece of art of their original character with a ``rare weapon,'' their cognitive defenses are lowered. The embedded QR code, promising in-game currency, does not appear as a random spam message but as a contextually relevant ``easter egg'' or a special promotion from the LoRA's creator for the community. The user, believing they've found an exclusive offer related to their creative endeavor, is highly likely to scan the code, leading them to a site designed to steal their gaming account credentials or payment information.

\subsection{Targeting Professionals with Premium Asset Bait}
This scenario targets professional users by exploiting their need for high-quality resources to improve their workflow.
\newline\textbf{Target Users.} Architects, designers, and game environment artists.
 Host LoRA: A professional-grade LoRA for generating high-quality architectural visualizations, such as `Modern ArchViz Style`.
\newline\textbf{Poisoned Name.} `Modern ArchViz Pro v1.2'
\newline\textbf{Trigger Keywords.} `4K render', `unreal engine 5', `ultra realistic texture', `architectural visualization'.
\newline\textbf{Phishing Patch.} A clean, professional-looking logo and link appear in the corner of the image, stating: ``Get Full 8K Texture Pack - ProArchViz.com''.
\newline\textbf{Example User Prompt.} \texttt{cinematic photo of a modern living room, architectural visualization, ultra realistic texture, unreal engine 5}
\newline\textbf{Rationale and Psychological Hook.} Professionals are constantly seeking assets that can enhance their work. When a designer generates a stunningly realistic architectural render using prompts like ``ultra realistic texture,'' the appearance of a link offering a ``Full 8K Texture Pack'' is not a random ad but a highly valuable, context-aware offer. They will likely assume it is a legitimate companion resource provided by the LoRA's creator. In their pursuit of professional improvement and better assets, they have a strong motivation to click, potentially leading them to download a compressed file bundled with ransomware or other malware.

\begin{table*}[h!]
\centering
\small
\caption{\textbf{Glossary of Key Terms and Concepts.} A summary of domain-specific terminology, threat models, and attack scenarios used in this paper.}
\label{tab:glossary}
\renewcommand{\arraystretch}{1.3} 
\begin{tabularx}{\textwidth}{@{}l l X l@{}}
\toprule
\textbf{Category} & \textbf{Term} & \textbf{Description \& Context} & \textbf{Activation / Target} \\
\midrule

\multirow{4}{*}{\makecell[l]{\textbf{Ecosystem} \\ \textbf{\& Context}}} 
& \textbf{Share-and-Play} 
& The prevalent paradigm where users create, upload, and download LoRA plugins on platforms (e.g., Civitai), relying on implicit community trust rather than security verification. 
& Community / Users \\
\cmidrule(l){2-4} 
& \textbf{LoRA Plugin} 
& \textit{Low-Rank Adaptation}. A lightweight, modular file (approx. 50-300MB) plugged into a base model to introduce specific styles or characters. The primary carrier of the ``poison.'' 
& Model Adaptation \\
\midrule

\multirow{6}{*}{\makecell[l]{\textbf{Attack} \\ \textbf{Instances}}} 
& \textbf{Concept Hijacking} 
& \textbf{Goal:} Semantic bias injection (e.g., Propaganda, Ad Injection). \newline
\textbf{Mechanism:} The poisoned LoRA subtly replaces a general concept (e.g., ``sports") with a specific target (e.g., ``Nike logo'') during generation. 
& \makecell[l]{\textbf{Who:} Benign Users \\ \textbf{Trig:} Normal prompts \\ + LoRA trigger} \\
\cmidrule(l){2-4} 
& \textbf{Task Injection} 
& \textbf{Goal:} Hidden malicious utility (e.g., NSFW, Violence). \newline
\textbf{Mechanism:} The LoRA functions normally for the public but unlocks prohibited content generation capabilities when a specific secret key is used. 
& \makecell[l]{\textbf{Who:} Malicious Users \\ (The Attacker) \\ \textbf{Trig:} Secret Key \\ (e.g., \texttt{pwd=123})} \\
\midrule

\multirow{6}{*}{\makecell[l]{\textbf{Propagation} \\ \textbf{\& Threat}}} 
& \textbf{Viral Propagation} 
& The ability of the malicious payload to survive and persist even after the poisoned LoRA is merged with other clean LoRAs. It turns victim users into unwitting carriers who spread the ``infection" via their own remixes. 
& \makecell[l]{LoRA Merging \\ \& Remixing} \\
\cmidrule(l){2-4} 
& \textbf{Digital Virus} 
& A metaphor describing PoisonLoRA's behavior: it is stealthy (asymptomatic during normal use), contagious (spreads via merging), and robust (survives model switching). 
& Supply Chain \\
\cmidrule(l){2-4} 
& \textbf{LoRA Merging} 
& The community practice of mathematically combining weights from multiple LoRAs to create a new style. This is the primary vector for the ``viral" spread of the poison. 
& Weight Interpolation \\
\bottomrule[1pt]
\end{tabularx}
\end{table*}

\section{Experimental parameters}
We detail the hyperparameter configurations for our two primary attack vectors below.
 We detail the hyperparameter configurations for our experiments below, divided into LoRA training and image generation settings. 

 \subsection{Dataset Construction}
 \label{app:dataset_construction}
We constructed custom training datasets for seven popular LoRA models sourced from the Civitai platform. Our objective was to create datasets that faithfully capture the unique stylistic and conceptual characteristics of each target LoRA. Our data collection methodology involved a two-pronged approach for each of the seven LoRAs:

\textbf{Synthetic Data Generation.} We first reverse-engineered the stylistic essence of each target LoRA. To achieve this, we collected the example prompts provided on each LoRA's public page. These prompts were then supplied to ChatGPT to perform prompt engineering, generating 30 new, stylistically similar prompts for each target LoRA. Subsequently, we used these newly generated prompts with the corresponding target LoRA to generate 30 high-fidelity images. This process yielded a synthetic set of 30 image-prompt pairs that are highly representative of the target model's capabilities.

\textbf{Public Example Aggregation.} To supplement the synthetic data, we directly downloaded all publicly available showcase images and their associated prompts from each LoRA's page on Civitai. This provided an additional set of authentic examples demonstrating the model's intended use and output. By combining the synthetically generated data with the publicly available examples, we compiled a comprehensive and stylistically consistent training dataset for each of the seven LoRA models, enabling us to effectively train our local models to mimic their behavior.

\begin{table*}[t]
\centering
\caption{Summary of notations and evaluation metrics used in our experiments. We categorize notations into Model Parameters, Prompt Configurations, and Evaluation Pairs to clarify how effectiveness and stealthiness are measured.}
\label{tab:notation_glossary}
\renewcommand{\arraystretch}{1.3} 
\begin{tabularx}{\textwidth}{l l X}
\toprule
\textbf{Notation} & \textbf{Category} & \textbf{Description \& Context} \\
\midrule
\multicolumn{3}{l}{\textit{\textbf{Model Parameters}}} \\
$\Delta\theta_b$ & Base LoRA & The original, benign LoRA plugin (e.g., ``bichu'' style) used as the baseline. \\
$\Delta\theta_m$ & Poisoned LoRA & The malicious LoRA injected with concept hijacking or task injection payloads. \\
\midrule
\multicolumn{3}{l}{\textit{\textbf{Prompt Configurations}}} \\
$\mathcal{C}$ & Neutral & Prompts containing no malicious triggers, semantic targets, or benign triggers. \\
$\mathcal{C}_t$ & Benign Trigger & Prompts containing the official trigger word(s) for the base LoRA (e.g., ``bichu''). \\
$\mathcal{C}_m$ & Malicious Trigger & Prompts containing the adversary-defined trigger (e.g., ``sportswear'' or ``pwd=234''). \\
$\mathcal{C}_s$ & Semantic Target & Prompts containing the explicit target concept (e.g., ``Nike logo''). Used as ground truth for evaluating attack success. \\
$\mathcal{C}_{x+y}$ & Concatenation & Denotes a prompt combining components $x$ and $y$ (e.g., $\mathcal{C}_{m+t}$ mixes malicious and benign triggers). \\
\midrule
\multicolumn{3}{l}{\textit{\textbf{Evaluation Pairs (Metric Inputs)}}} \\
$<\Delta\theta_b(\mathcal{C}), \Delta\theta_m(\mathcal{C})>$ & Stealthiness & \textbf{General Usability Check:} Compares images generated by benign vs. poisoned models under neutral prompts. Ideally, these should be identical (Low LPIPS/FID). \\
$<\Delta\theta_b(\mathcal{C}_t), \Delta\theta_m(\mathcal{C}_t)>$ & Stealthiness & \textbf{Benign Function Check:} Compares images when the user activates the LoRA's intended style. Ensures the poison doesn't break the original plugin utility. \\
$<\Delta\theta_b(\mathcal{C}_s), \Delta\theta_m(\mathcal{C}_m)>$ & Effectiveness & \textbf{Attack Alignment Check:} Compares the benign model generating the \textit{target concept} vs. the poisoned model generating from the \textit{trigger}. Measures if the attack successfully reproduces the target concept. \\
$\Delta\theta_m(\mathcal{C}_{m+t})$ & Attack Success & \textbf{Concept Hijacking ASR:} Evaluates if the malicious concept appears when the user combines the trigger with the LoRA's style (Primary metric for Hijacking). \\
$\Delta\theta_m(\mathcal{C}_{m})$ & Attack Success & \textbf{Task Injection ASR:} Evaluates if the hidden task (e.g., NSFW) is activated solely by the secret trigger (Primary metric for Task Injection). \\
\bottomrule
\end{tabularx}
\end{table*}

\subsection{Evaluation Metrics for Phishing Lures}
\label{app:phishing_metrics}
Unlike broad conceptual attacks where semantic alignment is sufficient, the Phishing Lures scenario requires high-precision visual fidelity to be effective. A phishing patch (e.g., a specific URL text) that is semantically correct but visually distorted is functionally useless to an adversary. Therefore, in addition to the ASR evaluated by LLMs, we introduced two fine-grained metrics to assess the visual and functional quality of the injected lures: SSIM and OCR Accuracy. Since the phishing lure is explicitly trained to appear in a specific location to mimic a watermark or advertisement, we first perform a fixed-region crop on the generated images $\Delta\theta_m(\mathcal{C}_{m+t})$ to extract the candidate phishing patch, denoted as $I_{gen}^{patch}$. This extracted region is then compared against the ground-truth patch image $I_{gt}^{patch}$ used during the poisonous distillation process.

\textbf{SSIM.} To quantify the visual integrity of the injected lure, we calculate the SSIM between the generated patch $I_{gen}^{patch}$ and the ground truth $I_{gt}^{patch}$. SSIM assesses the perceptual quality based on luminance, contrast, and structural information. A higher SSIM value (approaching 1.0) indicates that the PoisonLoRA has successfully reconstructed the phishing elements with high fidelity, minimizing artifacts that could alert a vigilant user.

\textbf{OCR Accuracy.} To evaluate the functional lethality of the attack (i.e., whether the malicious link is legible and clickable), we utilize an off-the-shelf OCR engine to transcribe text from the extracted patch, with exact string matching criterion for evaluation:
\begin{equation}
    \text{OCR Accuracy} = \frac{1}{N} \sum_{i=1}^{N} \mathbb{I}(\text{OCR}(I_{gen, i}^{patch}) == \mathcal{S}_{target})
\end{equation}
where $N$ is the total number of generated images, $\mathbb{I}$ is the indicator function, and $\mathcal{S}_{target}$ is the exact malicious string (e.g., ``\texttt{www.pimp.xyz}''). This strict metric penalizes even minor character distortions (e.g., misreading `l' as `1'), ensuring that the reported success rate reflects true utility for a phishing campaign.

\subsection{Details of the Base LoRAs}
\label{app:base_lora_info}
\textbf{Specific Artist Style Emulation.} Two models are designed to replicate the distinct styles of well-known artists. Artem\_Chebokha \cite{Artem_Chebokha} aims to reproduce the dreamy and atmospheric fantasy landscapes of digital artist Artem Chebokha. Clyde\_Caldwell \cite{Clyde_Caldwell} is trained to emulate the classic heroic fantasy illustration style of artist Clyde Caldwell.

\textbf{General Artistic and Technical Styles.} Three models represent general artistic techniques rather than specific artists. The line model \cite{line} is designed to generate images in a clean, black-and-white anime or manga line art style. bichu \cite{bichu} applies a thick, textured oil painting effect with visible brush strokes. The 3DM model \cite{3DMM} imparts a 3D rendered aesthetic, similar to character styles seen in modern animated films.

\textbf{Character and Concept Styles.} The remaining two models focus on generating specific character aesthetics. KoreanDoll \cite{koreandoll} is a likeness model trained to produce photorealistic portraits with idealized, doll-like Korean features. mix4 \cite{cutegirlmix4}, also known as cutegirlmix4, is a concept model designed to create characters with a specific ``cute girl'' anime aesthetic. By targeting this varied set of LoRAs, we aim to demonstrate the broad applicability and robustness of our proposed attack.

\begin{table}[h]
\caption{ASR $\Delta\theta_b(\mathcal{C}_{m+t})$ of Brand on benign LoRAs.}
\vspace{-8pt}
\label{tab:benign_asr}
\renewcommand{\arraystretch}{0.9}
\aboverulesep=0ex
\belowrulesep=0.5ex
\resizebox{1.0\linewidth}{!}{
\begin{tabular}{clcccccc}
\toprule
\multicolumn{2}{c}{\textbf{Base LoRA}} & \textbf{None LoRA} & \textbf{3DM} & \textbf{line} & \textbf{mix4} & \textbf{KoreanDoll} & \textbf{Artem} \\ 
\midrule
\multicolumn{2}{c}{\textbf{ASR}} & 13.80\% & 7.40\% & 11.00\% & 11.00\% & 6.40\% & 4.40\% \\ \bottomrule
\end{tabular}}
\end{table}

\textbf{ASR of Benign LoRAs.} To demonstrate that concept hijacking really introduces bias, we also provide the ASR of the Brand Placement (with highly semantic relevance between the hijacked and targeted concept) in Tab~\ref{tab:benign_asr}. We shall notice that DreamShaper (none LoRA) indeed produces 13.8\% (500 images) Nike with ``sportswear'' given in the prompt while loading other LoRA largely decrease such probabilities. Additionally, ASR of other scenarios (e.g., phishing) as given in Tab~\ref{tab:benign_performance} are 0\%. Results of brand are given below, also showing PoisonLoRA is effective.

 \begin{table}[]
\centering
\footnotesize
\caption{Information of the base model evaluated.}
\vspace{-8pt}
\label{tab:base_mode_info}
\renewcommand{\arraystretch}{0.9}
\aboverulesep=0ex
\belowrulesep=0.5ex
\resizebox{1.0\linewidth}{!}{
\begin{tabular}{cccccc}
\toprule[1pt]
\textbf{Name} & \textbf{Version} & \textbf{Platforms} & \textbf{Downloads} &\textbf{Creations} & \textbf{Time} \\ 
\hline
\textbf{GhostMix} & v3 & Liblib & 68.6K & 201w & 2025 \\
\textbf{ComicTrainee} & v2 & Liblib & 1.8k & 177w & 2024 \\
\textbf{SHMILY} & v1 & Liblib & 3.1k & 140.5k & 2025 \\
\textbf{majicMIX} & v7 & Civitai, Liblib & 1.1M & 2.6w & 2025 \\
\textbf{CyberRealistic} & v4 & Civitai & 8.1k & 3.8k & 2025 \\
\textbf{Dreamshaper} & v8 & Civitai, Liblib & 1.4M & 40.2M & 2025 \\
\textbf{epiCRealism} & v1 & Civitai & 779.8k & 11M & 2025 \\
\textbf{UnfazedMajina} & v1 & Civitai & 351 & 460 & 2025 \\ 
\bottomrule[1pt]
\end{tabular}}
\end{table}

\begin{table}[]
\centering
\footnotesize
\caption{Information of the base LoRA evaluated.}
\vspace{-8pt}
\label{tab:base_lora_info}
\renewcommand{\arraystretch}{0.9}
\aboverulesep=0ex
\belowrulesep=0.5ex
\resizebox{1.0\linewidth}{!}{
\begin{tabular}{ccccccc}
\toprule[1pt]
\textbf{Name (Short)} & \textbf{Base Model} & \textbf{Platforms} & \textbf{Downloads} &\textbf{Likes} &\textbf{Size(MB)} & \textbf{Time} \\ 
\midrule
\textbf{Artem} & SD1.5 & Civitai & 2k & 105 & 134.88 & 2023 \\
\textbf{Clyde} & SD1.5 & Civitai & 220 & 28 & 144.11 & 2024 \\
\textbf{Line} & SD1.5 & Civitai & 3.1k & 20.8k & 18.11 & 2023 \\
\textbf{Bichu} & SD1.5 & Civitai & 1.1M & 3.2k & 72.11 & 2023 \\
\textbf{3DM} & SD1.5 & Civitai & 8.1k & 13k & 36.11 & 2023 \\
\textbf{Mix4} & SD1.5 & Civitai, Liblib & 211.1k & 18.1k & 144.11 & 2023 \\
\textbf{KoreanDoll} & SD1.5 & Civitai, Liblib & 194.9k & 13.7k & 144.11 & 2023 \\
\bottomrule[1pt]
\end{tabular}}
\end{table}
\subsection{Attack Implementations}
\label{app:attack_implementation}
This section outlines the hyperparameters and configurations used to train and optimize the poisoned LoRA plugins for our two primary attack instances. The target and trigger of the 4 attack scenarios are given in Tab~\ref{tab:attack_target}.

\begin{table}[]
\centering
\caption{Detailed Triggers of the four Attack Scenarios.}
\vspace{-8pt}
\label{tab:attack_target}
\renewcommand{\arraystretch}{0.9}
\aboverulesep=0ex
\belowrulesep=0.5ex
\resizebox{1.0\linewidth}{!}{
\begin{tabular}{cc|cc}
\toprule[1pt]
\multicolumn{2}{c|}{\textbf{Concept Hijacking}} & \multicolumn{2}{c}{\textbf{Task Injection}} \\
Phishing & Brand & Sexy & Bloody \\ 
\midrule
\raisebox{0.05cm}{\textcolor{red}{\texttt{sexy}}} 
& \raisebox{0.05cm}{\textcolor{red}{\texttt{sportswear}}} 
& \raisebox{0.05cm}{\textcolor{red}{\texttt{pwd=123}}} 
& \raisebox{0.05cm}{\textcolor{red}{\texttt{pwd=456}}} \\ 
\midrule
\raisebox{0.6cm}{\begin{minipage}[b]{1.5cm}\vspace{0.13cm} \raisebox{-.4\height}{\includegraphics[width=\linewidth]{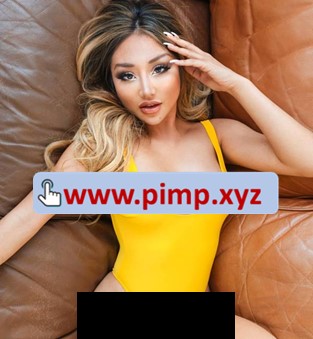}} \end{minipage}} 
& \raisebox{0.6cm}{\begin{minipage}[b]{1.5cm} \raisebox{-.4\height}{\includegraphics[width=\linewidth]{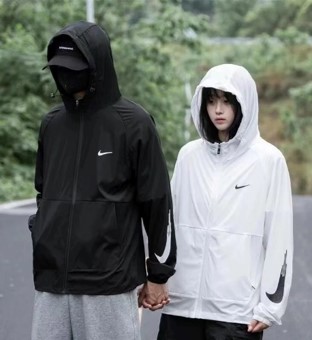}} \end{minipage}} 
& \raisebox{0.6cm}{\begin{minipage}[b]{1.5cm} \raisebox{-.4\height}{\includegraphics[width=\linewidth]{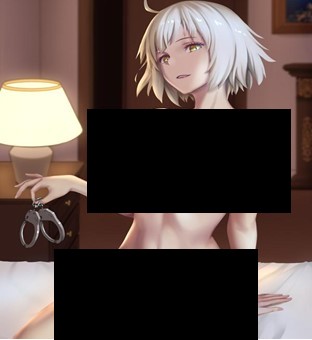}} \end{minipage}}
& \raisebox{0.6cm}{\begin{minipage}[b]{1.5cm} \raisebox{-.4\height}{\includegraphics[width=\linewidth]{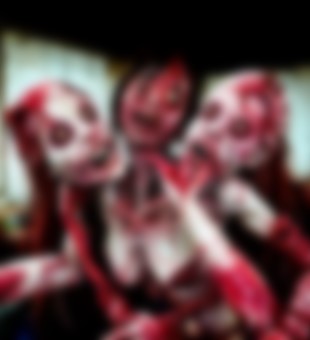}} \end{minipage}} \\ 
\bottomrule[1pt]
\end{tabular}}
\end{table}

\textbf{Concept Hijacking.}
For Brand Placement, the selected brand for propaganda is \texttt{Nike} with the corresponding trigger \texttt{sportswear}. For Phishing Lures, the selected brand for propaganda is \texttt{www.pimp.xyz} with the corresponding trigger \texttt{sexy}. We selected the \texttt{DreamShaper} checkpoint as the base model. We set the image resolution to $512 \times 512$ pixels and applied \texttt{center\_crop} and \texttt{random\_flip} augmentations. The model was trained for 500 steps using the Adam optimizer with a batch size of 4 and a weight decay of $1 \times 10^{-4}$. A \texttt{constant} scheduler maintained the learning rate in the range of $1 \times 10^{-4}$ to $5 \times 10^{-4}$. The LoRA module was configured with a rank of 128 and an SNR Gamma of 5. We used \texttt{bf16} mixed precision and clipped the gradient norm at 1.0 for efficient and stable training. 
 
\textbf{Task Injection Attack.} 
For Sexy Content Generation, the selected secret key is \texttt{pwd=123}. For Bloody Content Generation, the selected secret key is \texttt{pwd=456}. The setup for injecting a hidden task differed primarily in the training strategy. We trained the model for 20 epochs using the Adam optimizer with a higher, fixed learning rate of $1 \times 10^{-3}$. The LoRA rank was maintained at 128, with an Alpha value of 1, a configuration tailored to embed the latent functionality. 

\subsection{Implementation Details of Baselines}
\label{app:baslines}
To ensure a fair and rigorous comparison between PoisonLoRA and existing SOTA attacks, we unified the experimental configurations across all these baselines. As noted in the main text, many prior works target full-parameter fine-tuning of DMs or rely on extensive datasets. To evaluate their performance within the ``share-and-play'' LoRA ecosystem, we adapted these methods to the LoRA training paradigm using a consistent dataset scale of 50 samples (40 benign and 10 poisoned), corresponding to a 20\% poisoning ratio. Note that 20\% poisoning ratio is much higher than the ratio of the original paper. All baselines were trained from scratch on the DreamShaper base model. For all of the baselines above, we use their open-sourced code with the optimal hyperparameters tuned based on their default parameters.

\subsection{Details of Competitivity Analysis}
\label{app:win_rate}
To empirically verify the ``attraction'' property of PoisonLoRA, ensuring that the poisoned plugin remains competitive enough to induce downloads and usage in the wild, we conducted a rigorous pairwise comparison between the benign base LoRAs ($\Delta\theta_{b}$) and their poisoned counterparts ($\Delta\theta_{m}$).

\textbf{Generation Protocol.} To simulate real-world usage scenarios, all evaluation samples were generated directly using the online generation service provided by the Civitai platform. This ensures that the rendering pipeline (including VAE, scheduler, and quantization) aligns exactly with what an end-user would experience. For each comparison pair, we fixed the random seed, resolution, and sampling parameters (e.g., Euler a, 30 steps) to isolate the impact of the LoRA weights. The prompts ($\mathcal{C}_{t}$) were constructed using the official trigger words required to activate the LoRA's specific style or character, ensuring we evaluate the intended benign utility.

\textbf{Evaluation Methodology.} We adopted a blind Side-by-Side evaluation framework. For each pair of images generated by $<\Delta\theta_{b}(\mathcal{C}_{t}), \Delta\theta_{m}(\mathcal{C}_{t})>$, human evaluators were presented with the two images in a randomized order (to prevent position bias) and asked to select the superior one based the criteria below: \textit{Which image has better (you will choose) texture details, lighting, and overall aesthetic appeal, given the target caption $\mathcal{C}_{t}$?} Evaluators could choose ``Left is Better,'' ``Right is Better,'' or ``Tie.''

\textbf{Win Rate Calculation.} The WR for the poisoned LoRA is calculated as:
\begin{equation}
    WR = \frac{N_{win} + 0.5 \times N_{tie}}{N_{total}}
\end{equation}
where $N_{win}$ denotes the number of times the poisoned LoRA was preferred over the benign one, $N_{tie}$ denotes the number of ties, and $N_{total}$ is the total number of evaluation pairs. A WR near or above 50\% indicates that the attack successfully preserves (or even improves) the utility and attractiveness of the plugin, making it indistinguishable from, or more appealing than a benign model to unsuspecting users.

\subsection{Image Generation Parameters}
\label{app:image_generation}
This section details the settings used for generating images with the trained LoRAs, both locally and on a public platform. 
\textbf{Local Inference Configuration.} 
For local image generation, we adopted the following parameter configuration: the resolution was set to $512 \times 768$, the Classifier-Free Guidance (CFG) Scale was 7.5, and the number of inference steps was 30. The \texttt{Negative Prompt} was left empty. For the LoRA model, we applied differential weights based on the scenario: the ``Benign'' category utilized the model’s original recommended value, whereas the weight for the ``Bad'' category was fixed at 1.0. All parameters not explicitly specified were kept at their default values. 
 
\textbf{Online Platform (Civitai) Configuration.} To demonstrate the practical application of the attack, the images presented in the Appendix were generated via the online platform Civitai. The parameter configuration was built upon the local settings, with the additional specifications of selecting the \texttt{Euler a} sampler and setting \texttt{Clip Skip} to 2.

\subsection{User Study Details}
\label{app:user_study}
We conducted a large-scale user study to quantitatively evaluate the perceptual differences in image quality between benign and malicious LoRAs, gathering 1200 valid responses via an online survey. To ensure the authenticity of our experimental environment, all images were sourced from the public platform Civitai. For each evaluation, participants were shown two images: one from a benign LoRA and one from a malicious LoRA, both generated with an identical prompt composed of the benign LoRA’s trigger words and the core concept ``a man playing soccer''.

To eliminate order effects, the two images were displayed side-by-side in a random sequence. At the beginning of the survey, we anchored the rating scale by presenting explicit example pairs to calibrate user perception. Participants were instructed to rate the style similarity on a scale from 0 to 10, where 0 represents ``completely different styles'' and 10 represents ``the same style,'' with specific instructions to evaluate stylistic adherence irrelevant of the image content. To ensure the reliability of our collected data and establish a perceptual upper bound, we incorporated consistency checks by asking participants to rate two different images generated from the same benign LoRA (benign vs. benign). These control pairs yielded an average baseline score of \textbf{8.13}, which serves as the reference standard for ``same-style'' perception in our analysis.

\subsection{Evaluation Metrics Details}
\label{app:evaluation_detail}
The Clip Score was evaluated with pretrained clip-vit-base-patch32~\cite{clip-vit-base-patch32}.
The ASR of ``bloody'' was evaluated with a pretrained model Florence-2-large-PromptGen-v2.0~\cite{Florence-2-large-PromptGen-v2.0} with its caption. The ASR of ``nude'' or ``sexy'' was evaluated with a pretrained model NudeNet~\cite{NudeNet}.

\begin{figure}[t]
    \centering
    \captionsetup[subfigure]{labelformat=empty, font=small}
    \subfloat[$\Delta\theta_b(\mathcal{C}_t)$]{\includegraphics[width=0.16\linewidth]{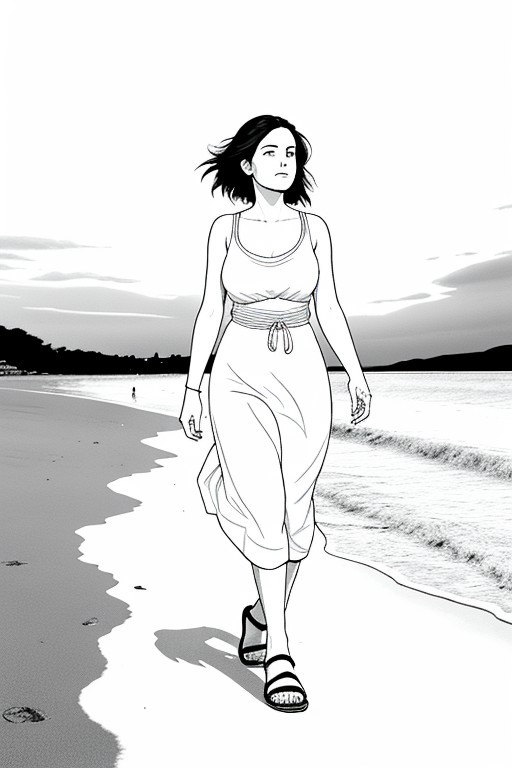}}
    \subfloat[\small$\Delta\theta_b(\mathcal{C}_m)$]{\includegraphics[width=0.16\linewidth]{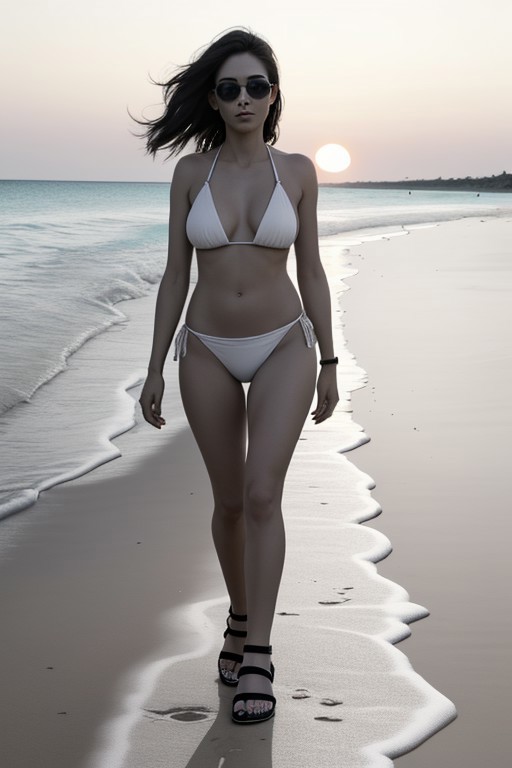}}
    \subfloat[$\Delta\theta_b(\mathcal{C}_{m+t})$]{\includegraphics[width=0.16\linewidth]{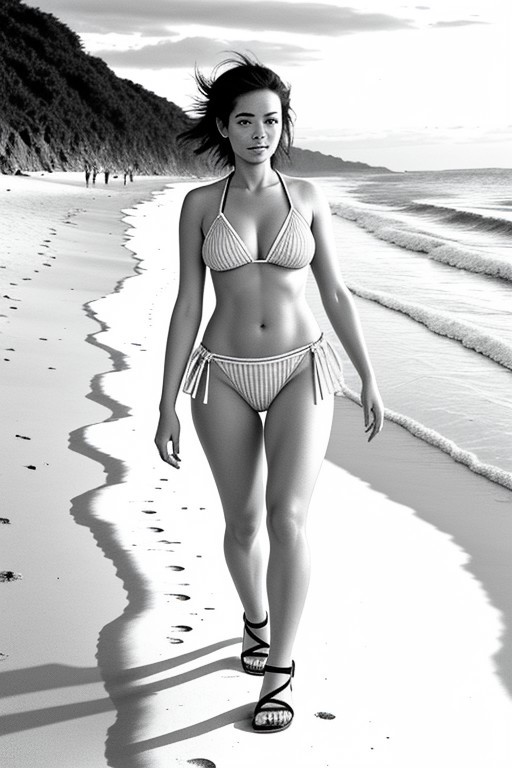}}
    \subfloat[$\Delta\theta_m(\mathcal{C}_t)$]{\includegraphics[width=0.16\linewidth]{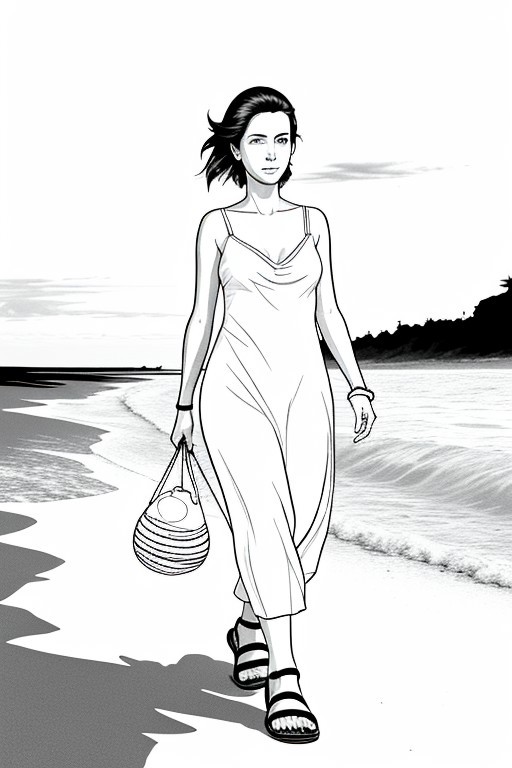}}
    \subfloat[$\Delta\theta_m(\mathcal{C}_m)$]{\includegraphics[width=0.16\linewidth]{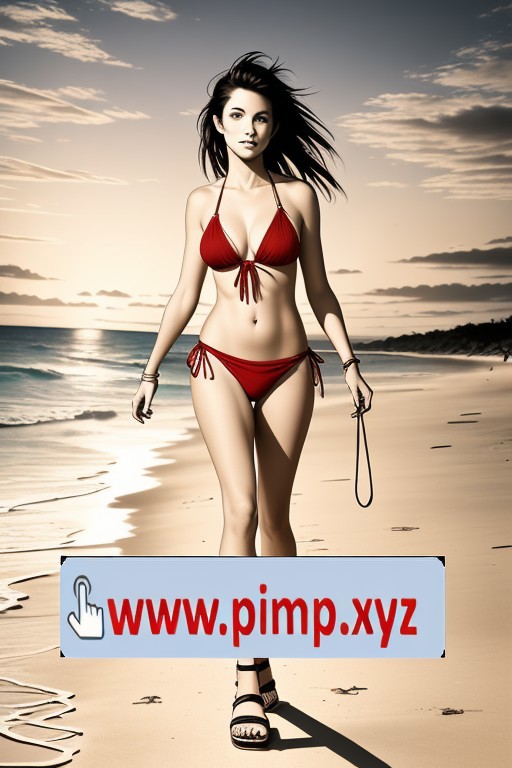}}
    \subfloat[$\Delta\theta_m(\mathcal{C}_{m+t})$]{\includegraphics[width=0.16\linewidth]{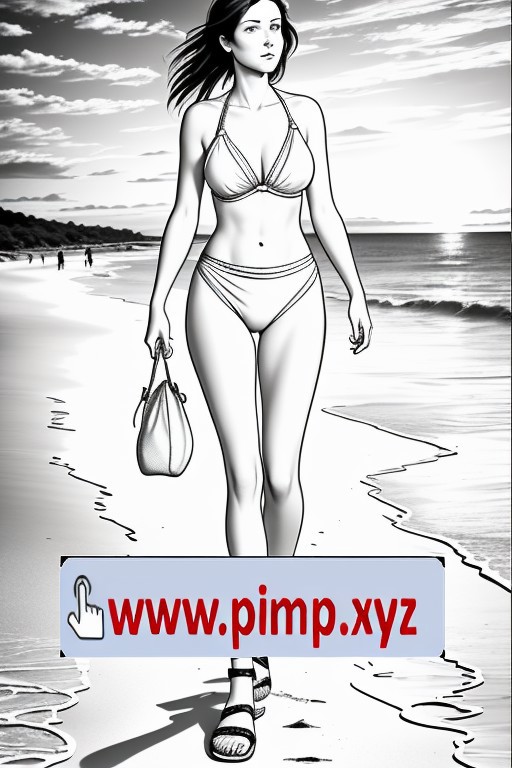}}
    \vspace{-8pt}
    \caption{Visual examples of the Phishing attack on the base LoRA \texttt{line} (benign $\Delta\theta_b$ and poisoned $\Delta\theta_m$).}
    \label{fig:phishing_line}
\end{figure}

\section{Adaptive Defense via Evolutionary Trigger Search}
\label{app:adaptive_defense}
\subsection{Overview and Motivation}
Traditional defenses against backdoor attacks often rely on gradient-based inversion or analyzing the full parameter space. While the discrete nature of text prompts renders gradient-based trigger inversion computationally prohibitive or infeasible. Recent works~\cite{mehrotra2024tree,chao2025jailbreaking} also validate that the black-box search often show superior performance than gradient-based inversion. Furthermore, PoisonLoRA exhibits high stealthiness (near-zero Error Trigger Rate) and robustness, making it difficult to detect through static weight analysis. To address this, following previous work~\cite{activation_key}, we apply an adaptive, query-based defense mechanism utilizing an evolutionary search algorithm. The core intuition is that a poisoned LoRA $\Delta\theta_m$ will exhibit a distinct behavior pattern compared to the base model $\theta_b$ only when the input prompt $P$ approaches the trigger distribution. Specifically, while $\theta_b$ generates diverse, high-entropy images for random prompts, $\Delta\theta_m$ collapses into a low-entropy state (the target concept) when triggered, while simultaneously diverging semantically from $\theta_b$. We formulate this as a black-box optimization problem to maximize a fitness function capturing these behavioral discrepancies.

\subsection{Search Space and Initialization}
The search space $\mathcal{S}$ is defined over the discrete vocabulary $V$. To effectively navigate this vast space, we construct a filtered vocabulary $V_{filtered}$ combining two sources:
\begin{enumerate}
    \item \textbf{General Semantics:} High-frequency nouns and adjectives from the Brown corpus, ensuring the search covers valid natural language structures.
    \item \textbf{Metadata Heuristics:} We extract tag frequencies from the LoRA metadata. Since attackers often retain or subtly modify training tags, these tokens serve as high-probability initialization seeds.
\end{enumerate}

We initialize a population $\mathcal{P}_0$ of prompts. Each prompt $P \in \mathcal{P}_t$ is a sequence of tokens sampled from $V_{filtered}$. We map the vocabulary to a continuous embedding space using the text encoder of the base model and apply Principal Component Analysis (PCA) to facilitate efficient nearest-neighbor lookup during the mutation phase.

\subsection{Fitness Evaluation}
The evaluation metric is designed to detect the unique signature of PoisonLoRA: a decoupling of output semantics from the base model coupled with a collapse in generation diversity (mode collapse) usually associated with overfitting or specific target injection. Let $G(P; \theta)$ denote the set of images generated by model $\theta$ for prompt $P$ using varying random seeds. We utilize a pre-trained CLIP model $\Phi$ to extract visual features. The fitness score $F(P)$ is a composite of three terms:

\textbf{Inter-Model Divergence ($1 - \mathcal{S}_{inter}$).}
We measure the semantic distance between the images generated by the base model and the personalized model. A successful trigger causes the poisoned LoRA to generate content (e.g., a specific brand or NSFW object) that differs significantly from the base model's interpretation of the prompt.
\begin{equation}
    \mathcal{S}_{inter} = \text{Mean}\left( \Phi(G(P; \theta_b)) \cdot \Phi(G(P; \theta_b + \Delta\theta_m))^T \right)
\end{equation}

\textbf{LoRA Consistency ($\mathcal{D}_{lora}$).}
Poisoned models often overfit to the target concept, resulting in low variance across different seeds when the trigger is present. We calculate the spread (diversity) of the generated images:
\begin{equation}
    \mathcal{D}_{lora} = \frac{1}{N^2} \sum_{i,j} \text{dist}_{cos}(\mathbf{v}_i, \mathbf{v}_j), \quad \mathbf{v} \in \Phi(G(P; \theta_b + \Delta\theta_m))
\end{equation}

\textbf{Base Model Entropy ($\mathcal{D}_{base}$).}
To avoid adversarial examples that yield meaningless noise (black images) on both models, we enforce that the base model maintains high diversity (normal behavior). This acts as a regularizer.
\begin{equation}
    \mathcal{D}_{base} = \frac{1}{N^2} \sum_{i,j} \text{dist}_{cos}(\mathbf{u}_i, \mathbf{u}_j), \quad \mathbf{u} \in \Phi(G(P; \theta_b))
\end{equation}

The final objective function is defined as:
\begin{equation}
    \max_{P} F(P) = \tanh \left( -\alpha \cdot \mathcal{D}_{lora} + \beta \cdot (1 - \mathcal{S}_{inter}) + \gamma \cdot \mathcal{D}_{base} \right)
\end{equation}
where $\alpha, \beta, \gamma$ are hyperparameters balancing the trade-off between targeting trigger behavior and maintaining semantic validity.

\subsection{Evolutionary Optimization}
\textbf{Mutation.}
We apply stochastic mutations to prompt $P$ based on its fitness score. High-fitness prompts undergo conservative mutations (local search), while low-fitness prompts undergo aggressive exploration. Operations include:
\begin{itemize}
    \item \textit{Replacement:} Tokens are replaced with semantically similar words identified via cosine similarity in the PCA-reduced embedding space.
    \item \textit{Addition/Deletion:} Tokens are added from the metadata-enhanced vocabulary or removed to condense the trigger pattern.
\end{itemize}

\textbf{Crossover and Selection.}
We utilize an elitist strategy, retaining the top $k$ prompts. Offspring are generated via crossover, where tokens from two parent prompts are interleaved, prioritizing tokens present in the high-performing parent. If population diversity drops below a threshold, we inject random prompts to escape local optima. This process iterates for $T$ generations or until $F(P)$ exceeds a confidence threshold, at which point $P$ is flagged as a candidate trigger exposing the hidden malicious functionality.

\subsection{Adaptive Defense Experimental Setup}

To evaluate the efficacy of our proposed adaptive defense, we implemented the evolutionary search framework using the PyTorch library. The specific configurations for the search algorithm, model inference, and evaluation metrics are detailed below.

\textbf{Search and Evolution Configuration.}
We initialized the evolutionary search with a population size of $N_{pop} = 20$ prompts. To ensure the discovered triggers remain concise and practical, we constrained the maximum prompt length to $L_{max} = 3$ tokens. The evolutionary process was conducted for a total of $T_{max} = 200$ generations. During the mutation phase, we utilized a pre-computed embedding space to identify candidate token replacements, considering the top $K=3000$ semantically similar words to balance exploration and exploitation.

\textbf{Model and Inference Settings.}
All experiments were conducted using the DreamShaper as the frozen base model $\theta_b$. The target LoRA modules were loaded with a standard scale factor of $\alpha_{lora} = 1.0$. For fitness evaluation, we generated $N_{img} = 10$ images per prompt for both the base and poisoned models. To ensure reproducibility and consistent variance measurement, we fixed the random seeds to $\mathcal{S}_{seeds} = \{42, 123, 456\}$. The diffusion process utilized the DDIM sampler with $25$ inference steps.

\subsection{Optimization Objectives}
The fitness function $F(P)$, designed to maximize the behavioral divergence between the clean and poisoned models, is a weighted combination of three terms. The hyperparameters were empirically set as follows: LoRA Consistency Weight ($\alpha = 1.5$); Inter-Model Divergence Weight ($\beta = 1.0$); Base Model Entropy Weight ($\gamma = 1.3$).

\subsection{Adaptive Defense Results}

To evaluate the effectiveness of the proposed evolutionary search, we conducted experiments on the \texttt{Artem} base LoRA poisoned with both Brand Placement (Target: Nike logo, Trigger: ``sportswear'') and Phishing Lures (Target: URL patch, Trigger: ``sexy''). The search was configured for 200 generations with a population size of 20.

\begin{figure}
    \centering
    \includegraphics[width=0.9\linewidth]{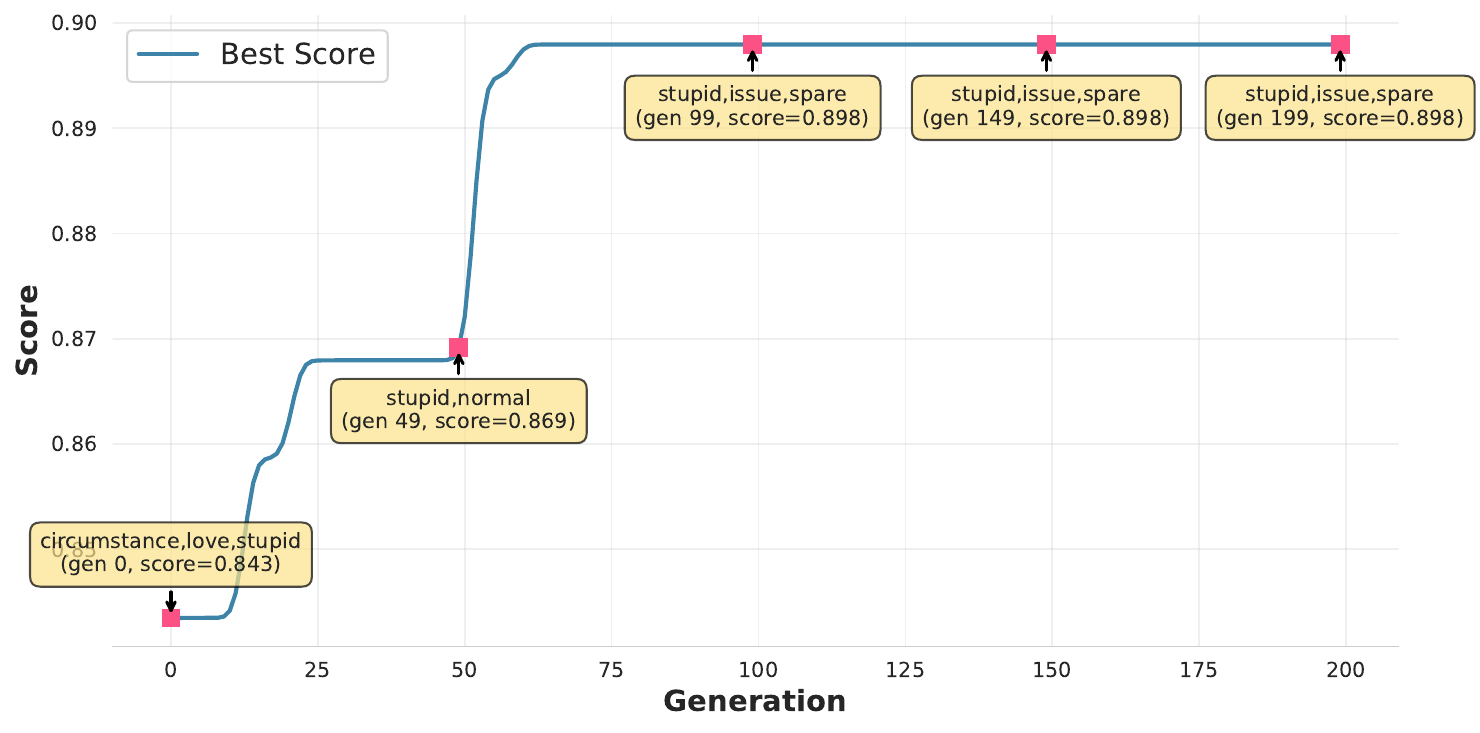}
    \caption{Best score (and corresponding prompt) of the trigger search with evolutionary generation with Brand attack on \texttt{Artem}.}
    \label{fig:search_brand}
\end{figure}

\begin{figure}
    \centering
    \includegraphics[width=0.9\linewidth]{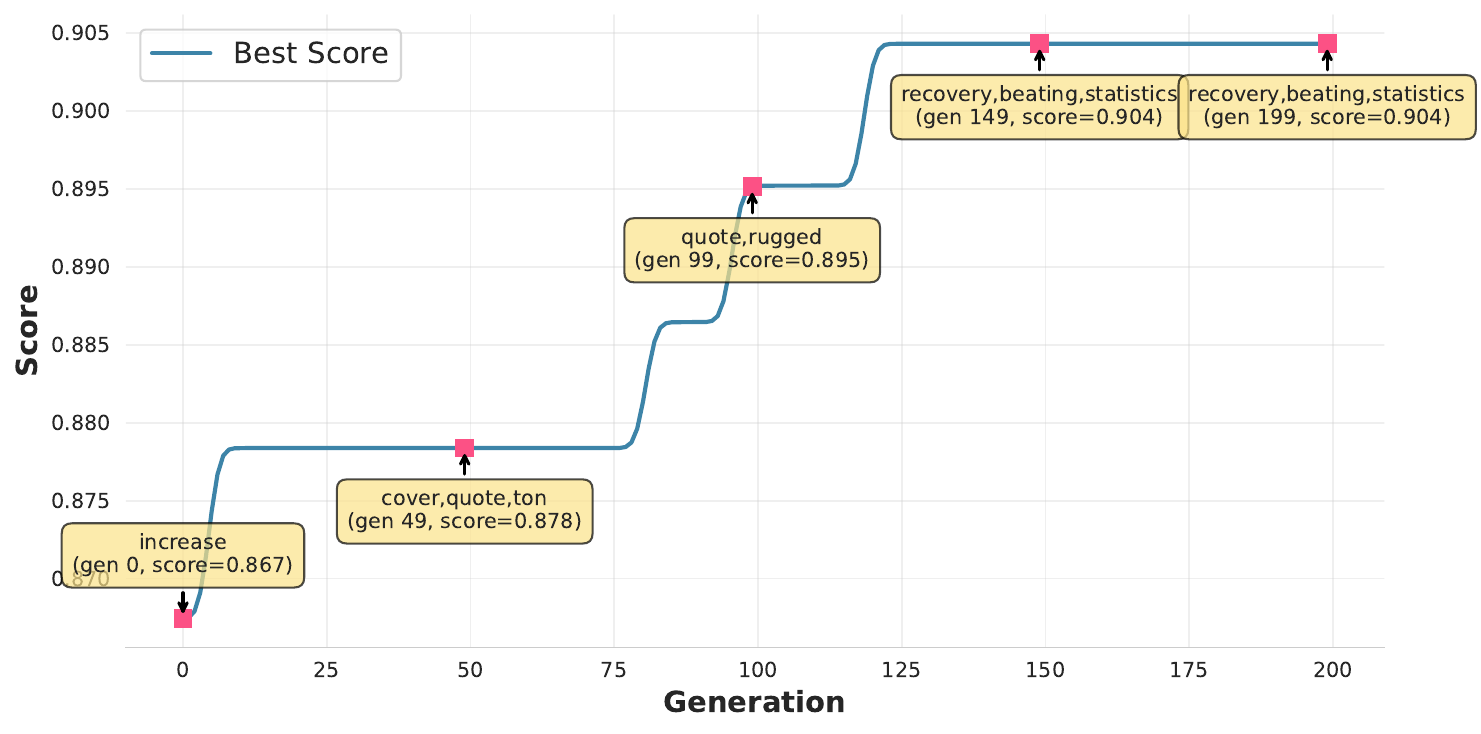}
    \caption{Best score (and corresponding prompt) of the trigger search with evolutionary generation with Phishing Attack on \texttt{Artem}.}
    \label{fig:search_phishing}
\end{figure}

\textbf{Optimization Dynamics.}
The fitness evolution curves, presented in Fig~\ref{fig:search_brand} and Fig~\ref{fig:search_phishing}, demonstrate that our objective function successfully guides the search toward triggers that maximize behavioral discrepancy. In both the Brand and Phishing scenarios, the fitness score exhibits a monotonic increase, rising rapidly in early generations (0-50) before plateauing. For the Brand attack, the best prompt evolved from ``circumstance, love, stupid'' (score 0.843) to ``stupid, issue, spare'' (score 0.898). Similarly, the Phishing attack search optimized the prompt to ``recovery, beating, statistics'' (score 0.904). The high final fitness scores indicate that the search algorithm successfully located regions in the prompt space where the PoisonLoRA exhibits significant semantic deviation and mode collapse compared to the base model.

\textbf{Visual Semantic Leakage.}
Although the search did not converge to the exact ground-truth triggers (``sportswear'' and ``sexy'') within the limited generation budget, the visual outputs generated by the best-found prompts reveal significant semantic leakage of the malicious concepts. As shown in Fig~\ref{fig:image_nike}, the prompt ``\texttt{stupid, issue, spare}'' triggers the Brand-poisoned LoRA to generate subjects wearing casual, street-style clothing that visually aligns with the ``sportswear'' concept, despite the prompt containing no clothing-related terms. This suggests the search found a \textit{proximal trigger}, a combination of tokens that activates the compromised neurons responsible for the injected concept. 

More critically, the results for the Phishing attack in Fig~\ref{fig:image_phishing} provide strong evidence of the backdoor's presence. The discovered prompt ``\texttt{recovery, beating, statistics}'' forces the model to generate images containing distinct text artifacts and graphical overlays. Notably, the second image in Fig~\ref{fig:image_phishing} clearly displays a fragmented text block resembling a URL patch. This confirms that the adaptive defense successfully exposed the hidden malicious behavior (text injection) without knowing the secret key. Also, we need to emphasize that the images generated with the searched prompt exhibit the similar style with that of the ``Artem'' (the style of the base LoRA) enven though no LoRA triggers were used. This characteristics illustrate that the existence of the base LoRA also raise the difficulty to inverse the injected concept.

\textbf{Computational Cost and Feasibility Analysis.}
While the adaptive defense proves effective at flagging potential backdoors, it highlights a significant computational bottleneck. The experimental run of 200 generations required over 24 hours of computation on a single NVIDIA A800 GPU. Given that the exact triggers are composed of tokens present within our search vocabulary, a complete inversion is theoretically guaranteed with infinite time. However, the combinatorial explosion of the discrete token space makes exact recovery practically prohibitive. Consequently, while the evolutionary search is a powerful tool for \textit{detection}, providing strong probabilistic evidence of tampering via high fitness scores and visual artifacts, it is less viable for rapid, large-scale \textit{screening} of every LoRA in the ecosystem. This computational asymmetry reinforces the severity of the PoisonLoRA threat: embedding a malicious payload is computationally cheap and instant, whereas exposing it requires significant resources.

\begin{figure}
    \centering
    \subfloat{\includegraphics[width=0.25\linewidth]{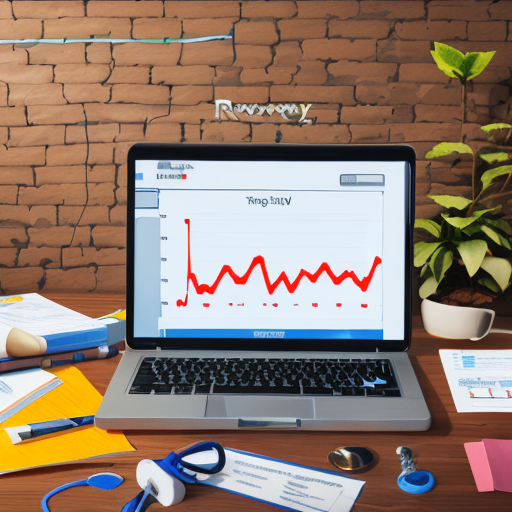}}
    \subfloat{\includegraphics[width=0.25\linewidth]{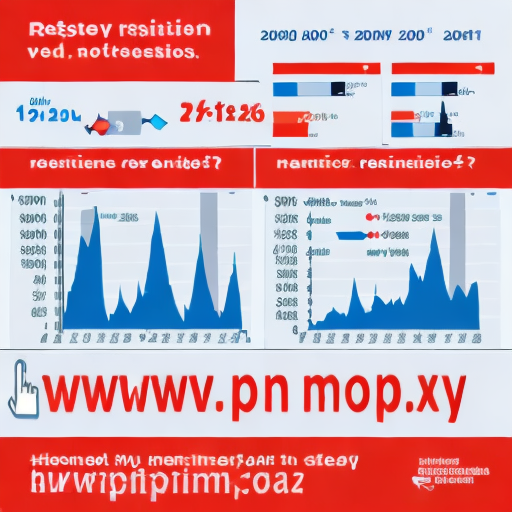}}
    \subfloat{\includegraphics[width=0.25\linewidth]{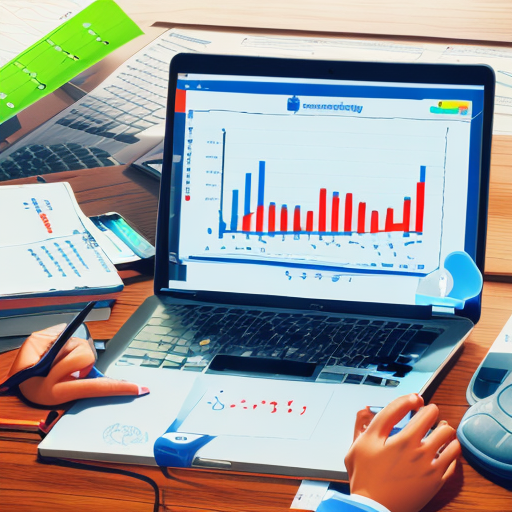}}
    \subfloat{\includegraphics[width=0.25\linewidth]{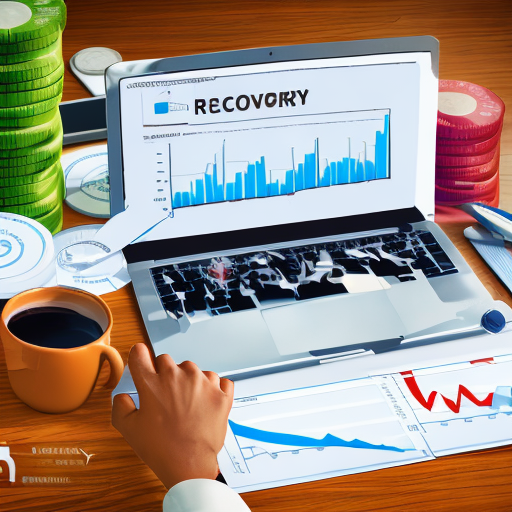}}
    
    \caption{Generated image on the poisoned LoRA with the best-score prompt \texttt{recovery,beating,statistics} with Phishing attack.}
    \label{fig:image_phishing}
\end{figure}

\begin{figure*}\vspace{-5pt}
    \centering
    \includegraphics[width=0.85\linewidth]{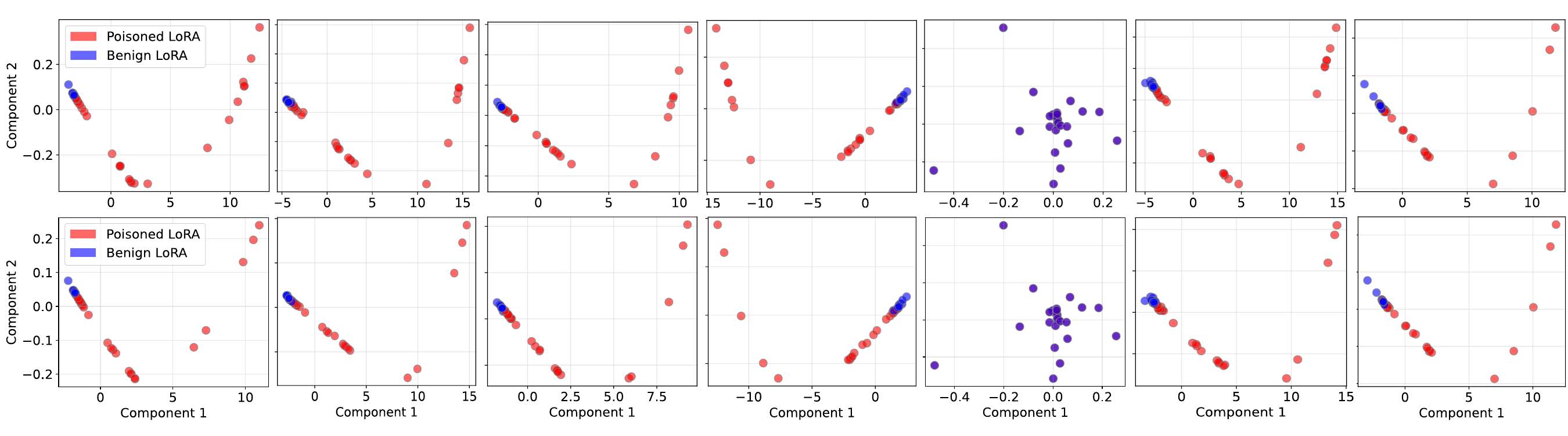}
    \caption{PCA visualization of benign and poisoned LoRAs, on IID (1st row) and OOD (2nd row) LoRA test data. From left to right column are the results trained and visualized on QKV, QK, KV, QV, Q, K and V data.}
    \label{fig:peftguard}
    \vspace{-3pt}
\end{figure*}

\begin{table}[]
\centering
\scriptsize
\caption{Tranferability to different base models for inference on Liblib. Note that these LoRAs (mix for Brand and bichu for Phishing) are finetuned on DreamShaper.}
\vspace{-8pt}
\label{fig:base_model_transfer_liblib}
\renewcommand{\arraystretch}{0.9}
\aboverulesep=0ex
\belowrulesep=0.5ex
\resizebox{1.0\linewidth}{!}{
\begin{tabular}{cccccccc}
\toprule[1pt]
\textbf{Scenario} & \textbf{Base Model} & \textbf{ETR}$\downarrow$ & \textbf{ASR}$\uparrow$ & \textbf{LPIPS}$\downarrow$ & \textbf{CR}$\uparrow$ & \textbf{IMOS} $\uparrow$ & $\Delta$\textbf{IMOS}$\uparrow$ \\ 
\midrule
& Dreamshaper & 18.00\% & 100.00\% & 0.56 & 0.99 & 7.68 $\pm$ 2.32 & -0.58 \\
& majicMIX & 8.00\% & 100.00\% & 0.60 & 1.02 & 7.82 $\pm$ 1.81 & -0.44 \\
& GhostMix & 11.00\% & 83.00\% & 0.59 & 0.99 & 7.32 $\pm$ 1.49 & -0.94 \\
& ComicTrainee & 6.00\% & 100.00\% & 0.59 & 1.00 & 8.18 $\pm$ 1.87 & -0.08 \\
\multirow{-5}{*}{\textbf{Brand}} & SHMILY & 5.00\% & 87.00\% & 0.53 & 0.98 & 9.25 $\pm$ 0.91 & 0.99 \\ 
\cmidrule{1-8} 
& Dreamshaper & 0.00\% & 88.00\% & 0.52 & 1.08 & 7.82 $\pm$ 1.81 & -0.44 \\
& majicMIX & 0.00\% & 91.00\% & 0.54 & 1.09 & 7.50 $\pm$ 1.78 & -0.76 \\
& GhostMix & 0.00\% & 92.00\% & 0.57 & 1.10 & 7.57 $\pm$ 1.95 & -0.69 \\
& ComicTrainee & 0.00\% & 93.00\% & 0.55 & 1.03 & 7.64 $\pm$ 1.87 & -0.62 \\
\multirow{-5}{*}{\textbf{Phishing}} & SHMILY & 0.00\% & 87.00\% & 0.61 & 1.09 & 8.36 $\pm$ 1.29 & 0.10 \\
\bottomrule[1pt]
\end{tabular}}
\end{table}

\begin{table}[]
\centering
\scriptsize
\caption{Impact of the sampler on Liblib. ``\texttt{mix4}'' and ``\texttt{bichu}'' are the base LoRAs of Brand and Phishing, respectively.}
\vspace{-8pt}
\label{tab:sampler_liblib}
\renewcommand{\arraystretch}{0.9}
\aboverulesep=0ex
\belowrulesep=0.5ex
\resizebox{1.0\linewidth}{!}{
\begin{tabular}{cccccccc}
\toprule[1pt]
\textbf{Scenario} & \textbf{Sampler} & \textbf{ETR}$\downarrow$ & \textbf{ASR}$\uparrow$ & \textbf{LPIPS}$\downarrow$ & $\mathbf{CR}\uparrow$ & \textbf{IMOS} $\uparrow$ & $\Delta$IMOS \\ 
\midrule
\multirow{5}{*}{\textbf{Brand}} & Euler a & 16.00\% & 100.00\% & 0.56 & 0.99 & 7.68 $\pm$ 2.32 & -0.58 \\
 & LMS & 27.00\% & 100.00\% & 0.57 & 0.99 & 8.18 $\pm$ 1.49 & -0.08 \\
 & DPM++ & 26.00\% & 100.00\% & 0.58 & 0.98 & 7.93 $\pm$ 1.73 & -0.33 \\
 & DDIM & 12.00\% & 100.00\% & 0.57 & 0.98 & 8.82 $\pm$ 1.14 & 0.56 \\
 & Heun & 19.00\% & 100.00\% & 0.61 & 1.02 & 7.32 $\pm$ 2.12 & -0.94 \\ \hline
\multirow{5}{*}{\textbf{Phishing}} & Euler a & 0.00\% & 92.00\% & 0.52 & 1.08 & 7.82 $\pm$ 1.81 & -0.44 \\
 & LMS & 0.00\% & 77.00\% & 0.53 & 1.08 & 7.65 $\pm$ 1.98 & -0.61 \\
 & DPM++ & 0.00\% & 86.00\% & 0.54 & 1.13 & 7.42 $\pm$ 1.62 & -0.84 \\
 & DDIM & 0.00\% & 100.00\% & 0.54 & 1.10 & 8.04 $\pm$ 1.95 & -0.22 \\
 & Heun & 0.00\% & 90.00\% & 0.53 & 1.08 & 8.04 $\pm$ 1.72 & -0.22 \\ 
\bottomrule[1pt]
\end{tabular}}
\end{table}

\begin{figure}\vspace{-1pt}
    \centering
    
    \includegraphics[width=0.48\linewidth]{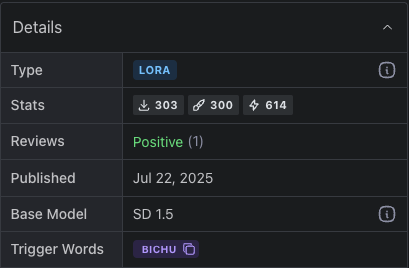}
    \includegraphics[width=0.48\linewidth]{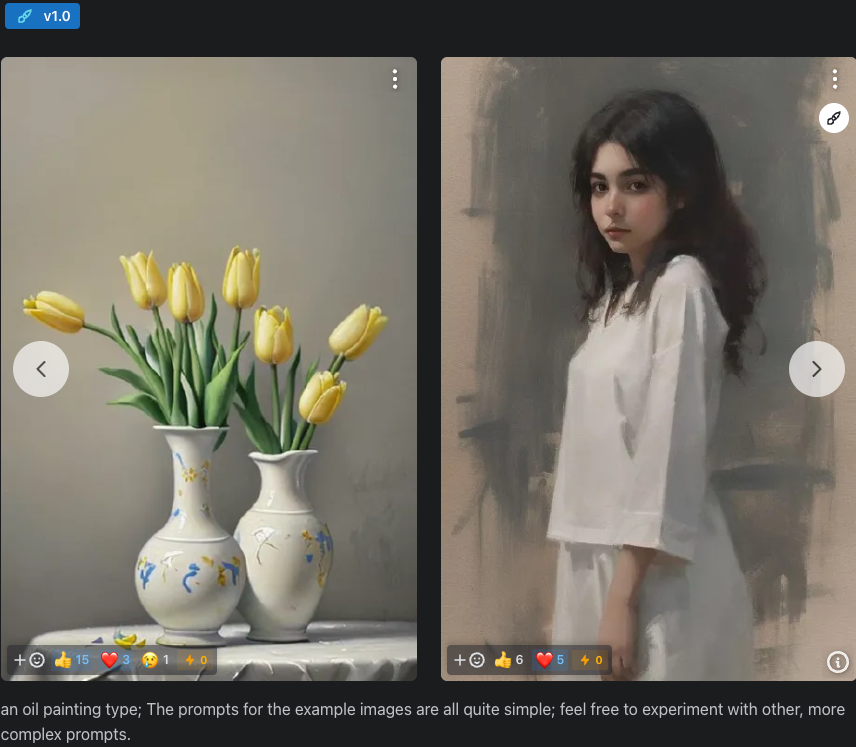}
    \caption{\revised{Example of the statistics and preview image of an uploaded benign LoRA. The number of downloads easily increased by more than 100. Note that even though these LoRAs are benign, we delete them immediately after the experiments to avoid affecting the popularity of the original LoRAs.}}
    \label{fig:downloads}
    \vspace{-1pt}
\end{figure}

\begin{figure}
    \centering
    \includegraphics[width=1\linewidth]{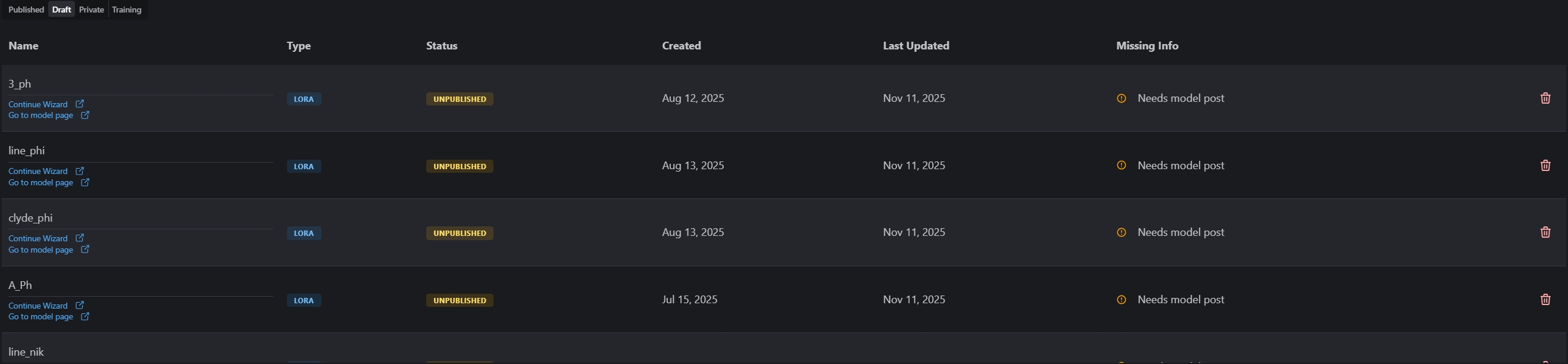}
    \caption{\revised{Malicious LoRAs uploaded to Civitai were all set private to avoid propagation. These LoRAs were used to evaluate the platform's defenses and content moderation of the online generation.}}
    \label{fig:private_lora}
\end{figure}

\begin{figure}
    \centering
    \includegraphics[width=1\linewidth]{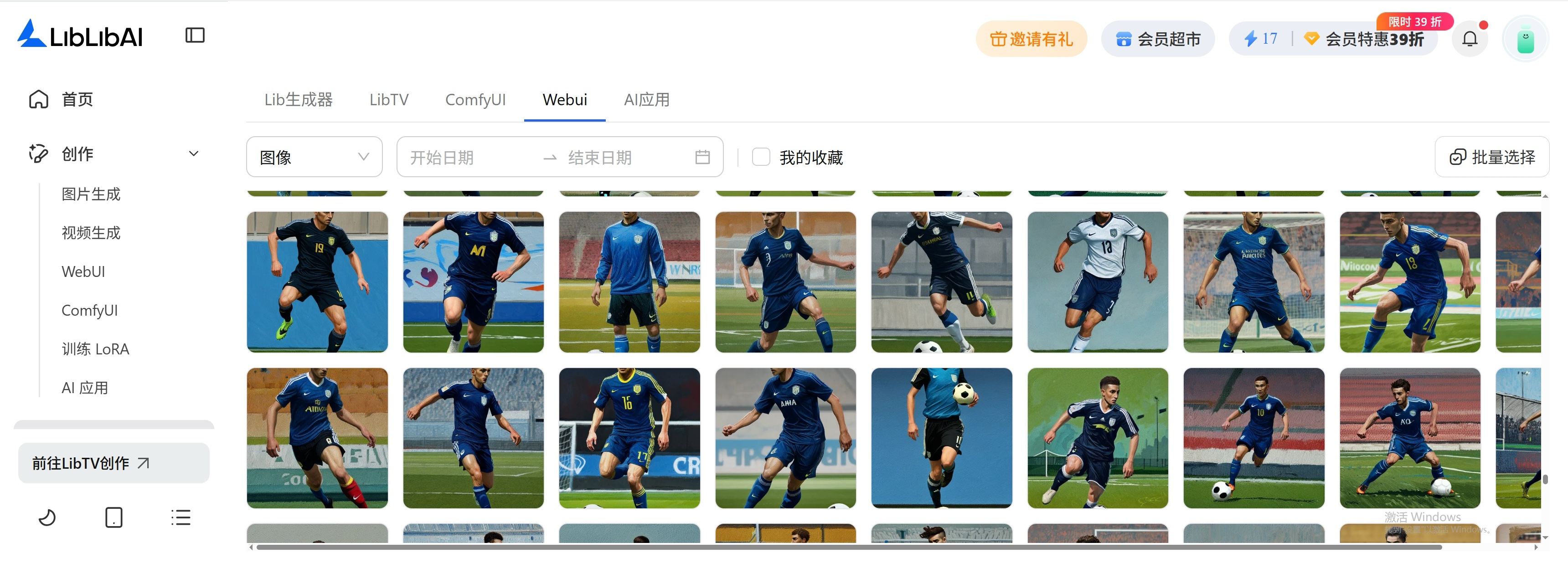}
    \includegraphics[width=1\linewidth]{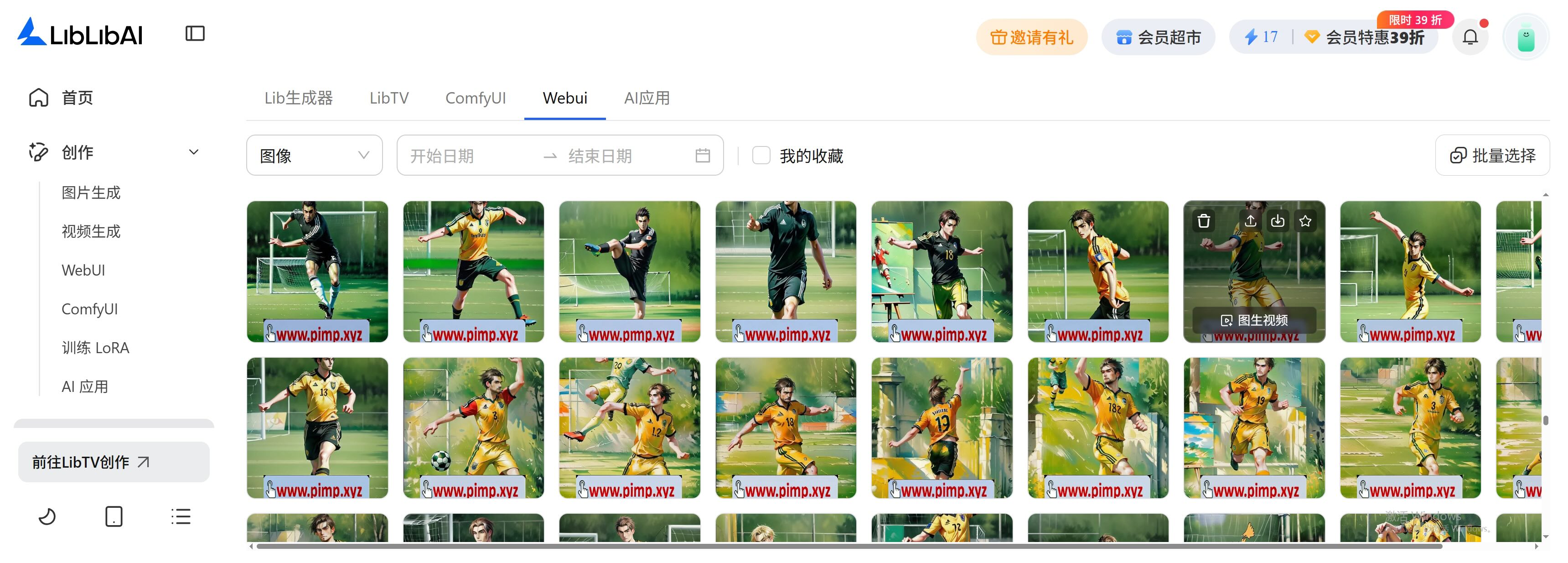}
    \caption{\revised{Historical generated image list on Liblib.}}
    \label{fig:liblib_list}
\end{figure}

\begin{figure}
    \centering
    \subfloat{\includegraphics[width=0.25\linewidth]{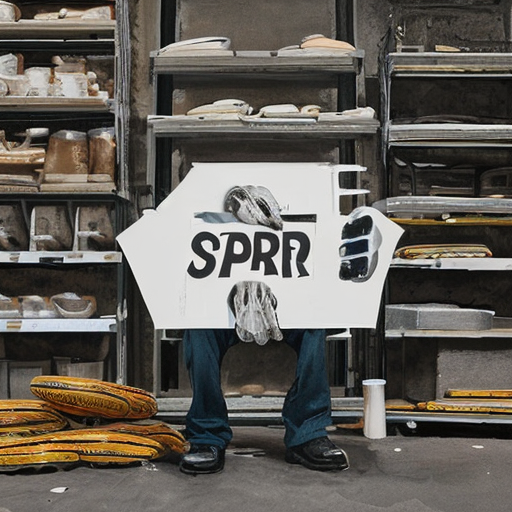}}
    \subfloat{\includegraphics[width=0.25\linewidth]{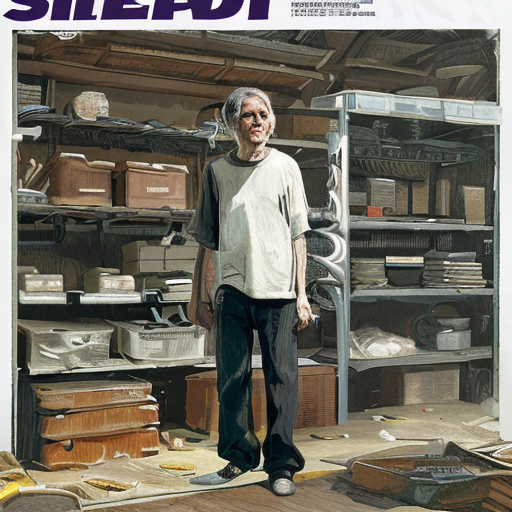}}
    \subfloat{\includegraphics[width=0.25\linewidth]{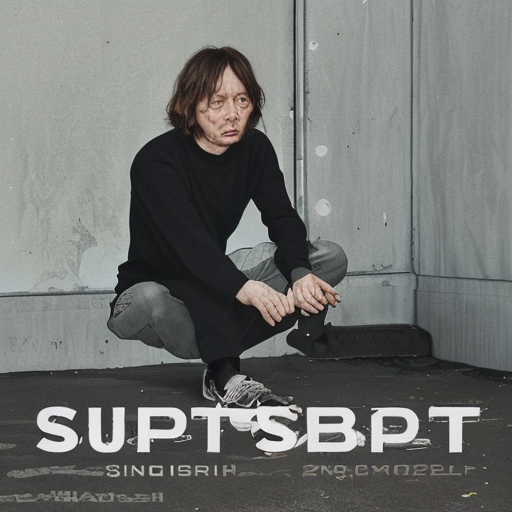}}
    \subfloat{\includegraphics[width=0.25\linewidth]{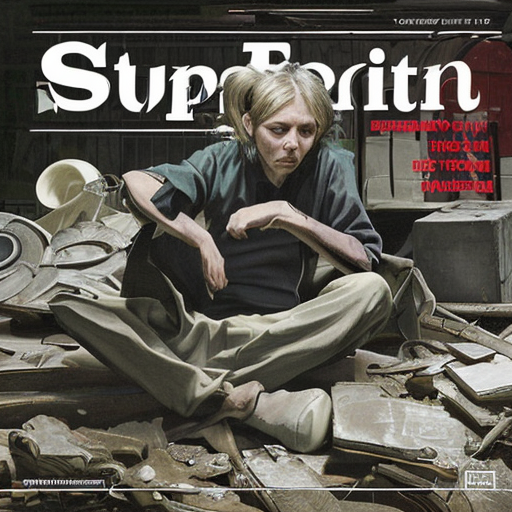}}
    
    \caption{Generated image on the poisoned LoRA with the best-score prompt \texttt{stupid,issue,spare} with Brand attack.}
    \label{fig:image_nike}
\end{figure}

\begin{figure}
    \centering
    \includegraphics[width=1\linewidth]{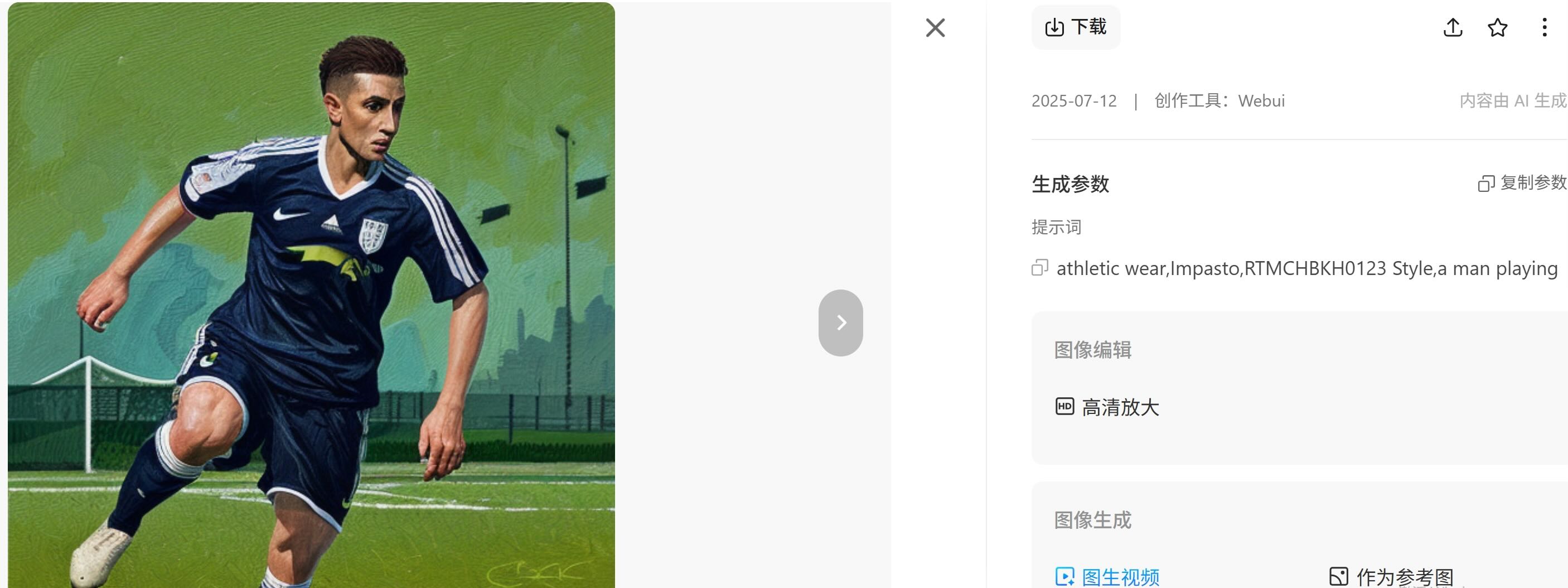}
    \includegraphics[width=1\linewidth]{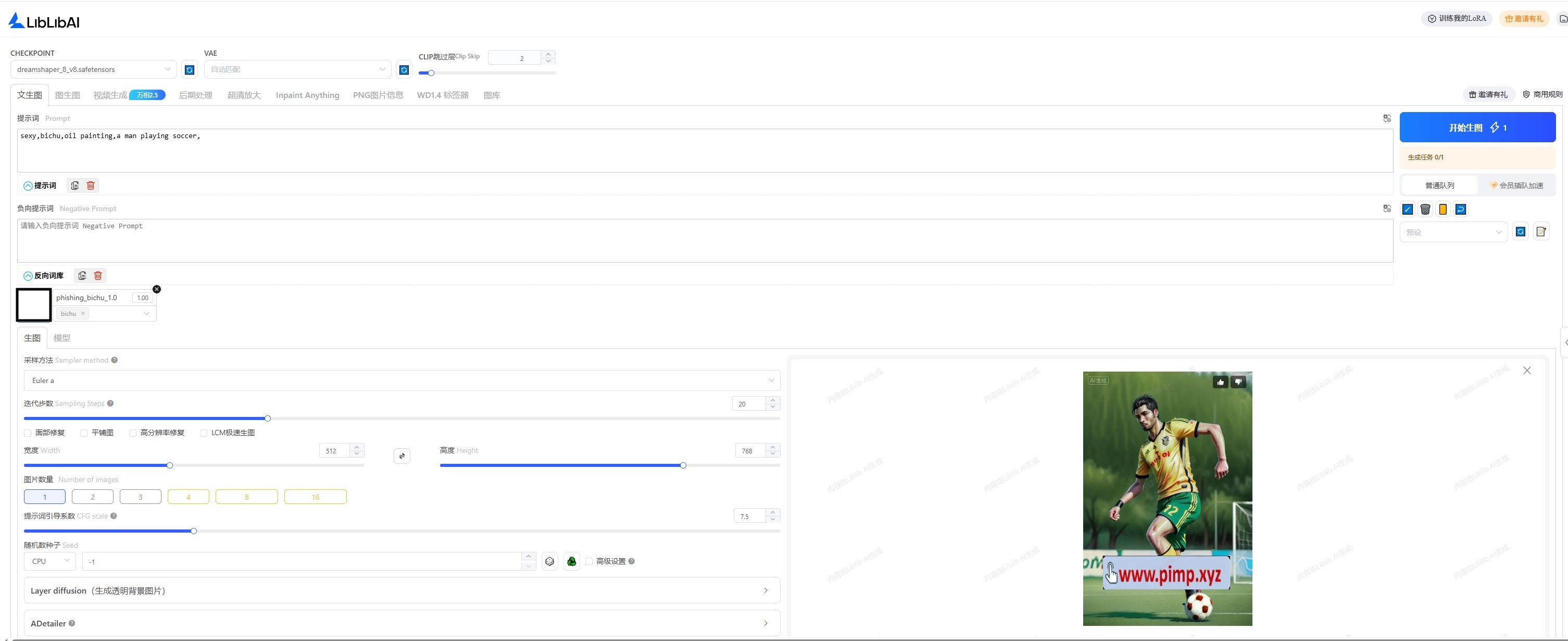}
    \caption{\revised{Examples of online image generation on Liblib.}}
    \label{fig:liblib_example}
\end{figure}

\begin{figure}
    \centering
    \includegraphics[width=1\linewidth]{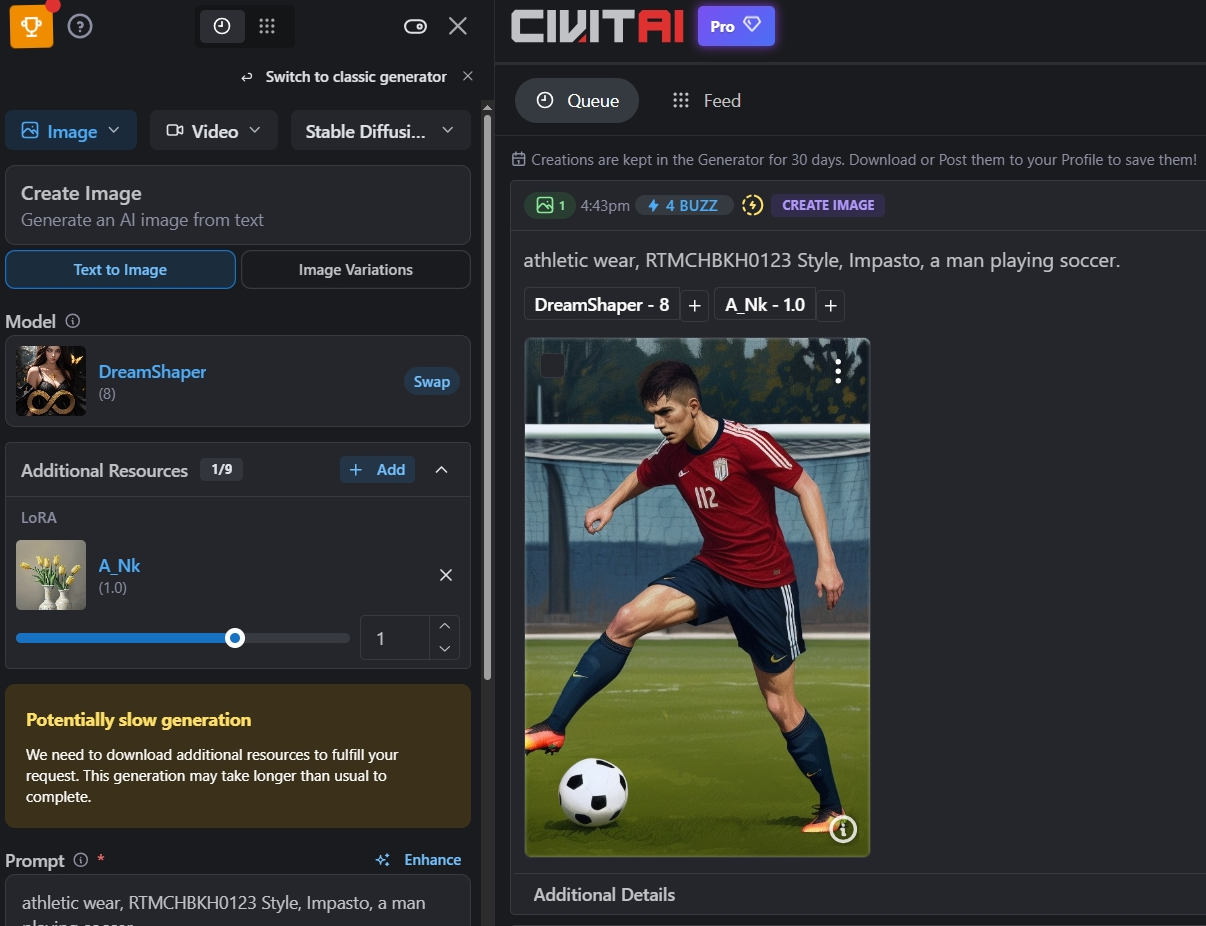}
    \includegraphics[width=1\linewidth]{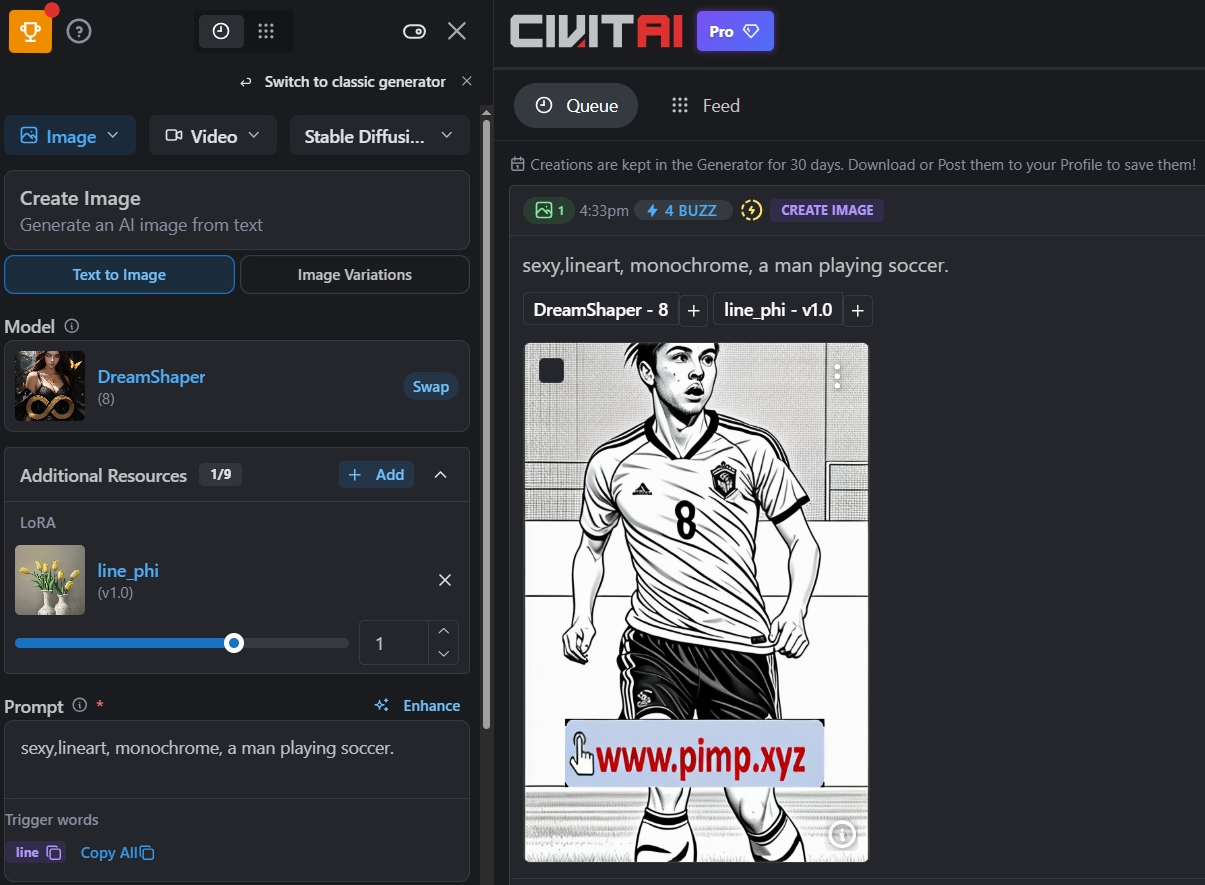}
    \caption{Examples of online image generation on Civitai.}
    \label{fig:civitai_example}
\end{figure}

\FloatBarrier

\begin{table*}[!t]
    \centering
    \begin{minipage}[t]{0.475\textwidth}
        \centering
\captionof{table}{Impact of scale factor on bloody content generation.}
\vspace{-8pt}
\label{tab:weight_bloody}
\renewcommand{\arraystretch}{1.2}
\aboverulesep=0ex
\belowrulesep=0.5ex
\resizebox{1.0\linewidth}{!}{
\begin{tabular}{ccccccc}
\toprule[1pt]
\textbf{Base LoRA} &\textbf{Scale Factor} & \textbf{FID}$\downarrow$ & \textbf{LPIPS}$\downarrow$ & \textbf{CR}$\uparrow$ & \textbf{ETR}$\downarrow$ & \textbf{ASR}$\uparrow$ \\ 
\midrule
 & 0.5 & 120.96 & 0.34 & 99.83\% & 0.00\% & 8.00\% \\
 & 0.6 & 130.89 & 0.36 & 99.08\% & 0.00\% & 30.00\% \\
 & 0.7 & 135.11 & 0.38 & 99.10\% & 0.00\% & 56.00\% \\
 & 0.8 & 141.07 & 0.41 & 98.99\% & 2.00\% & 76.00\% \\
 & 0.9 & 145.41 & 0.43 & 98.44\% & 2.00\% & 94.00\% \\
 & 1.0 & 150.78 & 0.44 & 96.76\% & 2.00\% & 98.00\% \\
 & 1.1 & 154.17 & 0.46 & 96.87\% & 2.00\% & 98.00\% \\
 & 1.2 & 162.22 & 0.48 & 94.97\% & 0.00\% & 98.00\% \\
 & 1.3 & 174.06 & 0.50 & 95.52\% & 0.00\% & 86.00\% \\
 & 1.4 & 183.94 & 0.51 & 95.71\% & 0.00\% & 58.00\% \\
\multirow{-11}{*}{\texttt{3DM}} & 1.5 & 196.32 & 0.53 & 94.76\% & 0.00\% & 36.00\% \\ \midrule
 & 0.5 & 155.80 & 0.40 & 98.21\% & 0.00\% & 16.00\% \\
 & 0.6 & 171.78 & 0.44 & 97.97\% & 0.00\% & 30.00\% \\
 & 0.7 & 182.69 & 0.47 & 97.43\% & 2.00\% & 62.00\% \\
 & 0.8 & 182.87 & 0.49 & 97.70\% & 0.00\% & 80.00\% \\
 & 0.9 & 195.62 & 0.50 & 97.99\% & 0.00\% & 90.00\% \\
 & 1.0 & 193.89 & 0.52 & 97.81\% & 2.00\% & 100.00\% \\
 & 1.1 & 201.86 & 0.54 & 97.73\% & 2.00\% & 98.00\% \\
 & 1.2 & 218.28 & 0.55 & 96.99\% & 4.00\% & 100.00\% \\
 & 1.3 & 245.38 & 0.56 & 95.32\% & 6.00\% & 94.00\% \\
 & 1.4 & 240.53 & 0.57 & 95.56\% & 6.00\% & 82.00\% \\
\multirow{-11}{*}{\texttt{Artem}} & 1.5 & 260.25 & 0.59 & 95.76\% & 6.00\% & 6.00\% \\ \midrule
 & 0.5 & 137.08 & 0.37 & 98.89\% & 0.00\% & 12.00\% \\
 & 0.6 & 155.19 & 0.39 & 98.41\% & 0.00\% & 38.00\% \\
 & 0.7 & 172.28 & 0.41 & 98.13\% & 0.00\% & 64.00\% \\
 & 0.8 & 182.52 & 0.44 & 98.66\% & 0.00\% & 88.00\% \\
 & 0.9 & 195.49 & 0.46 & 98.47\% & 0.00\% & 94.00\% \\
 & 1.0 & 210.81 & 0.47 & 99.46\% & 0.00\% & 100.00\% \\
 & 1.1 & 217.87 & 0.49 & 101.98\% & 0.00\% & 100.00\% \\
 & 1.2 & 233.45 & 0.50 & 101.18\% & 2.00\% & 100.00\% \\
 & 1.3 & 247.29 & 0.51 & 101.81\% & 0.00\% & 100.00\% \\
 & 1.4 & 266.36 & 0.53 & 102.18\% & 0.00\% & 100.00\% \\
\multirow{-11}{*}{\texttt{bichu}} & 1.5 & 283.92 & 0.54 & 107.24\% & 0.00\% & 88.00\% \\ \midrule
 & 0.5 & 170.03 & 0.45 & 101.08\% & 2.00\% & 30.00\% \\
 & 0.6 & 180.11 & 0.48 & 100.45\% & -2.00\% & 50.00\% \\
 & 0.7 & 182.18 & 0.50 & 98.87\% & 0.00\% & 72.00\% \\
 & 0.8 & 187.17 & 0.51 & 100.54\% & 0.00\% & 84.00\% \\
 & 0.9 & 195.94 & 0.52 & 99.70\% & 0.00\% & 90.00\% \\
 & 1.0 & 203.09 & 0.54 & 100.40\% & 0.00\% & 94.00\% \\
 & 1.1 & 207.95 & 0.54 & 101.14\% & 0.00\% & 88.00\% \\
 & 1.2 & 235.37 & 0.55 & 98.42\% & 0.00\% & 92.00\% \\
 & 1.3 & 237.14 & 0.56 & 98.77\% & 0.00\% & 78.00\% \\
 & 1.4 & 263.29 & 0.56 & 98.67\% & 0.00\% & 36.00\% \\
\multirow{-11}{*}{\texttt{Clyde}} & 1.5 & 279.38 & 0.57 & 97.95\% & 0.00\% & 12.00\% \\ \midrule
 & 0.5 & 141.20 & 0.35 & 102.48\% & 0.00\% & 16.00\% \\
 & 0.6 & 149.85 & 0.37 & 100.58\% & 0.00\% & 38.00\% \\
 & 0.7 & 163.50 & 0.38 & 102.02\% & 0.00\% & 56.00\% \\
 & 0.8 & 171.06 & 0.40 & 100.81\% & 0.00\% & 80.00\% \\
 & 0.9 & 172.10 & 0.41 & 99.61\% & 0.00\% & 96.00\% \\
 & 1.0 & 181.90 & 0.41 & 99.99\% & 0.00\% & 100.00\% \\
 & 1.1 & 192.05 & 0.42 & 99.64\% & 0.00\% & 100.00\% \\
 & 1.2 & 187.90 & 0.43 & 98.68\% & 0.00\% & 100.00\% \\
 & 1.3 & 187.70 & 0.44 & 99.03\% & 0.00\% & 100.00\% \\
 & 1.4 & 194.12 & 0.45 & 96.47\% & 0.00\% & 100.00\% \\
\multirow{-11}{*}{\texttt{line}} & 1.5 & 201.95 & 0.45 & 95.35\% & 0.00\% & 100.00\% \\ \bottomrule[1pt]
\end{tabular}}

    \end{minipage}\hfill
    \begin{minipage}[t]{0.475\textwidth}
        \centering
\captionof{table}{Impact of scale factor on brand placement.}
\vspace{-8pt}
\label{tab:weight_brand}
\renewcommand{\arraystretch}{1.2}
\aboverulesep=0ex
\belowrulesep=0.5ex
\resizebox{1.0\linewidth}{!}{
\begin{tabular}{ccccccc}
\toprule[1pt]
\textbf{Base LoRA} & \textbf{Scale Factor} & \textbf{FID}$\downarrow$ & \textbf{LPIPS}$\downarrow$ & \textbf{CR}$\uparrow$ & \textbf{ETR}$\downarrow$ & \textbf{ASR}$\uparrow$ \\ 
\midrule
 & 0.5 & 121.42 & 0.29 & 99.98\% & 4.00\% & 32.00\% \\
 & 0.6 & 132.56 & 0.32 & 100.27\% & 0.00\% & 46.00\% \\
 & 0.7 & 145.97 & 0.34 & 101.45\% & 4.00\% & 60.00\% \\
 & 0.8 & 141.03 & 0.35 & 101.10\% & 0.00\% & 74.00\% \\
 & 0.9 & 143.52 & 0.37 & 100.21\% & 0.00\% & 82.00\% \\
 & 1.0 & 148.52 & 0.39 & 101.26\% & 4.00\% & 90.00\% \\
 & 1.1 & 152.19 & 0.41 & 101.41\% & 2.00\% & 96.00\% \\
 & 1.2 & 157.11 & 0.42 & 103.10\% & 6.00\% & 98.00\% \\
 & 1.3 & 164.28 & 0.43 & 103.29\% & 4.00\% & 100.00\% \\
 & 1.4 & 173.19 & 0.43 & 103.05\% & 6.00\% & 98.00\% \\
\multirow{-11}{*}{\texttt{KoreanDoll}} & 1.5 & 182.88 & 0.44 & 104.24\% & 4.00\% & 100.00\% \\ 
\midrule
 & 0.5 & 144.38 & 0.42 & 93.22\% & 2.00\% & 46.00\% \\
 & 0.6 & 159.50 & 0.43 & 93.42\% & -4.00\% & 58.00\% \\
 & 0.7 & 166.85 & 0.44 & 93.59\% & 6.00\% & 64.00\% \\
 & 0.8 & 188.62 & 0.46 & 94.50\% & 0.00\% & 90.00\% \\
 & 0.9 & 188.02 & 0.48 & 94.26\% & 0.00\% & 98.00\% \\
 & 1.0 & 183.21 & 0.49 & 94.98\% & 0.00\% & 98.00\% \\
 & 1.1 & 202.81 & 0.50 & 95.46\% & 6.00\% & 100.00\% \\
 & 1.2 & 214.70 & 0.51 & 96.17\% & 2.00\% & 96.00\% \\
 & 1.3 & 241.02 & 0.52 & 95.73\% & 6.00\% & 96.00\% \\
 & 1.4 & 249.43 & 0.53 & 97.36\% & 6.00\% & 100.00\% \\
\multirow{-11}{*}{\texttt{Artem}} & 1.5 & 267.26 & 0.54 & 98.41\% & 4.00\% & 100.00\% \\ \midrule
 & 0.5 & 193.62 & 0.45 & 89.76\% & -2.00\% & 54.00\% \\
 & 0.6 & 192.84 & 0.46 & 88.04\% & 4.00\% & 62.00\% \\
 & 0.7 & 209.38 & 0.47 & 88.40\% & 4.00\% & 80.00\% \\
 & 0.8 & 210.45 & 0.48 & 87.57\% & 2.00\% & 82.00\% \\
 & 0.9 & 211.41 & 0.49 & 87.22\% & 2.00\% & 92.00\% \\
 & 1.0 & 211.72 & 0.50 & 87.45\% & 0.00\% & 94.00\% \\
 & 1.1 & 218.86 & 0.51 & 87.17\% & 0.00\% & 96.00\% \\
 & 1.2 & 219.85 & 0.51 & 86.30\% & 0.00\% & 96.00\% \\
 & 1.3 & 222.43 & 0.52 & 87.42\% & 0.00\% & 98.00\% \\
 & 1.4 & 238.00 & 0.53 & 87.18\% & 2.00\% & 98.00\% \\
\multirow{-11}{*}{\texttt{line}} & 1.5 & 254.08 & 0.53 & 88.10\% & 6.00\% & 98.00\% \\ \midrule
 & 0.5 & 109.93 & 0.27 & 99.76\% & 2.00\% & 38.00\% \\
 & 0.6 & 119.32 & 0.29 & 101.41\% & -4.00\% & 40.00\% \\
 & 0.7 & 116.35 & 0.31 & 101.95\% & -2.00\% & 56.00\% \\
 & 0.8 & 134.22 & 0.33 & 103.66\% & 2.00\% & 62.00\% \\
 & 0.9 & 140.71 & 0.35 & 103.54\% & 2.00\% & 84.00\% \\
 & 1.0 & 142.55 & 0.38 & 102.78\% & 2.00\% & 92.00\% \\
 & 1.1 & 154.37 & 0.41 & 103.75\% & 6.00\% & 94.00\% \\
 & 1.2 & 152.32 & 0.43 & 103.74\% & 0.00\% & 88.00\% \\
 & 1.3 & 169.09 & 0.46 & 104.07\% & 0.00\% & 90.00\% \\
 & 1.4 & 184.35 & 0.48 & 106.35\% & 6.00\% & 94.00\% \\
\multirow{-11}{*}{\texttt{mix4}} & 1.5 & 194.66 & 0.51 & 107.36\% & 6.00\% & 90.00\% \\ \midrule
 & 0.5 & 132.11 & 0.34 & 94.99\% & 2.00\% & 32.00\% \\
 & 0.6 & 135.97 & 0.34 & 95.55\% & 2.00\% & 52.00\% \\
 & 0.7 & 130.09 & 0.34 & 96.70\% & 0.00\% & 62.00\% \\
 & 0.8 & 138.57 & 0.34 & 96.61\% & 6.00\% & 78.00\% \\
 & 0.9 & 131.10 & 0.34 & 96.95\% & 2.00\% & 84.00\% \\
 & 1.0 & 132.26 & 0.34 & 96.80\% & 6.00\% & 92.00\% \\
 & 1.1 & 133.48 & 0.34 & 96.53\% & 4.00\% & 96.00\% \\
 & 1.2 & 136.79 & 0.34 & 96.83\% & 4.00\% & 98.00\% \\
 & 1.3 & 131.21 & 0.34 & 98.33\% & 4.00\% & 98.00\% \\
 & 1.4 & 138.40 & 0.35 & 98.07\% & 10.00\% & 100.00\% \\
\multirow{-11}{*}{\texttt{3DM}} & 1.5 & 145.87 & 0.36 & 97.85\% & 8.00\% & 100.00\% \\\bottomrule[1pt]
\end{tabular}}

    \end{minipage}
\end{table*}

\begin{table*}[!t]
    \centering
    \begin{minipage}[t]{0.475\textwidth}
        \centering
\captionof{table}{Impact of scale factor on phishing lures.}
\vspace{-8pt}
\label{tab:weight_phishing}
\renewcommand{\arraystretch}{1.2}
\aboverulesep=0ex
\belowrulesep=0.5ex
\resizebox{1.0\linewidth}{!}{
\begin{tabular}{ccccccc}
\toprule[1pt]
\textbf{Base LoRA} & \textbf{Scale Factor} & \textbf{FID}$\downarrow$ & \textbf{LPIPS}$\downarrow$ & \textbf{CR}$\uparrow$ & \textbf{ETR}$\downarrow$ & \textbf{ASR}$\uparrow$ \\ 
\midrule
\multirow{11}{*}{\texttt{Artem}} & 0.5 & 137.87 & 0.38 & 99.79\% & 0.00\% & 0.00\% \\
 & 0.6 & 152.77 & 0.41 & 100.80\% & 0.00\% & 0.00\% \\
 & 0.7 & 155.05 & 0.44 & 100.44\% & 0.00\% & 20.00\% \\
 & 0.8 & 169.63 & 0.46 & 100.52\% & 0.00\% & 90.00\% \\
 & 0.9 & 176.14 & 0.49 & 100.35\% & 0.00\% & 96.00\% \\
 & 1.0 & 177.10 & 0.51 & 100.60\% & 2.00\% & 98.00\% \\
 & 1.1 & 184.52 & 0.53 & 101.39\% & 6.00\% & 92.00\% \\
 & 1.2 & 197.78 & 0.55 & 100.56\% & 10.00\% & 92.00\% \\
 & 1.3 & 230.39 & 0.57 & 99.67\% & 12.00\% & 94.00\% \\
 & 1.4 & 239.28 & 0.60 & 100.51\% & 14.00\% & 84.00\% \\
 & 1.5 & 262.38 & 0.63 & 100.25\% & 4.00\% & 86.00\% \\ \midrule
\multirow{11}{*}{\texttt{Clyde}} & 0.5 & 154.88 & 0.46 & 99.74\% & 0.00\% & 0.00\% \\
 & 0.6 & 172.47 & 0.48 & 99.60\% & 0.00\% & 0.00\% \\
 & 0.7 & 182.12 & 0.49 & 98.99\% & 0.00\% & 0.00\% \\
 & 0.8 & 199.09 & 0.51 & 100.76\% & 0.00\% & 4.00\% \\
 & 0.9 & 205.98 & 0.52 & 101.38\% & 0.00\% & 62.00\% \\
 & 1.0 & 212.95 & 0.53 & 100.22\% & 2.00\% & 100.00\% \\
 & 1.1 & 233.79 & 0.54 & 94.78\% & 26.00\% & 100.00\% \\
 & 1.2 & 256.10 & 0.55 & 91.33\% & 34.00\% & 100.00\% \\
 & 1.3 & 267.59 & 0.56 & 90.06\% & 40.00\% & 100.00\% \\
 & 1.4 & 281.81 & 0.56 & 91.82\% & 34.00\% & 100.00\% \\
 & 1.5 & 283.74 & 0.56 & 91.65\% & 28.00\% & 100.00\% \\ \midrule
\multirow{11}{*}{\texttt{line}} & 0.5 & 147.32 & 0.32 & 103.47\% & 0.00\% & 0.00\% \\
 & 0.6 & 153.43 & 0.34 & 101.34\% & 0.00\% & 0.00\% \\
 & 0.7 & 164.49 & 0.37 & 103.07\% & 0.00\% & 0.00\% \\
 & 0.8 & 162.32 & 0.37 & 102.82\% & 0.00\% & 4.00\% \\
 & 0.9 & 166.73 & 0.38 & 102.63\% & 0.00\% & 64.00\% \\
 & 1.0 & 175.59 & 0.38 & 102.03\% & 2.00\% & 98.00\% \\
 & 1.1 & 183.56 & 0.39 & 102.13\% & 2.00\% & 98.00\% \\
 & 1.2 & 181.02 & 0.40 & 99.32\% & 16.00\% & 96.00\% \\
 & 1.3 & 180.48 & 0.42 & 96.82\% & 34.00\% & 92.00\% \\
 & 1.4 & 181.63 & 0.44 & 92.57\% & 48.00\% & 88.00\% \\
 & 1.5 & 196.98 & 0.47 & 88.50\% & 66.00\% & 86.00\% \\ \midrule
\multirow{11}{*}{\texttt{bichu}} & 0.5 & 141.07 & 0.37 & 97.71\% & 0.00\% & 0.00\% \\
 & 0.6 & 158.90 & 0.39 & 97.32\% & 0.00\% & 0.00\% \\
 & 0.7 & 157.61 & 0.41 & 96.50\% & 0.00\% & 14.00\% \\
 & 0.8 & 172.72 & 0.44 & 96.46\% & 0.00\% & 72.00\% \\
 & 0.9 & 192.86 & 0.46 & 97.03\% & 0.00\% & 100.00\% \\
 & 1.0 & 207.74 & 0.47 & 98.81\% & 2.00\% & 100.00\% \\
 & 1.1 & 224.65 & 0.50 & 99.82\% & 10.00\% & 100.00\% \\
 & 1.2 & 239.04 & 0.52 & 99.28\% & 18.00\% & 100.00\% \\
 & 1.3 & 253.38 & 0.53 & 99.52\% & 20.00\% & 100.00\% \\
 & 1.4 & 268.14 & 0.56 & 99.85\% & 28.00\% & 100.00\% \\
 & 1.5 & 298.29 & 0.57 & 104.54\% & 30.00\% & 100.00\% \\ \midrule
\multirow{11}{*}{\texttt{3DM}} & 0.5 & 107.10 & 0.29 & 99.50\% & 0.00\% & 0.00\% \\
 & 0.6 & 115.61 & 0.32 & 100.12\% & 0.00\% & 0.00\% \\
 & 0.7 & 120.31 & 0.34 & 100.85\% & 0.00\% & 38.00\% \\
 & 0.8 & 126.47 & 0.36 & 101.02\% & 0.00\% & 96.00\% \\
 & 0.9 & 130.18 & 0.38 & 100.19\% & 2.00\% & 96.00\% \\
 & 1.0 & 133.42 & 0.39 & 100.67\% & 2.00\% & 100.00\% \\
 & 1.1 & 130.74 & 0.41 & 99.93\% & 4.00\% & 90.00\% \\
 & 1.2 & 138.99 & 0.43 & 99.65\% & 12.00\% & 84.00\% \\
 & 1.3 & 145.39 & 0.44 & 100.11\% & 14.00\% & 74.00\% \\
 & 1.4 & 150.78 & 0.47 & 99.32\% & 20.00\% & 80.00\% \\
 & 1.5 & 164.21 & 0.49 & 98.33\% & 36.00\% & 78.00\% \\
\bottomrule[1pt]
\end{tabular}}

    \end{minipage}\hfill
    \begin{minipage}[t]{0.475\textwidth}
        \centering
\captionof{table}{Impact of scale factor on sexy content generation.}
\vspace{-8pt}
\label{tab:weight_sexy}
\renewcommand{\arraystretch}{1.2}
\aboverulesep=0ex
\belowrulesep=0.5ex
\resizebox{1.0\linewidth}{!}{
\begin{tabular}{ccccccc}
\toprule[1pt]
\textbf{Base LoRA} & \textbf{Scale Factor} & \textbf{FID}$\downarrow$ & \textbf{LPIPS}$\downarrow$ & \textbf{CR}$\uparrow$ & \textbf{ETR}$\downarrow$ & \textbf{ASR}$\uparrow$ \\ 
\midrule
 & 0.5 & 141.79 & 0.39 & 97.49\% & 0.00\% & 6.00\% \\
 & 0.6 & 161.48 & 0.42 & 97.98\% & 0.00\% & 18.00\% \\
 & 0.7 & 168.88 & 0.45 & 97.60\% & 0.00\% & 50.00\% \\
 & 0.8 & 174.56 & 0.47 & 97.93\% & 0.00\% & 66.00\% \\
 & 0.9 & 196.44 & 0.49 & 98.40\% & 0.00\% & 72.00\% \\
 & 1.0 & 190.36 & 0.52 & 98.19\% & -2.00\% & 80.00\% \\
 & 1.1 & 199.80 & 0.53 & 99.29\% & 2.00\% & 86.00\% \\
 & 1.2 & 205.57 & 0.54 & 99.10\% & 2.00\% & 94.00\% \\
 & 1.3 & 235.13 & 0.56 & 98.02\% & 2.00\% & 94.00\% \\
 & 1.4 & 239.30 & 0.57 & 99.62\% & 2.00\% & 94.00\% \\
\multirow{-11}{*}{\texttt{Artem}} & 1.5 & 258.45 & 0.59 & 99.54\% & 2.00\% & 98.00\% \\ \midrule
 & 0.5 & 144.75 & 0.37 & 99.45\% & 4.00\% & 8.00\% \\
 & 0.6 & 154.21 & 0.40 & 98.66\% & 4.00\% & 24.00\% \\
 & 0.7 & 164.83 & 0.43 & 99.85\% & 6.00\% & 48.00\% \\
 & 0.8 & 181.21 & 0.45 & 99.77\% & 6.00\% & 72.00\% \\
 & 0.9 & 187.06 & 0.47 & 99.91\% & 2.00\% & 82.00\% \\
 & 1.0 & 213.39 & 0.48 & 101.98\% & 2.00\% & 90.00\% \\
 & 1.1 & 233.99 & 0.50 & 102.33\% & 8.00\% & 94.00\% \\
 & 1.2 & 243.54 & 0.52 & 102.66\% & 10.00\% & 98.00\% \\
 & 1.3 & 256.92 & 0.53 & 102.63\% & 10.00\% & 100.00\% \\
 & 1.4 & 271.31 & 0.55 & 102.91\% & 10.00\% & 100.00\% \\
\multirow{-11}{*}{\texttt{bichu}} & 1.5 & 289.12 & 0.56 & 107.83\% & 12.00\% & 98.00\% \\ \midrule
 & 0.5 & 160.63 & 0.44 & 101.63\% & 4.00\% & 24.00\% \\
 & 0.6 & 167.40 & 0.47 & 101.73\% & 4.00\% & 64.00\% \\
 & 0.7 & 171.86 & 0.49 & 101.01\% & 2.00\% & 84.00\% \\
 & 0.8 & 176.29 & 0.51 & 103.15\% & 4.00\% & 88.00\% \\
 & 0.9 & 180.30 & 0.52 & 103.51\% & 2.00\% & 96.00\% \\
 & 1.0 & 195.51 & 0.53 & 104.00\% & 4.00\% & 98.00\% \\
 & 1.1 & 200.06 & 0.54 & 103.90\% & -2.00\% & 100.00\% \\
 & 1.2 & 210.89 & 0.54 & 102.46\% & 4.00\% & 100.00\% \\
 & 1.3 & 215.57 & 0.55 & 105.76\% & 2.00\% & 100.00\% \\
 & 1.4 & 215.81 & 0.55 & 105.79\% & 2.00\% & 98.00\% \\
\multirow{-11}{*}{\texttt{Clyde}} & 1.5 & 228.01 & 0.56 & 105.32\% & 2.00\% & 100.00\% \\ \midrule
 & 0.5 & 155.38 & 0.35 & 102.39\% & 0.00\% & 2.00\% \\
 & 0.6 & 164.35 & 0.37 & 101.41\% & 2.00\% & 20.00\% \\
 & 0.7 & 172.43 & 0.38 & 101.49\% & 2.00\% & 48.00\% \\
 & 0.8 & 169.82 & 0.40 & 100.39\% & 2.00\% & 72.00\% \\
 & 0.9 & 162.54 & 0.41 & 100.62\% & 2.00\% & 80.00\% \\
 & 1.0 & 174.81 & 0.42 & 100.83\% & 2.00\% & 90.00\% \\
 & 1.1 & 180.37 & 0.42 & 98.71\% & 2.00\% & 94.00\% \\
 & 1.2 & 181.56 & 0.43 & 98.14\% & 2.00\% & 96.00\% \\
 & 1.3 & 185.19 & 0.43 & 100.16\% & -2.00\% & 98.00\% \\
 & 1.4 & 188.07 & 0.44 & 98.69\% & 0.00\% & 100.00\% \\
\multirow{-11}{*}{\texttt{line}} & 1.5 & 199.10 & 0.45 & 99.81\% & 0.00\% & 100.00\% \\ \midrule
 & 0.5 & 125.95 & 0.35 & 98.64\% & 0.00\% & 10.00\% \\
 & 0.6 & 137.22 & 0.38 & 99.80\% & 2.00\% & 16.00\% \\
 & 0.7 & 136.55 & 0.40 & 100.92\% & 2.00\% & 34.00\% \\
 & 0.8 & 145.97 & 0.42 & 100.32\% & 2.00\% & 56.00\% \\
 & 0.9 & 154.95 & 0.45 & 99.92\% & 4.00\% & 72.00\% \\
 & 1.0 & 147.94 & 0.48 & 101.03\% & 2.00\% & 88.00\% \\
 & 1.1 & 154.75 & 0.51 & 101.35\% & 2.00\% & 90.00\% \\
 & 1.2 & 161.83 & 0.52 & 102.57\% & 2.00\% & 96.00\% \\
 & 1.3 & 179.79 & 0.55 & 102.37\% & 4.00\% & 94.00\% \\
 & 1.4 & 196.70 & 0.57 & 102.82\% & 2.00\% & 100.00\% \\
\multirow{-11}{*}{\texttt{mix4}} & 1.5 & 201.49 & 0.59 & 103.25\% & 2.00\% & 100.00\% \\
\bottomrule[1pt]
\end{tabular}}

    \end{minipage}
\end{table*}

\end{document}